\documentclass[12pt, letterpaper]{article}
\title{Optimization-friendly generic mechanisms without money}

\author{Mark Braverman\thanks{Department of Computer Science, Princeton University. Research supported in part by the NSF Alan T. Waterman Award, Grant No. 1933331, a Packard Fellowship in Science and Engineering, and the Simons Collaboration on Algorithms and Geometry. Any opinions, findings, and conclusions or recommendations expressed in this publication are those of the author and do not necessarily reflect the views of the National Science Foundation.  \newline {\bf Acknowledgments.  } I gratefully acknowledge the many comments that helped shape this paper from Itai Ashlagi, Antonio Molina Lovett, Aviad Rubinshtein, Matt Weinberg, and the detailed comments from Sahil Singla.
This paper was influenced by discussions with Georgy Noarov, Sahil Singla, Matt Weinberg, Leeat Yariv, and Yufei  Zheng.}}
\date{\vspace{-5ex}}
\usepackage{fullpage}
\usepackage{amsthm}
\usepackage{amsmath,amssymb,amsfonts,nicefrac,algorithm,algpseudocode}
\usepackage{xspace}
\usepackage{color}
\usepackage{url}
\usepackage{hyperref}
\usepackage{bm}
\usepackage{bbm}
\usepackage{times}

\usepackage{enumitem}

\newtheorem{thm}{Theorem}

\newtheorem{lemma}[thm]{Lemma}

\newtheorem{claim}[thm]{Claim}

\newtheorem{definition}[thm]{Definition}

\newtheorem{problem}{Problem}

\newcommand\E{\mathop{\mathbb{E}}}

\newcommand{\NP}{{\sf NP}}
\newcommand{\ti}{\tilde}

\newcommand{\cX}{\mathcal{X}}
\newcommand{\RR}{\mathbb{R}}

\newcommand{\ra}{\rightarrow}
\newcommand{\la}{\lambda}
\newcommand{\La}{\Lambda}
\newcommand{\al}{\alpha}
\newcommand{\be}{\beta}
\newcommand{\bla}{\bar{\lambda}}

\newcommand{\ignore}[1]{{}}

\newcommand{\cH}{\mathcal{H}}
\newcommand{\cA}{\mathcal{A}}
\newcommand{\cO}{\mathcal{O}}
\newcommand{\cP}{\mathcal{P}}
\newcommand{\cV}{\mathcal{V}}
\newcommand{\cM}{\mathcal{M}}
\newcommand{\lamax}{\bar{\la}}

\newcommand{\barX}{\bar{\mathbf X}}

\newcommand{\rom}[1]{\uppercase\expandafter{\romannumeral #1\relax}}

\newcommand{\ba}{{\mathbf{a}}}
\newcommand{\bA}{{\mathbf{A}}}
\newcommand{\ve}{\varepsilon}
\newcommand{\si}{\sigma}
\newcommand{\de}{\delta}
\newcommand{\De}{\Delta}
\newcommand{\ga}{\gamma}

\makeatletter
\newtheorem*{rep@theorem}{\rep@title}
\newcommand{\newreptheorem}[2]{%
\newenvironment{rep#1}[1]{%
\def\rep@title{\bf #2 \ref*{##1} \text{(Restated)} }%
\begin{rep@theorem} }%
{\end{rep@theorem} } }

\newtheorem*{rep@claim}{\rep@title}
\newcommand{\newrepclaim}[2]{%
\newenvironment{rep#1}[1]{%
\def\rep@title{\bf #2 \ref*{##1} \text{(Restated)} }%
\begin{rep@claim} }%
{\end{rep@claim} } }

\newcommand{\apex}{{\sf APEX }}

\newtheorem*{rep@lemma}{\rep@title}
\newcommand{\newreplemma}[2]{%
\newenvironment{rep#1}[1]{%
\def\rep@title{\bf #2 \ref*{##1} \text{(Restated)} }%
\begin{rep@lemma} }%
{\end{rep@lemma} } }

\makeatother
\newreptheorem{theorem}{Theorem}
\newrepclaim{claim}{Claim}
\newreplemma{lemma}{Lemma}

\begin{document}
\maketitle
\begin{abstract}
The goal of this paper is to develop a generic framework for converting modern optimization algorithms into mechanisms where inputs come from self-interested agents. 

We focus on aggregating preferences from $n$ players in a context without money. Special cases of this setting include voting, allocation of items by lottery, and matching. 
Our key technical contribution is a new meta-algorithm we call \apex (Adaptive Pricing Equalizing Externalities). The framework is sufficiently general to be combined with any optimization algorithm that is based on local search. We outline an agenda for studying the algorithm's properties and its applications. 

As a special case of applying the framework to the problem of one-sided assignment with lotteries, we obtain a strengthening of the 1979 result by  Hylland and Zeckhauser on allocation via a competitive equilibrium from equal incomes (CEEI). The [HZ79] result posits that there is a (fractional) allocation and a set of item prices such that the allocation is a competitive equilibrium given prices. We further show that there is always a reweighing of the players' utility values such that running unit-demand VCG with reweighed utilities leads to a HZ-equilibrium prices. Interestingly, not all HZ competitive equilibria come from VCG prices. As part of our proof, we re-prove the [HZ79] result using only Brouwer's fixed point theorem  (and not the more general Kakutani's theorem). This may be of independent interest. 
\end{abstract}

\newpage

\subsection*{Overview and summary of results}

\paragraph{Motivation.} Our main goal is to develop a generic reduction for converting algorithms based on iterated local optimization into mechanisms. We focus on mechanisms without money (which are generally more difficult to design). We would like our reduction to work for heuristics that have good empirical performance even in lieu of formal guarantees. Therefore, our reduction aims to change the algorithm as little as possible, while attaining good incentive properties. 

Specifically, we start with $n$ players who have preference functions $f_i$ over an outcome space $\cX$. There is an optimization heuristic $\cH$ for maximizing functions over $\cX$. The heuristic $\cH$ is {\em local} --- giving a recipe for constructing a sequence $x_0,x_1,\ldots\in \cX$ that (hopefully) converges to a high-value outcome. Our goal is to use $\cH$ to produce a mechanism that (1) matches the performance of $\cH$ as much as possible; (2) leads to a correlated equilibrium where  for each player $i$,  reporting $f_i$ truthfully is an approximately dominant strategy. 

\paragraph{Main ingredients.} We connect three main ingredients: (1) online learning and its connection to correlated equilibria in games; (2) the VCG mechanism --- a mechanism {\em with} money, where truthful reporting by participants is a dominant strategy; and (3) bandits with knapsacks (BwK) --- a special type of online learning where players 
obtain a reward and experience a capacity cost every time they pull an arm, and where the aim is to maximize total reward subject to a capacity budget. We will elaborate on these ingredients in Sections~\ref{sec:11}--\ref{sec:14}. 

\paragraph{The \apex algorithm and framework.} Our main \apex (adaptive pricing equalizing externalities) framework is given by Algorithm~\ref{alg:1} in Section~\ref{sec:apex}. It is most closely related to CEEI (competitive equilibria from equal incomes) in the mechanism design literature, with a major distinction being that the equilibrium gets discovered  together with the participants via an iterated optimization procedure using heuristic $\cH$ by the principal.

The principal receives utility functions $f_1,\ldots,f_n$ from the players\footnote{The principal may not need to actually collect the $f_i$'s. In fact, the players may not even need to know their own $f_i$'s. The algorithm can be implemented purely using gradient queries, where players are asked to give their local preference $\nabla f_i(x)$ at a given point $x$.}, and a sequence of coefficients $\la_{i,t}\ge 0$ representing the weight the $i$-th player's preferences should be given at round $t$. Each player receives a fixed token endowment $B$ at the beginning of the execution.

 At round $t$, the principal uses heuristic $\cH$ to (locally) optimize the objective $\sum_i \la_{i,t} f_i(x)$ --- possibly with a regularizer added to it. The principal then uses $\cH$ to perform local optimization in order to calculate VCG prices (in tokens) that the players will get charged. The players, from their end, will run bandits with knapsacks algorithms to produce $\la_{i,t}$'s in order to maximize their utility subject to the token budget of $B$. 
 
 The output of the mechanism is the long-term trajectory of this iterated game. Assuming players have negligible regrets, all players are either maximally happy or exhaust their token budgets --- which mean that their averaged VCG payments (and thus their externalities) are equalized.

\paragraph{General results.} The \apex algorithm {\em converges} if after some number of iterations: (1) a good-value solution to the optimization problem is reached; (2) players have negligible regret with respect to their actions in the bandits-with-knapsacks game. 
The framework is very general, and it is unlikely that a full-generality convergence result can be proved (especially since we do not wish to make assumptions on the heuristic $\cH$). 
However, {\em assuming} the algorithm converges, we can show that it leads to an approximate correlated equilibrium where truthful reporting of $f_i$ is dominant:

\medskip
\noindent {\bf Lemma~\ref{lem:RegToEq1}.~~[restated]}
{\em 	Suppose that $f_i(x)\in[0,1]$ for all $x\in\cX$, and that during the execution of Algorithm~\ref{alg:1}  with budget $B_i$ and a truthfully reported $f_i$,  Player $i$ has strong regret $\le\ve$. Suppose further that heuristic $\cH$ is locally correct.  Then reporting $f_i$ truthfully is an $(2\ve)$-dominant strategy for the menu of options available to Player~$i$ that is induced by the mechanism. }
	\medskip
	
The conditions of Lemma~\ref{lem:RegToEq1} can be verified given an execution of the mechanism. Even if we can't be sure that the mechanism will converge, we can be assured of good incentive properties given a convergent execution. Therefore, the framework allows us to convert heuristic algorithms into heuristic mechanisms. 

\paragraph{Application: new results for the one-sided matching problem.}
In Section~\ref{sec:app}, we discuss applications of the framework to three classical areas of  mechanisms without money: voting, one-sided assignment, and two-sided matching. Voting with cardinal preferences (Section~\ref{sec:voting}) is subject of an ongoing work and is mostly beyond the scope of this paper. Efficient two-sided matching (Section~\ref{sec:twosided}) --- the Gale-Shapley setting but with cardinal preferences ---  is perhaps the most interesting immediate application of our framework. Plugging the two-sided matching setup into the \apex framework gives interesting initial results, but a key definition of ``externality" in this setting appears to be non-canonical. We discuss this issue in detail in Section~\ref{sec:twosided}. 

One area where we are able to immediately use the \apex framework to obtain new results is one-sided assignment. A classical result of Hylland and Zeckhauser \cite{hylland1979efficient} states that in a setting without money $n$ items can be allocated to $n$ unit-demand players via lotteries using an equilibrium from equal incomes (CEEI). Given utilities $u_{ij}\ge 0$ --- the utility of player $i$  for item $j$ --- there are prices $P_j\ge 0$ assigned to items, and a bi-stochastic $n\times n$ allocation matrix $X$, such that $X_i$ is the best distribution on items Player~$i$ can afford with a unit budget and prices $\{P_j\}$. 

By plugging the one-sided allocation into the \apex framework, where one iteration is unit-demand VCG with utilities $\{\la_{i,t}\cdot u_{ij}\}$, we obtain a strengthening of the HZ result. We show that there is always a scaling of utilities such that resulting VCG prices support a HZ equilibrium.

\medskip \noindent {{\bf Theorem~\ref{thm:VCG-HZ}.} {\em [restated]}~~~}
{\em	Let $U=\{u_{ij}\}_{i,j=1..n}$ be a matrix utilities with $u_{ij}\ge 0$. Then there exist numbers $\la_i\ge 0$, prices $C=\{C_j\}$ and an allocation $X=\{x_{ij}\}$ with the following properties. 
	\begin{enumerate}
		\item 
		$X$ is a valid allocation: $\forall j: \sum_i x_{ij}=1$ and
		$\forall i:\sum_j x_{ij}=1$; 
		\item
		$C_j$ are the VCG prices for utilities given by $u'_{ij} = \la_i u_{ij}$;
		\item
		$X$ is a combination of optimal allocations under $u'$: for every $\pi:[n]\ra[n]$ with $\forall i~x_{i\pi(i)}>0$ we have 
		$$
		\sum_i u'_{i \pi(i)} = \max_{\si} 	\sum_i u'_{i \si(i)}. 
		$$
		\item 
		The players can purchase their allocations with budget not exceeding $1$. For each player $i$, 
		$$
		\sum_j C_j x_{ij} \le 1. 
		$$
		\item 
		Prices $C_j$ and allocation $X$ form a HZ equilibrium. That is, for every player $i$
		$$
	\sum_{j} u'_{ij} x_{ij} = \max_{\displaystyle{y:} \begin{array}{c}\sum_j C_j y_j \le 1 \\ \sum_j y_j =1\end{array}} \sum_j u'_{ij} y_j. 
	$$
	\end{enumerate}}
\medskip

Theorem~\ref{thm:VCG-HZ} was discovered via the \apex framework, but we prove it directly using Brouwer's fixed-point theorem. Our proof is arguably simpler than the original proof of \cite{hylland1979efficient}, although it relies on the fact that properties of unit-demand VCG auctions are very well-understood at this point. 

Interestingly, we show that not all HZ CEEI prices are VCG prices, and thus equilibria supported by VCG prices on scaled utilities form a proper subset of all HZ equilibria. 

Even though Theorem~\ref{thm:VCG-HZ} is proved directly without using the \apex algorithm, we do show that any low-regret execution execution of the \apex algorithm on a {\em regularized} objective 
$$
F_t(x) := \sum_{i,j} \la_{i,t} \cdot u_{ij} x_{ij} + F_0(x)
$$ 
will lead to an approximate HZ competitive equilibrium.

\medskip \noindent {{\bf Theorem~\ref{thm:reg}.} {\em [informal, restated]}~~~}
 {\em	In the unit-demand allocation setting without money with $n$ players and $n$ items, let $\{u_{ij}\}\ge 0$ be utilities. 
	
	For each $\de>0$, there is an $\ve>0$ and a concave regularizer $F_0(x)$ such that  an execution of the \apex algorithm with  regret $<\ve \cdot T$, leads to an allocation $X$ and VCG prices $C$, such that $X$ is supported by a $\de$-approximate competitive equilibrium from equal income with prices $C$.  }
\medskip

\section{Introduction}

Algorithms play an increasingly important role in coordinating 
a broad range of human activity. Algorithm design addresses the problem 
of attaining a desired outcome on a given input. For example, finding a 
good allocation of tasks to machines, finding and maintaining a communication 
route between a client and a host, or optimizing an online advertisement 
campaign for maximum impact. More open-ended (in terms of objective function) important algorithmic tasks include internet search and matching consumers to goods and services. Another very important special class of algorithms has to do with building models for predicting the future --- for example in the context of planning and control. Any process affecting the well-being of participants it does not directly control invites manipulation by those participants for their own benefit, and algorithm-driven processes are not exempt from this rule. Examples range from strategic voting to the search engine optimization industry. This motivates the field of {\em algorithmic mechanism design}, whose goal is to design algorithms that work ``well" even when 
the inputs come from self-interested players. 

\paragraph{Algorithmic mechanism design and the price of anarchy.} In mathematical terms, an algorithm induces a mechanism which can be analyzed 
using game-theoretic tools. Typically, the design goal is to attain a good performance in some kind of game-theoretic {\em equilibrium}: a situation where the output of the  algorithm is good according to the prescribed performance metric, while participants cannot change their behavior (such as their input to the algorithm) to drastically improve their own well-being. 

Almost any optimization problem can be cast into one (or more) mechanism design problem based on which participants are allowed to behave strategically, and the space of allowed strategic behaviors. Needless to say, the mechanism design problem is significantly more difficult both mathematically (having to deal with game-theoretic equilibria instead of simple objective function values), and in terms of the performance one can guarantee. In the algorithmic game theory literature the gap between the performance of the best optimization algorithm and the best (equilibrium) performance of a mechanism for the same problem is called the {\em price of anarchy}, and it can be significant in many cases. 

\paragraph{Online algorithms.} An algorithmic setting which plays an important role both in practice and in the theory of machine learning is that of {\em online algorithms}. In the online setting, at time $t$ the algorithm receives an input $X_t$, and needs to produce an action $A_t(X_{1..t},Y_{1..t-1})$. It then learns the state of nature $Y_t$ at time $t$, and experiences loss $L(X_t,A_t,Y_t, R_t)$ (where $R_t$ is random and is not observed directly). Consider the example of, say, learning to label objects. In this setting the algorithm receives object $X_t$, produces a label $A_t$. The reference label $Y_t$ is then revealed and the loss function measures some kind of distance between $A_t$ and $Y_t$. A low-loss algorithm would translate into a function that correctly predicts the mapping $X_t\mapsto Y_t$.

\medskip

 There exist multiple connections between optimization algorithms, online algorithms and game theory. Fueled by machine learning applications, there has been significant progress in both theoretical and empirical understanding of online algorithms (and optimization algorithms that are tightly connected to them). Our goal is to investigate generic ways to extend this progress to algorithmic mechanism design. We start by exploring the three-way connection between optimization algorithms, online algorithms, and game theory. 

\subsection{Bandits and regret minimization}
\label{sec:11}

In this section we will briefly survey the simplest setup for online algorithms, namely expert and bandit games. These games serve as an instructive (but tractable) model for more general machine learning scenarios, and have important connections to game theory and equilibria. 

\paragraph{Setup.}
 The game is played repeatedly for $T$ time steps. At each step, the player is allowed to pull one of $K$ available arms from set $\cA$\footnote{The size $K$ of $\cA$ varies depending on the application domain. Generally, in the context of machine learning, $\cA$ corresponds to the hypothesis class and is exponentially large in $T$, while in game-theoretic applications $\cA$ may correspond to available strategies and is often smaller.}. The player incurs loss $\ell_{it}$ for pulling arm $i$ at time $t$. The player's goal is to minimize total loss
 of the sequence of pulls $\mathbf{i}=\{i_t\}_{t=1}^T$:
 $$
 L(\mathbf{i}) := \sum_{t=1}^T \ell_{i_t, t}. 
 $$
 One standard benchmark for the player to meet is to attain small {\em weak regret}, or regret against fixed strategies:
 \begin{equation}
 	\label{eq:reg1}
 	R(\mathbf{i}):= \sum_{t=1}^T \ell_{i_t, t} - \min_{j\in \cA}\left(
 	\sum_{t=1}^T \ell_{j, t}\right).
 \end{equation}
That is, the goal is to perform better (or at least not much worse) than 
the best arm in hindsight. 

An important distinction in this context needs to be made between the bandits and the experts setting. In the bandits setting, the player only learns the loss resulting from her own action, while in the experts setting she also learns the (counterfactual) loss of actions not taken. The bandits setting is more appropriate in game-theoretic scenarios, where a player does not typically know the hypothetical outcome of actions not taken. The experts setting is a  good fit for machine learning problems, where it is possible to evaluate the performance of any model on past examples. 

Generally speaking, minimizing regret as in \eqref{eq:reg1} is a well-understood problem, in both the bandits and the experts setting \cite{lattimore2020bandit}. Assume losses are bounded in $[0,1]$, and inputs are adversarial (that is, the player is allowed to randomize her strategy, and wishes to attain a low regret in expectation for {\em any} possible loss function). Then the best regret on can attain in the bandit setting is $\sim \sqrt{T\cdot K}$, and in the expert setting is 
$\sim \sqrt{T \cdot \log K}$. In particular, in the bandits setting, regret becomes $o(T)$ whenever $K\ll T$ holds\footnote{When $K> T$ we can't expect low regret without additional assumptions, since the player won't even get a chance to try all arms.}. In the expert setting regret becomes $o(T)$ whenever $K=2^{o(T)}$. 

The best (or at least conceptually simplest) algorithms for online regret minimization come in the form of multiplicative weight updates: maintain vector of ``weights" on arms (corresponding to the next arm the player will pull); 
upon learning the outcome of a pull, update this vector, penalizing the weight more if the loss was high. 
 
There are two interesting variants of the bandit problem, which we would like to mention before exploring the connections between bandits and online algorithms further. The first one has a direct connection to game theory equilibria, while the second will be important for our reductions from algorithms to VCG-based mechanisms. 

\paragraph{Swap regret minimization.} It is clear that the regret notion in \eqref{eq:reg1} is just one of many possible regret notions, and that it can be strengthened by considering a richer class of strategies with which the player must compete. For example, one may not merely consider strategies taking a single action for all $T$ periods of time, but ``two-action" strategies that are allowed to e.g. take one action $j_1$ during steps $1..T_1$ and then a different action $j_2$ during steps $T_1+1..T$. In many cases, such enhanced classes can be just thought of as enlarging the set $\cA$: in the example with two intervals, one can just think of regret with respect to all ``two-action" strategies as competing with 
all strategies from $\cA'=[T]\times\cA\times\cA$. 

Other notions of regret are internal --- in the sense that they depend on actions taken by the player. One such notion is {\em swap regret} (also known as {\em internal regret}). It is important due to its connections to correlated equilibria in game theory, as we will see later. In weak regret, the player contemplates having played the same action $j$ at each round, and compares resulting loss with her realized loss. In swap regret, the player contemplates replacing each of her actions $i\in \cA$ with a different action $\si(i)$, where $\si:\cA\ra\cA$ is an arbitrary function. Swap regret is thus defined as:
\begin{equation}
	\label{eq:reg2}
	R_{swap}({\mathbf i}):= \sum_{t=1}^T \ell_{i_t, t} - \min_{\si:\cA\ra \cA}\left(
	\sum_{t=1}^T \ell_{\si(i_t), t}\right) = \sum_{i\in \cA}\left( \sum_{t:i_t=i}\ell_{i,t} - \min_{j\in\cA} \sum_{t:i_t=i}\ell_{j,t}\right).
\end{equation}
Note that swap regret is larger than weak regret, since weak regret is captured 
by constant functions of the form $\si(i)\equiv j$. 
Even though it does not follow from general low-regret theorems that it is possible to attain low swap regret, there exists a black-box reduction from regret minimization to swap regret minimization \cite{blum2007external}. A regret bound of $O(K\sqrt{T \log K})$ can be attained \cite{stoltz2005incomplete}. Note that regret again becomes $o(T)$ for $T$ polynomially large in $K$. 

\paragraph{Bandits with knapsacks.} So far we have been dealing with scenarios where (for a sufficiently large $T$) regret per-step vanishes: total regret is $o(T)$ (and also $o(OPT)$ --- even though we have not considered this dependence explicitly). Unfortunately, in some cases it is impossible to make decisions online in a way that would lead to vanishing regrets. Specifically, when decisions between rounds are linked. 

One such generic setup is the {\em bandits with knapsacks} (BwK) model \cite{badanidiyuru2013bandits}. It represents a setting where the player is given limited amounts of some resources, which are consumed by arms. For our purposes it will suffice to consider the scenario with one resource being consumed. Let us say that the player has a budget $B$. Since the game has to stop once the budget is exhausted, it is more natural to think of the arms as providing rewards $r_{i,t}$ instead of losses. At each step, after pulling an arm $i_t$, the player learns the attained reward $r_{i_t,t}$, and the incurred cost $c_{i_t,t}$. The game stops (and no further rewards are obtained) when either $t=T$ or $\sum_{\tau=1}^t c_{i_\tau,\tau}\ge B$ --- that is, the budget has been exhausted. 

It is not hard to see why we can't hope to get an $o(OPT)$ regret bound against the best fixed strategy in the bandits with knapsacks setting. Assuming that the budget (and not time) is the main constraint, we would like to maximize the reward-to-cost ratio $r_{i_t,t}/c_{i_t,t}$ over time, weighed by $c_{i_t,t}$. This is not a decision that can be made in an online fashion. Suppose that for $t\in [1..T/2]$ arm $1$ gives the player $1$ unit of reward per unit of cost; then for $t\in [T/2+1..T]$ arm $2$ gives player $R$ units of reward per unit of cost, where $R=0$ with probability $1/2$ and $R=2$ with probability $1/2$. No matter which mixed strategy player uses, her regret will be at least $OPT/3$ at least half of the time. 

A related setting is the {\em online knapsack problem} where the costs are seen by the player before the arm is chosen, and where the player competes with the best {\em sequence} in hindsight \cite{devanur2019near,devanur2009adwords,agrawal2014fast}. Our setting is a hybrid of the online knapsack problem and the bandits with knapsacks problem: we only learn the cost after pulling an arm {\em and} are hoping to compete with the best sequence in hindsight. Throughout the paper, we will refer to our setting as the BwK setting, since the setup is the BwK setup, with the only difference being the more ambitious regret goal. While this goal appears hopeless in general, it is plausibly attainable in our application. 

An important feature making our setting easier is that not knowing the target ratio between reward and cost is the only obstacle to attaining low-regret online algorithms for bandits with knapsacks. Moreover, while it may not be possible to prevent having regret, in hindsight the player is able to tell whether $o(OPT)$ regret has been attained or not. This means that bandits with knapsacks have the potential to be used heuristically, where a good solution may not be guaranteed, but is self-certified once attained. 

\paragraph{Summary and connections to learning and to games.} 
Although the study of multi-arm bandit problems is still subject to very active research, it is fair to say that the problem is very well understood. Absent a simple information-theoretic obstacle (such as having to pull each arm at least once, or not having knowledge of future reward-to-cost ratios) there are algorithms attaining optimal or close-to-optimal regret bounds. In addition, these algorithms are efficient {\em in the number of arms $K$}. 

Most machine learning tasks can be cast as bandit problems, where ``loss" is the gap between predicted label and ground truth. In fact, much progress on bandit problems originated from the machine learning theory literature. This link is not without limitations, however. When trying to learn a predictive model, the space of arms is typically exponential in the number of parameters. Thus algorithms with running time polynomial in the number of arms cannot be used without modifications. In addition, in many cases, such as neural-net models, practically attainable regret values are significantly better than the ones guaranteed by generalization bounds. 

The typical way in which machine learning algorithms convert a search problem with exponentially many model candidates into a tractable one is by turning it into an optimization problem. Typical bandit (and expert) algorithms have to maintain a vector of dimension $K$ tracking the performance of each arm over time. When $K$ is exponential, one cannot hope to do that, and has to settle for maintaining the ``best arm so far". Here, ``best" is a combination of retrospective loss and simplicity (to avoid overfitting), and defining what ``best" means is an important art within machine learning practice. We will dive deeper into these questions in the next section.

\subsection{Online algorithms and learning based on empirical regret minimization}
\label{sec:12}

In this section we will explore the connection between online algorithms and optimization. We frame the discussion in machine learning terms, but the same applies to any online algorithmic task. 

Suppose we are trying to learn a model $A\in \cA$ mapping inputs $X_t$ to 
labels $Y_t$.  A na\"ive ``follow the leader" approach would be to always propose
the best strategy in hindsight (known as ``follow the leader"):
\begin{equation}
	\label{eq:1}
	A_t:=\arg\min_{A\in \cA} \sum_{i=1}^{t-1} L(X_i, A(X_i),Y_i). 
\end{equation}
It turns out that this approach underperforms both in theory and in practice. 
Most importantly, (both in theory and in practice) is that using \eqref{eq:1} will lead to {\em overfitting}: the model $A_t$ will perform better than expected on the training samples $\{X_1,\ldots,X_{t-1}\}$, and worse than expected on the test sample $X_t$. A second problem is a computational one: assuming the model $A$ is non-convex in its parameters, the function in the RHS of \eqref{eq:1} is non-convex in the parameters of $A$. This makes $A_t$ both computationally difficult to find, and potentially unstable in the inputs. 

A generic solution to the overfitting problem is using a {\em regularizer}\footnote{Another solution --- in line with expert algorithms discussed in the previous section --- would be to maintain a number of models weighed according to their past performance, and to aggregate them together. This is generally computationally too cumbersome --- instead of using $d$ models of size $S$ one can train a single model of size $d\cdot S$, which will perform better. The regularizer in some sense serves as a proxy for maintaining many ``good" models, and combining them to make the prediction. } function $\Psi(A)$, and choosing the strategy (known as ``follow the regularized 
leader"):
\begin{equation}
	\label{eq:2}
	A_t:=\arg\min_{A\in \cA}\left(\Psi(A)+ \sum_{i=1}^{t-1} L(X_i, A(X_i),Y_i)\right). 
\end{equation}

Using a regularizer leads to better performance guarantees. Intuitively, $\Psi$ induces a metric on which models $A$ are more likely to occur --- typically ``simpler" models according to some notion of simplicity\footnote{If one does Bayesian maximum likelihood estimation then $\Psi(A)$ literally comes from 
a prior distribution of models.} --- and prevents over-fitting by penalizing large deviations to accommodate a small number of examples. In terms of optimization, if $\Psi(A)$ is a nice (e.g. strongly convex) function, then one can hope that, at least locally around the previous optimum $A_{t-1}$ the function
$$
F_t(A):=\Psi(A)+ \sum_{i=1}^{t-1} L(X_i, A(X_i),Y_i) = 
F_{t-1}(A)+ L(X_{t-1}, A(X_{t-1}),Y_{t-1})
$$ 
will appear convex. Then local gradient descent will allow us to get from $A_{t-1}$ to $A_t$. Note that the gradient of $F_t$ at $A_{t-1}$ is given by 
\begin{equation}
	\label{eq:3}
\nabla F_t(A_{t-1}) = \nabla F_{t-1}(A_{t-1})+\nabla  L(X_{t-1}, A_{t-1}(X_{t-1}),Y_{t-1}) = \nabla  L(X_{t-1}, A_{t-1}(X_{t-1}),Y_{t-1}), 
\end{equation}
since $A_{t-1}$ minimizes $ F_{t-1}(A)$, and thus its gradient is zero. Therefore, after seeing the pair $(X_{t-1},Y_{t-1})$, $A_{t-1}$ will need to
move only in the direction reducing the loss of $A_{t-1}$ on $X_{t-1}$ --- leading to significant computational savings. This step is known as back-propagation in the neural networks training literature. 

There are additional improvements that can be made to \eqref{eq:3} to speed up convergence and improve generalization. The local geometry around $A_{t-1}$ may be transformed --- to make the function $F_t$ more isotropic (using adaptive learning rates) --- so that gradient descent converges faster. In addition, in the context of model training, where labeled data is usually scarce, multiple passes over the samples are used (so $t$ pairs $(X_i,Y_i)$ may be sampled $T\gg t$ times with repetition). Given the practical importance (and the significant investment) in developing and training machine learning models, there is a significant body of applied and theoretical knowledge about each of the steps described here. It is not our goal to survey this knowledge. Rather, let us summarize some higher-level points, which will guide us in suggesting meta-mechanisms based on online optimization algorithms of the form \eqref{eq:2}. 

\paragraph{Upper bounds: proof-to-applications pipeline.} In the context of optimization (linked to online algorithms or otherwise), there is a stark divide between the convex and the non-convex case. Generally speaking, convex optimization (where the loss function and the regularizer are convex functions and the domain of possible models $\cA$ is an efficiently-specified convex set) is tractable. Oversimplifying decades of research, versions of gradient descent can be used to solve these problems efficiently. Many purely ``combinatorial" problems such as maximum bipartite matching are in fact tractable because they are instances of linear programming (an important special case of convex programming). 

In the context of learning and online algorithms, a promising approach has been to use the convex case to derive rigorous performance bounds, and then port them as heuristics into the more general non-convex model classes that one wants to train in practice. Thus the pipeline is (1) prove rigorous performance results about the convex case (e.g. in terms of accumulated loss, running time etc.); (2) use the same algorithms (or their natural extensions) in the more general setting. While (1) is a mathematically robust exercise, (2) is a matter of accumulated wisdom about what kind of things are likely to port into the non-convex domain. A necessary (but not sufficient) condition for (2) to work is 
that the optimization procedures suggested in (1) tend to be local and continuous: local updates based on local quantities such as gradients and Hessians. 

\paragraph{Lower bounds are overly pessimistic.} In many cases, a procedure suggested by (2) will at least lead to some kind of a local optimum --- one cannot hope to do much better provably, because global optimization of non-convex functions is almost always \NP-hard --- a fact that does not appear to be the main bottleneck in achieving performance (as will be discussed later, approximation and generalization errors seem to play a more dominant role). 

More importantly, in many cases, the resulting models significantly outperform generalization bounds. While the exact theoretical mechanisms explaining this are subject of active research, this is a phenomenon that mechanisms built on top of such algorithms should be prepared to take into account. 

\paragraph{Approximation and generalization error: bound to be an art.} 
In general, an online algorithm (or equivalently a learning algorithm) that 
takes actions based on a model it trains suffers from three sources of loss: (A) {\em approximation error:} how close is the best model $A$ in the class $\cA$ is to the truly best model? (B) {\em optimization error:} in the language of \eqref{eq:2}, how close is $\tilde{A}_t$ obtained by the algorithm to minimizing the expression in \eqref{eq:2}? (C) {\em generalization error:} how close is the optimizer of \eqref{eq:2} to producing the smallest possible regret?  

Except when $\cA$ is the set of all possible functions, designing the class $\cA$ to attain a small approximation error is an important application-specific task, an in many cases it is more of an art than an exact science. This is especially true since the choice of $\cA$ affects other sources of loss. 

Classical learning theory, such as PAC-learning and multi-arm bandit theory, allow one to give rigorous bounds on the generalization error based on optimization error. Roughly speaking, if a model performs well on a randomly selected training set, and it does not have enough parameters to be able to overfit to the training data, 
then its performance on $X_{t}$ must be in line with its performance on $X_{1..t-1}$. Observed performance in training neural nets often significantly exceeds these guarantees --- generalization error is typically estimated empirically by examining the model's performance on a holdout set. Being in a regime where a low generalization error is an empirical fact and not mathematically guaranteed means that the optimization heuristic may affect generalization performance --- an optimization heuristic attaining the lowest optimization error may underperform a heuristic attaining a higher optimization error, but a lower generalization error. 

\paragraph{Summary.} The main upshot of the discussion so far is that some of the more important modern algorithms are a result of domain-specific experience and are not easily replaced with a functionally equivalent algorithm based on the problem the algorithm is trying to solve. This is in contrast to classical discrete algorithms for problems such as maximum matching or network flow, where all correct algorithms will output the same (correct) answer. 

Our goal is to develop new reductions from algorithms to mechanisms, and 
in this context this means that {\em the reduction should happen at the level of the algorithm and not at the level of the problem that the algorithm is trying to solve}.  In other words, given a heuristic $\cH$ for a problem $\cP$, the mechanism should assume that it won't be able to attain a comparable performance on $\cP$ without using $\cH$ (or a version of it) as a sub-routine.  This is a departure from most existing algorithmic mechanism design literature, which we feel is necessary in order to keep up with advances in applied algorithms. 

A very important special case is when $\cH$ is just a local optimization (such as stochastic gradient descent) aimed at minimizing an expression of the form \eqref{eq:2}. In this case, the algorithmically difficult part is devising the class of models $\cA$ and the regularizing function $\Psi$, and we would like to develop generic mechanisms that make use of these while having good game-theoretic properties. 

\subsection{Games, online algorithms, and equilibria}

So far, we have seen that even ``one shot" optimization, such as finding 
the best classification model given labeled data can be naturally cast as an 
online algorithms problem. Next we will see that the same is true about 
game theory, making online optimization a natural language to connect the two. 

When an algorithm is turned into a mechanism (by allowing inputs to come from self-interested participants), it induces a strategic game among the participants (in which the mechanism --- or the ``principal" is sometimes a party as well). The basic question facing the mechanism designer is ``what outcome will this game lead to?". One notion of a plausible strategic outcome is that of an {\em equilibrium}: a steady state in which no player benefits from deviating from their current strategy. Of particular importance for mechanism design are mechanisms  that ask participants for their inputs (``direct revelation mechanisms"), and where in the induced game reporting inputs truthfully is an equilibrium\footnote{The reason for focusing on
direct-revelation mechanisms is something called the ``revelation principle". The revelation principle asserts that any mechanism can be converted into a truthful direct revelation mechanism by appointing a perfectly informed advocate for each player as part of the mechanism. A player then reveals her type (truthfully) to the advocate, who uses this knowledge to interact with the mechanism in a way that maximizes player's happiness. The extent to which this reduction is realistic or practical is a very important question whose answer depends on the setting. Regardless, it is clear that truthful direct revelation mechanisms are the most natural extension of algorithms to which one should aspire. }. 

Unlike being ``optimal" in the combinatorial sense (a solution's objective value is close to the best objective value attainable), a ``good" or ``optimal" equilibrium is very much a function of the mechanism implementation details. This leads to significant complications, both in theory and in practice, since whether a good equilibrium can be sustained is a function of participants' behaviors in practice (or modeling assumptions in theory). A significant portion of games, both in theory and in practice, have multiple equilibria, and it is often impossible to rule out ``bad" equilibria. Moreover, in some cases there are lower bounds known as ``the price of anarchy" showing that {\em all} equilibria attain a substantially lower objective function value than the combinatorially optimal outcome. 

Despite these challenges, equilibria become significantly nicer objects to deal with once they are presented in the language of {\em online regret minimization} (corresponding to correlated equilibria). In addition, at least when money can be used arbitrarily, the {\em Vickrey–Clarke–Groves (VCG)} mechanism attains combinatorially optimal performance, while incentivizing participants to report their types truthfully. As we will see, VCG is not a ``cure-all" mechanism since suffers from several important shortcomings that make it more appealing in theory than in practice. One of our goals is to mitigate some of the shortcomings while preserving its desirable properties in a generic way.

\paragraph{Nash equilibria.}  To keep the exposition simple, consider a basic two-player strategic game, where row player Row and column player Col each pick actions $i$ and $j$, respectively,  from a set of $n$ available actions. On actions $(i,j)$ the payoff or Row is given by the matrix $R_{ij}$, and the payoff of Col is given by $C_{ij}$. A Nash equilibrium $(p,q)$ is a distribution of actions by the two players such that no player benefits from deviating. For Row, it means that no action in the support of $p$ is strictly dominated by another action Row may take. Let $U_i:=\sum_j q_{j} R_{ij}$ be Row's expected payoff 
under action $i$. Then the equilibrium condition can be written as:
\begin{equation}
	\label{eq:equ1}
	\text{whenever }p_i>0 \Rightarrow \text{for any alternative action }i'~~U_{i'}\le U_i
\end{equation}
It should be noted that while Nash equilibria are perhaps the best-known notions of equilibria, they are arguably not the best-suited in the context of algorithmic mechanism design. To argue that action distributions $(p,q)$ are a plausible answer to the question ``What will Row and Col do?", one needs to assume that e.g. Row is perfectly informed about the distribution $q$ so that the non-zero-probability actions $i$ with $p_i>0$ make sense under \eqref{eq:equ1}. If Row is misinformed about $q$, the equilibrium may fail to materialize. This is especially true in a game induced by a mechanism with many participants. A more robust notion of an equilibrium would be based on participants making decisions regardless of their beliefs about others' actions. 

\paragraph{Dominant strategy equilibria.} In a dominant strategy equilibrium,
no player takes an action that is dominated by another action for some realization by the other players. In the two-player example, Row would not take an action that is dominated by another action {\em for some action $j$ of Col}:
 \begin{equation}
 	\label{eq:equ2}
 	\text{whenever }p_i>0 \Rightarrow \text{for any alternative action $i'$ and for any $j$}~~R_{i'j}\le R_{ij}
 \end{equation}
Condition \eqref{eq:equ2} is clearly much stronger than \eqref{eq:equ1}. In particular, whenever it holds, it is easier to believe that the outcome $(p,q)$ will be realized. This is especially true when $p$ and $q$ are just single actions, i.e. $p_i=1$, $q_j=1$ for some $(i,j)$. 

Unfortunately, it is easy to see that dominant strategy equilibria do not always exist --- for example there is no ``best" strategy in the Rock-Paper-Scissors game (or any zero-sum game for that matter). On the other hand, Nash's celebrated theorem guarantees the existence of a Nash equilibrium in any game.
In mechanism design we typically have some degree of control over the game, and can aspire for an equilibrium where truthful reporting of one's type is, in fact, a dominant strategy. This is often impossible to attain, a slightly less ambitious goal is for truthful reporting to be an {\em approximately} dominant strategy.   

\paragraph{Approximate equilibria.} For any notion of an equilibrium such as above, there is an associated notion of an {\em approximate} or an $\ve$-equilibrium. The approximation here refers to the benefit a player can derive by deviating from her current action profile. Assuming a player's utilities for outcomes are in a bounded interval $[0,C]$, a set of actions is an $\ve$-equilibrium if no player can improve her utility by more than $\ve \cdot C$ by deviating. Thus, assuming $R_{ij}\in[0,1]$, \eqref{eq:equ1} becomes the condition 
\begin{equation}
	\label{eq:equ3}
	\text{whenever }p_i>0 \Rightarrow \text{for any alternative action }i'~~U_{i'}\le U_i+\ve
\end{equation}
for being  an $\ve$-Nash equilibrium, and \eqref{eq:equ2} becomes the condition 
 \begin{equation}
	\label{eq:equ4}
	\text{whenever }p_i>0 \Rightarrow \text{for any alternative action $i'$ and for any $j$}~~R_{i'j}\le R_{ij}+\ve
\end{equation}
for being $\ve$-dominant strategy equilibrium. 

We note that while $\ve$-approximate equilibria are easier to find and attain\footnote{For example, while finding a Nash equilibrium, even for a two-player game, is {\sf PPAD}-complete \cite{daskalakis2009complexity,chen2006settling}, and finding a good-value Nash equilibrium is {\sf NP}-complete \cite{conitzer2008new}, both problems can be solved in quasi-polynomial time $n^{O(\log n)}$ in the $\ve$-approximate setting \cite{lipton2003playing}. }, a $\ve$-dominant-strategy equilibrium may still not exist. An important advantage of $\ve$-approximate equilibria is that they can be tied into online learning and regret bounds. 

\paragraph{Learning, regret minimization, and correlated equilibria.} Suppose a strategic game or mechanism were presented to someone with no prior knowledge of game theory, with the question of `what will happen?'. A reasonable approach would be to run simulations with participants trying to ``learn to play" the game to the best of their ability, and to see what happens. 

As defined, the game is played only once, and ``learning" as such doesn't make sense. A reasonable solution is to let the players play the game $T\gg 1$ times in a row, and observe the distribution to which their actions converge. If a player ignores the effect her play may have one future plays by other players\footnote{This is a very important simplifying assumption --- repeated games where stages are linked are often much more complicated than the base game.}, then the problem she is facing is exactly the online learning problem we discussed earlier. 

Considering the two-player setting for simplicity, and taking Row's viewpoint, at each round $t=1..T$, if Col plays $j_t$, Row faces payoff $R_{i j_t}$ for action $i$ at step $t$. Row will be solving the online learning problem with the goal of maximizing payoff $\sum_{t=1}^T R_{i_t j_t}$. Col will be solving a similar problem. At a minimum, Row and Col should be running a low-regret learning algorithm, although we should note that any family of online learning algorithms would lead to an outcome with potentially interesting properties. 

Running two online learning algorithms will unfortunately not lead to a Nash equilibrium. Still, the outcome of such a process has an important interpretation: it leads to something called a {\em coarse correlated equilibrium}. Moreover, if the players run a low swap regret online learning algorithm (of the kind discussed in Section~\ref{sec:11}), then the resulting outcome is an $\ve$-correlated equilibrium. Here $\ve\ra 0$ as $T$ grows. When $T\ra\infty$ we obtain a {\em correlated equilibrium}. 

A correlated equilibrium is an equilibrium where a suggested action is presented to each player. The players are free to not follow the suggested action, and instead choose a different action (which may depend on the suggested action). The suggested actions form a correlated equilibrium if no player gains by deviating from the proposed actions. Formally, in the two-player case, the correlated equilibrium is a distribution $\mu$ on pairs of strategies $(i,j)$, such that no player benefits from not following the suggested play. For Row, this condition becomes:
\begin{equation}
	\text{for all $i$ and potential substitutes $i'$, }\sum_{j} \left(\mu(i,j) \cdot R_{ij} \right)\ge 
	\sum_j\left( \mu(i,j) \cdot R_{i'j}\right)
\end{equation} 
A canonical example of a correlated equilibrium is one induced by a traffic light, where it is a dominant strategy for each driver to stop on red, expecting crossing traffic to not stop on green.

\paragraph{Competitive equilibria and markets.} More pertinently for mechanism design, a market that sets prices is also a form of a correlated equilibrium. Market-based solution concepts for reallocation of goods, such as Fisher and Arrow-Debreu markets, are based on a concept of a {\em competitive equilibrium}, which is a natural type of a correlated equilibrium. 

For example, in the case of a Fisher market, each player is given an endowment, and wishes to spend it on a bundle of (divisible) goods. A solution to the Fisher market problem produces a vector of prices $\vec{p}$ for the goods. Each player then spends her endowment to buy her favorite bundle at the given prices. The prices ``clear" the market if all players spend their budget, and all goods are sold. Given the prices $\vec{p}$, each player gets her favorite bundle at these prices, and therefore truthful reporting of valuations over goods is a dominant strategy, and we obtain a correlated equilibrium where truthful reporting is a dominant strategy.

Note that this sidesteps the question of {\em how} prices $\vec{p}$ are obtained, and indeed, solutions to Fisher markets do not yield a truthful mechanism if one considers the effect players' reports have on prices. This parallels a broader points about market-based mechanisms: if one treats market prices as fixed, then interaction with the market is typically truthful. However, when a player considers her impact on market prices, most market mechanisms are not truthful (for example, one can try to feign reduced interest in an item to get its price to drop). This concern is very real (and leads to reduced overall welfare) when there are few players in the market. When the market is large --- and the impact of each individual player on the market is small, one can hope that the resulting mechanism will be approximately truthful even when one accounts for a player's impact on prices: the impact on the prices (and thus the potential benefit) of misreporting preferences is small, and the effect of misreporting on the bundle one gets is always non-positive. 

The upshot of the discussion above is that in market-based solution concepts, a competitive equilibrium is a type of a correlated equilibrium where reporting preferences truthfully is a dominant strategy --- as long as we manage to sidestep the issues of how prices are arrived yet. Taking a cue from the Nash equilibria $\rightarrow$ correlated equilibria simplification, a natural source of these prices is through repeated play. Note that this is how prices in `real' large markets are discovered: participants repeatedly interact with the market, with supply and demand serving as signals that update market prices. This process has similarities to t\^atonnement in equilibria theory, except we will consider the time-average of the outcomes at all steps, and not just a ``limit" outcome --- giving our model more flexibility. 

\paragraph{Prices and outcomes through repeated play.} The discussion above gives us a blueprint for producing an equilibrium outcome using individual preferences and an aggregation heuristic. We will formalize parts of it as the \apex algorithm in Section~\ref{sec:apex} below. 	
	 It consists of the following components:
\begin{enumerate}
	\item A repeated game where at step $t$ each player $i$ submits its preferences function $f^i_t$ to a central principal;
	\item at each step $t$ the principal runs a heuristic to produce an outcome $o_t$; 
	\item at each step $t$ the principal calculates prices to be charged to participants (since we're dealing with a mechanism without money, prices are charged in tokens); 
	\item 
	participants are not allowed to exceed their token budget;  
	\item in each step $t$ the principal also outputs prices which allow each player $i$ to estimate the price and outcome $o_t^{i}(\tilde{f}^i_t)$ 
	under reported preferences $\tilde{f}^i_t$ instead of $f^i_t$;
	\item 
	the outcome of the mechanism is the time-average $\bar{o}:=\frac{1}{T}\sum_{t=1}^T o_t$; 
	\item 
	each player runs an online learning strategy with the goal of maximizing its utility; here we make a distinction between an ``equilibrium" and a ``competitive equilibrium" notion of maximizing utility\footnote{In the competitive equilibrium approach, each player ignores her effect on future plays by other participants of the mechanism, treating the other players and a principal as ``nature" in the sense of online algorithms.}.
\end{enumerate}

Players in the mechanism above will run a bandits with knapsacks online algorithm. If incentives are correctly aligned, in each step, players will report their types truthfully up to a constant scaling factor\footnote{Such factors are unavoidable in mechanisms without money: a player with utility function $u_i$ should be treated the same a player with utility function $2\cdot u_i$, since there is no functional means of distinguishing the two.}. Thus if the normalized preferences function of player $i$ at time $t$ is $\hat{f}_t^i$, the player will report $f_t^i = \la_t^i\cdot \hat{f}_t^i$, where $\la_t^i\ge 0$ are chosen to that the reward per marginal token spent is equalized across rounds. 

There are important details to be filled in the above blueprint, primarily 
around the principal's heuristic and the prices it would induce. 
In the spirit of the discussion about algorithms and optimization, we will not want to limit the scope of possible heuristics, except we would expect the outcome $o_t$ to try and maximize $\sum_i f^i_t (o_t)$, possibly with a regularization term. In practice this might mean either computing $o_t$ from scratch or computing it from $o_{t-1}$ via some kind of gradient update. 

The previous part of the description is necessarily vague, since it needs to accommodate various types of optimization algorithms and heuristics. Assuming the algorithm for converting the $f^i_t$'s into an outcome $o_t$ is a good one, we still need to take care of incentivizing players to report their preferences $f^i_t$ truthfully. 

There is a generic tool in mechanism design, called the 
Vickrey–Clarke–Groves (VCG) mechanism under which truthful reporting is a dominant strategy. While theoretically the VCG mechanism is very appealing, it has some significant practical drawbacks that stand in the way of it being adopted. We will argue that in our case most of these drawbacks either don't occur\footnote{For example, because we are considering mechanisms without money, and thus do not need to be ``budget-neutral".}, or would occur to the same extent under other mechanisms. 

The final  piece of our mechanism will be using a local version of VCG to set prices accruing to the players. 
Before putting all the pieces together more formally, let us briefly discuss the VCG mechanism, and some challenges in using it in practice.

%

\subsection{Mechanism design: the VCG mechanism and its shortcomings}
\label{sec:14}

The VCG mechanism is  a mechanism with money --- meaning that there is a way for participants to store residual utility after the mechanism completes. In this paper we are dealing with mechanisms without money, but since the mechanism is multi-round, tokens serve the role of money, allowing us to use VCG locally. 

\paragraph{The mechanism.} Suppose there is an outcome space $\cO$, and $n$ players. Player $i$ has utility $u_i(o)\ge 0$ for an outcome $o\in\cO$. Since we are dealing with a mechanism with money, we can assume that $u_i$ is in currency units. The mechanism will choose an outcome maximizing {\em total utility}:
 \begin{equation}
 	\label{eq:VCG1}
 	o_{VCG}:= \arg\max_{o\in\cO}\sum_{i\in[n]} u_i(o). 
 \end{equation}
Each player is then charged $p_i$ --- the calculated {\em externality} she imposes on other players:
 \begin{equation}
	\label{eq:VCG2}
	p_i(o_{VCG}):= \max_{o\in\cO}\sum_{j\in[n];~j\neq i} u_j(o)-
	\sum_{j\in[n];~j\neq i} u_j(o_{VCG}). 
\end{equation}
In other words, $p_i$ is the extra utility other players could have attained if they didn't need to take $i$'s preferences into account. Note that the realized utility for player $i$ from outcome $o'$ is given by
$$
u_i(o') - p_i(o') = u_i(o')+
\sum_{j\in[n];~j\neq i} u_j(o') - \max_{o\in\cO}\sum_{j\in[n];~j\neq i} u_j(o) =
\sum_{j\in[n]} u_j(o') -  \max_{o\in\cO}\sum_{j\in[n];~j\neq i} u_j(o).
$$
The second term does not depend on $u_i$, and the first term is maximized when $o'=o_{VCG}$, which is obtained when player $i$ reports her type truthfully. Therefore VCG is dominant-strategy truthful. An example of a problem VCG ``solves" in principle is that of combinatorial auctions: optimally selling goods to players who may have arbitrary preferences over bundles of goods. It should be noted that in the case where the goal of the mechanism is to sell a single item (that is, $\cO=[n]$ determines which player gets the item, and $u_i(o):= \mathbf{1}_{o=i}\cdot u_i$), VCG turns into the second-price auction. 

There is a number of important practical reasons why VCG in its pure form has remained primarily a theoretical tool. An analysis of some of the more important issues can be found in \cite{ausubel2006lovely,rothkopf2007thirteen}.  We will have to keep these issues in mind when we integrate a version of the VCG mechanism into our reduction. For our purposes, the issues can be broken down into several categories. We present these in order of relevance, from the least to the most relevant.

\paragraph{Revenue sub-optimality; budget non-neutrality.} One of the main reasons VCG is not used in actual auctions is that while it maximizes participants' welfare, it does not generally maximize the principal's revenue. In fact, in some cases, such as unit-demand auctions, it is the mechanism {\em minimizing} the amount of revenue raised. In part, the fact that VCG is maximally efficient already suggests that it won't be revenue-maximizing: in practice, the way to fetch a higher price for a good is to be willing to set a reservation price, and not sell it for a lower price even if withholding it is inefficient. A related problem is that VCG is not ``budget neutral" --- even when the goal of the mechanism is to facilitate a transaction between two players, the mechanism might prescribe payments to/from players that do not add up to $0$, requiring an outside subsidy for the mechanism to run. 

Neither of these are a problem for us, since we will be using VCG in the context 
of a mechanism without money, with players endowed with tokens. The objective is to maximize aggregate utility subject to some notion of fairness (such as players starting with an equal token endowment). 

\paragraph{Computational difficulty of bidding; optimization and numerical instability of payments and utility.} The second set of problems for using 
VCG in practice is computational. For an individual player, figuring out the function $u_i$ and communicating it to the mechanism may be prohibitively expensive. In addition, solving the optimization problem \eqref{eq:VCG1} precisely is often {\sf NP}-hard. 

More importantly, even if \eqref{eq:VCG1} can be solved heuristically to a high degree of precision, the effect of the approximation error on price calculations in \eqref{eq:VCG2} may be prohibitive, since it involves a difference between two approximate quantities. To illustrate, if there are $n=1{,}000$ participants in an auction involving $\$5{,}000{,}000$ worth of goods, then getting the optimal allocation of goods to within $1\%$ (or $\$50{,}000$) is acceptable, but calculating prices charged to individual players to within an additive $\$50{,}000$ (where the average purchase is only worth $\$5{,}000$) is unacceptable.

In our setting, even when the space $\cO$ is quite complicated, we will only need gradient access to the function $u_i$, so the bidding complexity will not be a problem. The optimization gap is a real concern. Depending on the setting, the optimization problem may be solvable exactly, in which case it is not an issue. If the optimization is being done by a heuristic, we will rely on the fact that at each step of the optimization $o_t$ is only adjusted locally, and even very complicated functions can be simplified locally (e.g. by taking a quadratic approximation), allowing us to compute local prices with an acceptable precision. One tool at our disposal is regularization, which at every step will turn the optimization problem into a locally convex one.   

\paragraph{Issues with chaining several VCGs one after another.} A significant issue for the truthfulness of the VCG mechanism, happens if multiple instances of the mechanism are chained one after another with players given a fixed total budget for all rounds.  It might be beneficial for a player to withhold bids in one round, and use tokens saved to bid in later rounds. Generally speaking, in situations where utilities are concave, and player $i$ is bidding $\la_i \hat{f}_i$, increasing $\la_i$ will result in a lower marginal utility per token. Therefore, assuming the algorithm converges, we can expect that in a typical round the marginal utility of player $i$ per token is about the same. However, the exact conditions for convergence will likely require further investigation and analysis. 

\paragraph{Susceptibility to various forms of cheating and collusion.} While VCG is immune to manipulation via misreported preferences, it is extremely susceptible to other forms of manipulation. Considering the simple second-price auction scenario, the mechanism is susceptible to shill bidding (a player colludes with the principal to extract more than the second price from the winner) and non-winning players being bribed to drop out (to reduce the price the winner has to pay). More sophisticated scenarios are also susceptible to 
a single player bidding under multiple identities. 

Some of these problems disappear in a mechanism without money. For example, in voting it is clear that a single player can benefit by ``voting under multiple identities", and preventing this from happening falls outside of the voting mechanism. In other cases, collusion between players is inevitable, and cannot be prevented by the mechanism. In case of voting, even an approximately truthful voting mechanism will be susceptible to voters forming a party and then voting as a block in favor of issues they all agree on, while avoiding canceling each other on issues they disagree on. 

In summary, unlike computational issues which we can hope to do away with, some of the issues around collusion are real, and will not disappear. We do not expect that using VCG will exacerbate these issues compared to other mechanisms, but this will need to be investigated further.

\subsection{New results on one-sided allocation}

One area where we have obtained new theoretical results using the \apex framework is one-sided allocation. The results are presented in detail in Section~\ref{sec:HZ}, we summarize them briefly here. We should emphasize that we didn't set out to obtain these results, and that they followed naturally by applying the framework to the one-sided matching setting. 

In the simplest one-sided allocation setting there are $n$ players and $n$ items. Player $i$ has utility $u_{ij}\in[0,1]$ for item $j$. A general solution is a bi-stochastic matrix $\{X_{ij}\}$, where player $i$ gets item $j$ with probability $X_{ij}$. The utility of player $i$ under such allocation is just her expected utility $$u_i(X)=\sum_{j} u_{ij} \cdot X_{ij}.$$ 

A classical result of Hylland and Zeckhauser \cite{hylland1979efficient} says that there is always a competitive equilibrium from equal endowments (CEEI) solution to the one-sided allocation problem. Informally, it means that if we give each player one unit of tokens, there is an allocation $X_{ij}$ and prices $P_j\ge 0$ (in tokens) on items such $X$ is a competitive equilibrium supported by prices $P_j$: the bundle $X_i$ costs at most one token, and is utility-maximizing for player $i$ among all cost-$\le 1$ bundles of total probability $1$. 

Note that the HZ competitive equilibrium has nice properties such as Pareto-efficiency and envy-freeness. On the other hand, there could be multiple HZ-equilibria (existence proof uses Kakutani's fixed-point theorem). We obtain the following refinement of the HZ-equilibrium existence:

\medskip
\noindent {\bf Theorem~\ref{thm:VCG-HZ}.~~}{\em [restated]} {\em For any $u_{ij}$, there exist scaling factors $\la_i\ge 0$ and an allocation $X$ such that $X$ is a result of running VCG on utilities $\la_i u_{ij}$. The resulting VCG prices $C_j$ support $X$ as a HZ-equilibrium. In the resulting VCG payments supporting $X$, all players either get their favorite item or pay exactly $1$ unit for their bundle.}
\medskip

Thus, there is always a HZ equilibrium supported by VCG prices applied to players' scaled utilities. It turns out that equilibria from Theorem~\ref{thm:VCG-HZ} form a proper subset of all HZ equilibria --- there exist HZ equilibria that are not supported by VCG prices. 

In addition to the existence result, in Theorem~\ref{thm:reg} we show that there is a natural online optimization dynamics on our general mechanism, such that whenever that dynamics converges it leads to a HZ equilibrium of the form guaranteed by Theorem~\ref{thm:VCG-HZ}. This gives a new attack route for both provable and heuristic approaches to calculating HZ equilibria. 

\subsection{Related works}

\paragraph{Note.} This section will be updated as I collect more relevant works across the different domains. 

\medskip 

The main thrust of the paper is to build a new three-way connection between {\em optimization}, {\em online learning}, and {\em mechanism design without money} --- particularly the {\em VCG mechanism}. Each of these topics forms a subject of a major discipline in Applied Mathematics and Economics. Each pair of these topics is also the subject of a significant body of work (at a level where textbooks or whole conferences dedicated to the subject exist). We will very briefly survey those here, before mentioning some more directly relevant works. 

\paragraph{Online learning and optimization.} The connection between online learning and optimization is well-established. In the convex setting, the textbook \cite{hazan2019introduction} provides a recent treatment of the subject. 

\paragraph{Optimization and mechanism design} is the subject of much of modern Algorithmic Game Theory (AGT) \cite{NisaRougTardVazi07,roughgarden2016twenty}. The effort to convert good algorithms (typically optimizing an objective) into a good mechanism is at the core of AGT. Our work is also part of this effort, with the added twist of viewing the optimization component of the algorithm as an iterative process similar to online optimization. 

\paragraph{Online learning and mechanism design} is perhaps the least developed of the three connections. Within classical game theory, it has been known that correlated equilibria correspond to online learning dynamics. More recently, the subject of learning in repeated games has received renewed interest due to its practical importance in areas such as online ad auctions. Recent references include Chapter~11 in \cite{slivkins2019introduction}, as well as articles such as \cite{braverman2019multi,feng2020intrinsic,deng2019strategizing}.

\paragraph{Specific related recent works.} Below we briefly discuss recent papers that are most closely related to the present one. 

In 
\cite{kandasamy2020mechanism} a framework for online learning with incentives is developed in the context of mechanism design {\em with money}. Participants learn their value for the different options as the algorithm progresses. The construction uses a combination of online learning techniques and the VCG mechanism to achieve both low regret and good incentive properties. 

 \cite{immorlica2019equality} develops a framework for joint decision making in a metric space with quadratic utilities. The primary goal of the work is actually to obtain a decision-making algorithm that complies with the normative requirement of ``equalizing influence among participants". The resulting outcome notion is in fact very similar to the notion of a competitive equilibrium from equal budgets in the present paper.

\section{Pseudo-market mechanisms}

With the components in place we are ready to start putting together generic mechanisms based on online optimization and other algorithmic heuristics. 

\subsection{Setup and the generic \apex mechanism}
\label{sec:apex}

\subsubsection{The mechanism}

We begin by stating a very general reduction from optimization heuristics 
to algorithms. We will then instantiate it in ways that seem to be most immediately useful. 

Our starting point is an online learning heuristic $\cH$. The heuristic takes a sequence of objective functions $\{F_s(x)\}_{s=1}^{t}$, and starting point $x_t$. It then generates a function $\Psi_{t+1}=\Psi_{t+1}(F_1,\ldots,F_t,x_t)$, and outputs a value $x_{t+1}$ that maximizes the function $\Psi_{t+1}(x)$. Typically, $\Psi_{t+1}$ will either be a concave function, or contain a regularization term that drops off sharply away from $x_t$, thus making computing $x_{t+1}$ from $x_{t}$ easy.

We will particularly focus on the effect $F_t$ has on $x_{t+1}$. Regularization, along with the fact that $F_t$ is only one of the functions feeding into $\Psi_{t+1}$ means that we may expect the dependence of $x_{t+1}$ on $F_t$ to be smooth even if the overall landscape of $\Psi_{t+1}$ is very complicated. For an alternative objective function $\ti{F_t}$ we can define  
$$\ti{\Psi}_{t+1}:=\Psi_{t+1}(F_1,\ldots,\ti{F}_t,x_t),$$
and let be $\ti{x}_{t+1}$ the outcome of maximizing $\ti{\Psi}_{t+1}(x)$ over $x$. 
For illustration purposes, one property we expect $\cH$ to satisfy is {\em monotonicity}, which can be viewed as a local relaxation of finding an actual maximizer:

\begin{definition}
	\label{def:mon} Heuristic $\cH$ is said to have the {\em monotonicity} property if for all $F_1,\ldots,F_t,\ti{F}_t$, and $x_t$, the following holds:
	\begin{equation}
		\label{eq:mon0}
		F_t(x_{t+1})-F_t(\ti{x}_{t+1}) \ge 
			\ti{F}_t(x_{t+1})-\ti{F}_t(\ti{x}_{t+1}).
	\end{equation}
In other words, moving from $\ti{F}_t$ to $F_t$ results in a shift more beneficial to $F_t$ than to $\ti{F}_t$. 
\end{definition}

\paragraph{Example.} Suppose that in a heuristic $\cH$, $\Psi_{t+1}$ takes the form $\Psi_{t+1}(x)=G(x)+F_t(x)$, where $G(x)$ is a function that depends on previous $F$'s and potentially on a regularizer\footnote{{\em Follow the leader}, and {\em follow the regularized leader} algorithms have this format.}. Then by optimality of $x_{t+1}$ we have 
$\Psi_{t+1}(x_{t+1})\ge \Psi_{t+1}(\ti{x}_{t+1})$, and thus $F_t(x_{t+1})+G(x_{t+1})\ge F_t(\ti{x}_{t+1})+G(\ti{x}_{t+1})$. Similarly,  $\ti{F}_t(x_{t+1})+G(x_{t+1})\le \ti{F}_t(\ti{x}_{t+1})+G(\ti{x}_{t+1})$. Thus, in this case, we get 
$$
	F_t(x_{t+1})-F_t(\ti{x}_{t+1}) \ge G(\ti{x}_{t+1})-G({x}_{t+1})\ge
\ti{F}_t(x_{t+1})-\ti{F}_t(\ti{x}_{t+1}),
$$
and the monotonicity property holds.

\medskip 

The \apex algorithm will use heuristic $\cH$ iteratively to find a solution sequence $X_0,\ldots,X_T$. The algorithm lets participants specify their objective $f_i$ (which remains fixed throughout the execution), and the intensity $\la_{i,t}$ of their preferences (which gets adjusted throughout the execution). The players get charged in tokens. Prices are calculated to be VCG prices. We will see in Lemma~\ref{lem:RegToEq1} that the dominant-strategy truthfulness of VCG implies that a low-regret execution of the \apex algorithm leads to a competitive equilibrium where truthful reporting of $f_i$ is an approximately dominant strategy for Player~$i$.

\begin{algorithm}
	\caption{\apex$(f_0,f_1,\ldots,f_n,\cH)$ algorithm given utilities $f_i:\cX\ra\RR$ by the participants, principal's utility $f_0:\cX\ra\RR$, local optimization heuristic $\cH$}
	 {\bf Main mechanism:}
	\begin{algorithmic}[1] 
		\State Fix starting point $X_0\in\cX$; 
		\For {$t=0..T-1$} 
		\State Collect bids $\la_{i,t}\ge 0$ from Player $i$;
		\State Selects a regularizer $R_t(x)$ that may depend on $t$, $X_0,\ldots,X_t$; 
		\State Set $F_t(x):=f_0(x)+\sum_{j=1..n} \la_{j,t} f_j(x)+R_t(x)$;
		\State Let $X_{t+1}$ be obtained by $\cH$ by maximizing $F_{t}(x)$
		starting at $X_t$;
	    \For{ each Player $i=1..n$}
	    \State Set  $F_t^{-i}(x):=f_0(x)+\sum_{j=1..n, j\neq i} \la_{j,t} f_j(x)+R_t(x)$
		\State Let $X_{t+1}^{-i}$ be obtained by $\cH$ by maximizing $F_{t}^{-i}(x)$ starting at $X_t$;
		\State Charge Player $i$, $C_{i,t}:=F_t^{-i}(X_{t+1}^{-i})-F_t^{-i}(X_{t+1})$
		tokens; 
		\EndFor
		\EndFor
			\end{algorithmic}
		{\bf Suggested algorithm for Player~$i$:}
		\begin{algorithmic}[1] 
			\State Report utility $f_i:\cX\ra\RR$ to the mechanism;
			\State Run a {\em Bandit with Knapsacks} online algorithm with initial budget $B_i$ to determine the $\{\la_{i,t}\}_{t=1}^T$; 
		\end{algorithmic}
	\label{alg:1}
\end{algorithm}

\paragraph{Players' actions.}
For now, we do not specify the algorithm the players will use to solve the {\em Bandits with Knapsacks (BwK)} set up by the main mechanism. BwK is a much more difficult problem than the ``usual" Bandits. Unlike the Bandits setting, a general $o(1)$ regret algorithm does not exist. 

At the same time, it is not difficult to come up with a reasonable heuristic for the BwK problem. The algorithm will try to learn the best marginal ``bang-per-buck'' ratio it can expect, and play accordingly. Fix a round $t$ and the bids of all other players. Then each value of the bid $\la_i$ induces a utility $u_{i,t}(\la_i)$ and a cost $C_{i,t}(\la_i)$. Under very mild local optimality conditions, $u_{i,t}$ and $C_{i,t}$ are monotonically non-decreasing in $\la_i$.

A ratio $r$ can be thought of as the exchange rate Player~$i$ is willing to pay in tokens per additional unit of utility\footnote{The discussion that follows can be easily restated in the language of constrained optimization. Given the objective of maximizing $\sum_t u_{i,t}(\la_i)$ subject to $\sum_t C_{i,t}(\la_i)\le B_i$. We can take a Lagrangian of the budget constraint with coefficient $r$ to get an upper bound on the possible utility. Under strong duality, we can attain the value OPT using this $r$. To keep the presentation more broadly accessible, we do the relevant calculations directly in this section.}.
 Let $\La$ be a distribution of strategy sequences $\{\la_{i,t}'\}$ that is feasible in expectation, that is:
\begin{equation}
	\label{eq:laprime}
\E_{\{\la_{i,t}'\}\sim\La} \left[	\sum_{t=1}^T C_{i,t}(\la_{i,t}')\right] \le B_i, 
\end{equation}
maximizing the payoff $\E_{\{\la_{i,t}'\}\sim\La} \left[	\sum_{t=1}^T u_{i,t}(\la_{i,t}')\right] $. 

If the inequality in \eqref{eq:laprime} is strict, that is, $\La$ does not spend its entire budget, then the budget constraint is irrelevant, and player~$i$ can attain maximal utility by bidding the same high value of $\la_i$ at every round.
Otherwise, the budget $B_i$ is a real constraint on player~$i$'s attainable utility. 
For simplicity, let us assume that $C_{i,t}(\la_i)$ is a continuous function\footnote{Otherwise, the same analysis still works, but we need to replace $\la_i$ with a distribution on a small interval $(\la_i-\delta,\la_i+\delta)$ to make the expected $C_{i,t}$ continuous in $\la_i$. This is essentially what happens in the proof of Theorem~\ref{thm:VCG-HZ} later in the paper.}. Since $C_{i,t}$ are monotonically non-decreasing in $\la_i$, there is a value $\la^i_{max}$  such that 
$$
\sum_{t=1}^T C_{i,t}(\la^i_{max}) = B_i. 
$$
Define 
$$
r_{min}:=\frac{1}{\la^i_{max}}. 
$$

\begin{claim}\label{cl:rmin1}
	Assuming local optimality of heuristic $\cH$, for each $t$ and for each $\la_i'$, we have
	$$u_{i,t}(\la_i')- r_{min} \cdot C_{i,t}(\la_i') \le 
	u_{i,t}(\la^i_{max})- r_{min} \cdot C_{i,t}(\la^i_{max}).$$
\end{claim}

\begin{proof}
	Let $X_{max}$ be the solution $\cH$ gives when we optimize $F_t^{-i}(x)+\la^i_{max}\cdot  f_i (x)$, and $X'$ be the solution $\cH$ gives when we optimize $F_t^{-i}(x)+\la_i'\cdot  f_i (x)$, and $X'$.
	We have 
	$$
	C_{i,t}(\la_i')-C_{i,t}(\la^i_{max}) = -F_t^{-i}(X') + F_t^{-i}(X_{max}), 
	$$
	since the $F_t^{-i}(X_{t+1}^{-i})$ term cancels out. 
	Therefore, 
	\begin{align*}
		u_{i,t}(\la^i_{max})- r_{min} \cdot C_{i,t}(\la^i_{max}) & =  r_{min} \cdot (\la^i_{max}\cdot 	u_{i,t}(\la^i_{max}) - C_{i,t}(\la^i_{max}) ) \\& =
		 r_{min} \cdot (\la^i_{max}\cdot 	f_i (X_{max}) - C_{i,t}(\la_i') +  F_t^{-i}(X_{max}) -  F_t^{-i}(X')) \\& \ge 
		 r_{min} \cdot (\la^i_{max}\cdot 	f_i (X') - C_{i,t}(\la_i') +  F_t^{-i}(X') -  F_t^{-i}(X')) \\&=u_{i,t}(\la_i')- r_{min} \cdot C_{i,t}(\la_i'),
	\end{align*}
where the inequality follows from $X_{max}$ being a local optimizer for 
$F_t^{-i}(x)+\la^i_{max}\cdot  f_i (x)$.
\end{proof}
%
%
%
Claim~\ref{cl:rmin1} implies that the simple strategy of playing $\la^i_{max}$ in every round matches or exceeds the performance of the optimal distributional strategy $\La$:
\begin{multline}\label{eq:laopt2}
\E_{\{\la_{i,t}'\}\sim\La} \left[	\sum_{t=1}^T u_{i,t}(\la_{i,t}')\right] = 
\E_{\{\la_{i,t}'\}\sim\La} \left[	\sum_{t=1}^T \left[u_{i,t}(\la_{i,t}')-r_{min}\cdot C_{i,t}(\la_{i,t}')\right]+
	r_{min}\cdot 	\sum_{t=1}^T C_{i,t}(\la_{i,t}')\right] \le \\
 	\sum_{t=1}^T \left[u_{i,t}(\la^i_{max})-r_{min}\cdot C_{i,t}(\la^i_{max})\right] + r_{min}\cdot B_i = 
			\sum_{t=1}^T u_{i,t}(\la^i_{max}).
 \end{multline}
Here, the first part of the inequality is by Claim~\ref{cl:rmin1}, and the second half is by the feasibility of $\la_{i,t}'$. 

Observe that the guarantee of \eqref{eq:laopt2} is a very powerful one: it doesn't just compete with the performance of the best fixed $\la_i$ in hindsight, but with respect to the best {\em sequence} of $\la_{i,t}$'s. One catch here (as in any discussion of competitive equilibria) is that we assume that the actions of other players are fixed and are not affected by the $\la_{i,t}$'s. This is acceptable given that our goal is indeed to obtain a competitive equilibrium. 

\subsubsection{Reading the output of the \apex mechanism}

Given an execution trace of Algorithm~\ref{alg:1}, there is a natural way to ``read off" the outcome and prices of the algorithm.

\paragraph{Outcome.} The (distributional) outcome is obtained by taking the time-average of the $X_t$'s:
\begin{equation}
	\barX := U(\{X_1,\ldots,X_T\}).
\end{equation}
Note that here $\barX$ is a uniform random variable\footnote{When the $X_t$'s are probability distributions themselves, this amounts to averaging them. In more general cases there might not be a generic way of mixing different $X_t$ beyond taking one of them at random.}, taking each of the $T$ values with probability $\frac{1}{T}$. 

\paragraph{Prices.} To obtain a competitive equilibrium, we expose Player~$i$ to a menu of possibilities. The menu will be based on the execution of Algorithm~\ref{alg:1}, and will be {\em separable by round}. In other words, the player will essentially be exposed to $T$ independent menus --- linked by a common budget, and by a common utility function $\ti{f}_i$. At round $t$, given a utility function $\ti{f}_i$ and a bid
$\ti{\la}_{i,t}$ the algorithm defines $$\ti{F}_t(x):=F_t^{-i}+\ti{\la}_{i,t}\ti{f}_i(x).,$$
which leads to 
the outcome $\ti{X}_{t+1}$ is obtained by maximizing $\ti{F}_{t}$ starting at $X_t$ using $\cH$. We use the values of $F_t^{-i}$ and $X_{t+1}^{-i}$ from the original execution of the algorithm. Player~$i$ is then charged
$$
\ti{C}_{i,t}(\ti{f}_i,\ti{\la}_{i,t}):=F_t^{-i}(X_{t+1}^{-i})-F_t^{-i}(\ti{X}_{t+1}).
$$

Player $i$ bids a utility function $\ti{f}_i$ and $\{\ti{\la}_{i,t}\}_{t=0..T-1}$. The sequence is required to be {\em feasible}, that is, 
$$
\sum_{t=0}^{T-1} \ti{C}_{i,t}(\ti{f}_i,\ti{\la}_{i,t}) \le B_i.
$$
If the sequence is feasible, then the outcome is just the uniform distribution 
 $$\barX^i:=U(\{\ti{X}_1,\ldots,\ti{X}_T\}).$$
Note that the sequence of bids with the truthful $f_i$,  $\ti{f}_i=f_i$ and $\{{\la}_{i,t}\}_{t=0..T-1}$ is feasible and leads to outcome $\barX$. 
 
 \subsubsection{From low-regret to an approximate correlated equilibrium.} As expected, our aim will be to link low-regret properties of the players' interaction with the algorithm to show that the outcome of the algorithm is a competitive equilibrium. 
 
 \begin{definition}
 	\label{def:regret1}
 	Consider a bandits-with-knapsacks game with a budget $B$, where the payoff of actions is in $[0,M]$. At each step an action $a_t$ leads to utility $u_t(a_t)$ and to cost $c_t(a_t)$. 
 	Let $\ba=(a_1,\ldots,a_T)$ be a sequence of actions satisfying the feasibility constraint $c(\ba):=\sum_{t=1}^T c_t(a_t)\le B$, that leads to utility 
 	$U(\ba):=\sum_{t=1}^T u_t(a_t)$. 
 	
 	We say that a sequence of actions has {\em strong regret} $\ve$, if for all possible distributions $\bA$ on sequences of actions $\ba'=(a'_1,\ldots,a'_T)$ satisfying 
 	$\E_{\ba'\sim\bA}\left[\sum_{t=1}^T c_t(a'_t)\right]\le B$, the resulting utility 
 	\begin{equation}
 		\E_{\ba'\sim\bA}[U(\ba')]:=	\E_{\ba'\sim\bA}\left[\sum_{t=1}^T u_t(a'_t)\right]\le U + \ve\cdot M. 
 	\end{equation}
 	\end{definition}

On the face of it, Definition~\ref{def:regret1} appears to be impossibly strong: we are considering regret with respect to {\em any} feasible strategy in hindsight. We even consider distributions over infeasible strategies as long as their average is feasible. However, in light of the discussion leading up to \eqref{eq:laopt2}, it is something that is potentially attainable in our context. We claim that strong regret bounds translate into approximate equilibria in the game induced by the \apex mechanism. This is a consequence of the truthfulness of the VCG mechanism.

\begin{lemma}
	\label{lem:RegToEq1}
	Suppose that $f_i(x)\in[0,1]$ for all $x\in\cX$, and that during the execution of Algorithm~\ref{alg:1}  with budget $B_i$ and a truthfully reported $f_i$,  Player $i$ has strong regret $\le\ve$. Suppose further that heuristic $\cH$ is locally correct.  Then reporting $f_i$ truthfully and playing $\{\la_{i,t}\}_{t=0..T-1}$ is an $(2\ve)$-dominant strategy for the menu of options available to Player~$i$ that is induced by the mechanism. 
\end{lemma}

The proof requires a few careful steps, but the basic intuition is that local correctness + the fact that the prices are VCG prices implies that there is no benefit (in tokens) in misrepresenting $f_i$ using some $\ti{f}_i$. Some extra effort is needed to see that the advantage of $f_i$ over $\ti{f}_i$ in tokens+utility can be converted into a pure advantage in utility. After we establish that bidding $f_i$ is near-dominant, the low regret property concludes the argument.

\begin{proof}
	Consider an alternative outcome based on the menu of options induced by the execution of Algorithm~\ref{alg:1}. In the alternative outcome, player $i$ reports type $\ti{f}_i$ and $\{\ti{\la}_{i,t}\}_{t=0..T-1}$ leading to outcome $\barX^i=U(\{\ti{X}_1,\ldots,\ti{X}_T\})$. 
	
	We first claim that there is no need for using $\ti{f}_i\neq f_i$. Fix a round $t$, let 
	$$
	C:=\ti{C}_{i,t}(\ti{f}_i,\ti{\la}_{i,t});~~U:=f_i(\ti{X}_{t+1})-f_i({X}_{t+1}^{-i}).
	$$
	That is, $C$ is the cost Player~$i$ pays in round $t$, and $U$ is the utility she derives compared to bidding $0$. We will show that, at least in expectation, up to an additive $\ve$, it is possible to attain the same (or better) cost/utility combination by bidding $f$ instead of $\ti{f}$. 
	
	Let 
	$$
	C(\la):=\ti{C}_{i,t}(f_i,\la);~~U(\la):=f_i({X}_{t+1}^\la)-f_i({X}_{t+1}^{-i}),
	$$
	where ${X}_{t+1}^\la$ is the outcome on bid $(f_i,\la)$. 
	That is, the cost and utility due to Player~$i$ when bidding the true $f_i$ and $\la_{i,t}=\la$. Clearly $C(0)=0$ and $U(0)=0$. 
	
	\paragraph{$C(\la)$ and $U(\la)$ are non-decreasing. }
	Suppose $0\le \la_1<\la_2$. Let $X_1:=X_{t+1}^{\la_1}$ and $X_2:=X_{t+1}^{\la_2}$. Then we have, by the local correctness of $\cH$, and thus by local optimality of $X_1$ and $X_2$,
	\begin{equation}
		\label{eq:mon1}
		F_t^{-i}(X_1)+\la_1 f_i(X_1) \ge F_t^{-i}(X_2)+\la_1 f_i(X_2),
	\end{equation}
and 
\begin{equation}
	\label{eq:mon2}
	F_t^{-i}(X_2)+\la_2 f_i(X_2) \ge F_t^{-i}(X_1)+\la_2 f_i(X_1).
\end{equation}
Adding $\eqref{eq:mon1}+\eqref{eq:mon2}$ and simplifying, we get 
$$
(\la_2-\la_1) f_i(X_2) \ge (\la_2-\la_1) f_i(X_1),  
$$
which implies $U(\la_2)\ge U(\la_1)$. 

Adding $\la_2\cdot \eqref{eq:mon1}+\la_1\cdot \eqref{eq:mon2}$ and simplifying, we get
$$
(\la_2-\la_1) F_t^{-i}(X_1) \ge (\la_2-\la_1) F_t^{-i}(X_2),  
$$
which implies that the VCG prices 
$$C(\la_2)=F_t^{-i}(X_{t+1}^{-i})-F_t^{-i}(X_2) \ge
F_t^{-i}(X_{t+1}^{-i})-F_t^{-i}(X_1) = C(\la_1).$$

\paragraph{Utility $U-\ve$ can be attained using $f_i$ at some cost.} Next, let $\de:=\ve/T$ and $\la:=C/\de$. We have 
$$
F_t^{-i}(X^\la_{t+1})+\la f_i(X^\la_{t+1}) \ge F_t^{-i}(\ti{X}_{t+1})+\la f_i(\ti{X}_{t+1});
$$
thus
$$
-C(\la)+ U(\la)\cdot C/\de \ge -C + U \cdot C/\de,
$$
and 
$$
U(\la)\ge U-\de. 
$$

\paragraph{Utility $\ge U-\ve$ can be attained using $f_i$ at cost $\le C$.}
Thus $U$ is a non-decreasing function with $U(0)=0$ and $U(\la)\ge U-\de$ for some $\la$\footnote{In many cases, there is in fact a $U$ such that $U(\la)\ge U$. In these cases we in fact lose no utility from reporting the true $f_i$.}. Therefore, there must exist a value $\la$ where the $U-\de$ threshold is crossed. Unfortunately, a point with $U(\la)=U-\ve$ may not exist. However, there is a value $\la$ such that $$U(\la^-):=\lim_{\eta\ra\la^-}
U(\eta) \le U-\de; \text{ and }
U(\la^+):=\lim_{\eta\ra\la^+}
U(\eta) \ge U-\de
$$
Let $\mu\in[0,1]$ be a parameter such that 
$$
\mu \cdot U(\la^-) + (1-\mu)\cdot U(\la^+) = U-\de. 
$$
Consider a mixed strategy that bids $\la^-$ (i.e. a value of $\la$ arbitrarily close to $\la$ from below) with probability $\mu$ and $\la^+$ with probability $(1-\mu)$. The expected utility of such a strategy is $U-\de$. It remains to calculate the expected cost, and to show that it is at most $C$. 

Let $X^-$ and $X^+$ be the outcomes of the bids $(f_i,\la^-)$ and $(f_i,\la^+)$, respectively. Then by optimality of $X^+$ and $X^-$ we have
\begin{equation}
	\label{eq:mon3}
	F_t^{-i}(X^-)+\la f_i(X^-) \ge F_t^{-i}(\ti{X}_{t+1})+\la f_i(\ti{X}_{t+1}),
\end{equation}
and 
\begin{equation}
	\label{eq:mon4}
	F_t^{-i}(X^+)+\la f_i(X^+) \ge F_t^{-i}(\ti{X}_{t+1})+\la f_i(\ti{X}_{t+1}).
\end{equation}
Taking the combination $\mu\cdot \eqref{eq:mon3} + (1-\mu)\cdot \eqref{eq:mon4}$, we get
\begin{multline*}
\mu \cdot 	F_t^{-i}(X^-) + (1-\mu)\cdot 	F_t^{-i}(X^+) + 
\mu\cdot \la f_i(X^-) + (1-\mu)\cdot \la f_i(X^+) \ge \\F_t^{-i}(\ti{X}_{t+1})+\la f_i(\ti{X}_{t+1}),
\end{multline*}
which implies 
$$
-\mu\cdot C(\la^-)-(1-\mu)\cdot C(\la^+) + U-\de \ge 
-C + U,
$$
and thus the expected cost satisfies 
$$
\mu\cdot C(\la^-)+(1-\mu)\cdot C(\la^+)\le C-\de. 
$$

\paragraph{Using strong regret to finish the argument.}
We have seen that it is possible to attain a total utility of at least 
$$
\sum_{t=0}^{T-1} f_i(\ti{X}_{t+1})-T\cdot \de = 
\sum_{t=0}^{T-1} f_i(\ti{X}_{t+1})-\ve
$$
using a mixed strategy over $\la_{i,t}$ that only uses the true utility function $f_i$. By the strong regret property, this mixed strategy attains utility within an additive $\ve$ of what Player~$i$ attains in the execution of Algorithm~\ref{alg:1}, leading to a total benefit of at most $2\ve$ from deviating. 
\end{proof}

\subsection{An infinitesimal version of the \apex algorithm}

Algorithm~\ref{alg:1} is written in the full generality of the VCG mechanism. As a result, individual prices need to be calculated by making $n$ calls to 
heuristic $\cH$ at every time step. In addition, while as we have seen in the analysis of the algorithm, it does induce a menu of (token) prices for each player at each step, these prices are difficult to interpret. 

In the special case where $\cX$ has a nice differentiable structure (for example, when $\cX$ is $\RR^k$ or the $k$-dimensional simplex $\Delta^k$ as in the voting example below), and $n$ is large, it is possible to use a quadratic approximation for the cost function to get a simplified version of Algorithm~\ref{alg:1}, with the added property that it produces a universal set of prices for effecting marginal change in the value of $X_{t+1}$. 

For simplicity, let us assume that $\cX$ is an open set. Alternatively, if $\cX$ has a boundary, we assume that the regularizer goes to $-\infty$ on the 
boundary $\partial \cX$, and thus $X_t$ is a point in the interior of $\cX$ for all $t$. In this case, assuming the objective functions and the regularizer are 
twice differentiable, we can write $F_t(x)$ around $X_{t+1}$ as 
\begin{equation}
	\label{eq:hes1}
	F_t(X_{t+1}+x)= F_t(X_{t+1}) - x^T H_t x + o(\|x\|^2). 
\end{equation}
Note that since $X_{t+1}$ is a local maximum of $F_t$, the linear term vanishes, and we may assume that $H_t \succcurlyeq 0$ is non-negative semi-definite. Assuming $H_t\succ 0$, and assuming the market is large\footnote{We don't need this assumption if $f_i$ is linear on $\cX$}, using approximation \eqref{eq:hes1} we can calculate approximate prices to charge Player~$i$ as follows. 

Write $$f_i(X_{t+1}+x)\approx f_i(X_{t+1})+ \nabla f_i(X_{t+1})^T \cdot x. $$
Then 
$$
F_t^{-i} (X_{t+1}+x)\approx F_t(X_{t+1}) - x^T H_t x - \la_{i,t} \nabla f_i(X_{t+1})^T \cdot x.
$$
Maximizing over $x$ gives
$$
X_{t+1}^{-i} \approx X_{t+1} - \la_{i,t} H_t^{-1} \nabla f_i(X_{t+1})/2, 
$$
and 
\begin{equation}
	\label{eq:Ci}
	C_{i,t}=F_t^{-i}(X_{t+1}^{-i}) - F_t^{-i}(X_{t+1}) \approx
	\frac{\la_{i,t}^2}{4}\cdot \nabla f_i(X_{t+1})^T H_t^{-1} \nabla f_i(X_{t+1}). 
\end{equation}
This leads to the specialized Algorithm~\ref{alg:2} below. Note that a very attractive feature of Algorithm~\ref{alg:2} is that we only need gradient access to $f_i$ in order to compute prices. Assuming $\cH$ is a heuristic based on gradient descent, one can expect to be able to run the entire algorithm with only gradient oracle access to the players' utilities. This is important both due to communication/privacy constraints and the fact that the players themselves may only have limited access to the $f_i$'s through a gradient (or even just a stochastic gradient) oracle.

\begin{algorithm}
	\caption{{\sf Infinitesimal-}\apex$(f_0,f_1,\ldots,f_n,\cH)$ algorithm given utilities $f_i:\cX\ra\RR$, principal's utility $f_0:\cX\ra\RR$, local optimization heuristic $\cH$}
	{\bf Main mechanism:}
	\begin{algorithmic}[1] 
		\State Fix starting point $X_0\in\cX$; 
		\For {$t=0..T-1$} 
		\State Collect bids $\la_{i,t}\ge 0$ from Player $i$;
		\State Selects a regularizer $R_t(x)$ that may depend on $t$, $X_0,\ldots,X_t$; 
		\State Set $F_t(x):=f_0(x)+\sum_{j=1..n} \la_{j,t} f_j(x)+R_t(x)$;
		\State Let $X_{t+1}$ be obtained by $\cH$ by maximizing $F_{t}(x)$
		starting at $X_t$;
		\State Write 
			$F_t(X_{t+1}+x)= F_t(X_{t+1}) - x^T H_t x + o(\|x\|^2)$;
		\For{ each Player $i=1..n$}
		\State Charge Player $i$, $C_{i,t}:=	\frac{\la_{i,t}^2}{4}\cdot \nabla f_i(X_{t+1})^T H_t^{-1} \nabla f_i(X_{t+1})$
		tokens; 
		\EndFor
		\EndFor
	\end{algorithmic}
	{\bf Suggested algorithm for Player~$i$:}
	\begin{algorithmic}[1] 
		\State Report utility $f_i:\cX\ra\RR$ to the mechanism;
		\State Run a {\em Bandit with Knapsacks} online algorithm with initial budget $B_i$ to determine the $\{\la_{i,t}\}_{t=1}^T$; 
	\end{algorithmic}
\label{alg:2}
\end{algorithm}

\paragraph{Relation to quadratic pricing.} We note that the competitive equilibrium induced by Algorithm~\ref{alg:2} exposes each player to quadratic prices over the space of outcomes. Quadratic pricing (and quadratic voting) has a rich history within the area of social choice --- suggesting another way in which such prices may occur ``naturally" as a result of repeated VCG-mediated interactions. We expect quadratic prices to occur whenever \eqref{eq:hes1} is an adequate approximation. Generally speaking, this should hold when $\dim(\cX)\ll n$. Therefore, quadratic pricing are natural to expect in voting and participatory budgeting, while we should expect other (potentially linear) prices to occur when $\dim(\cX)$ is high, such as in allocation of items or in bipartite matching.

%
%
%
%

\subsection{General analysis and open problems}

As we have seen in Lemma~\ref{lem:RegToEq1}, {\em if} the players attain low regret, the \apex Algorithm leads to a competitive equilibrium in which reporting $f_i$ truthfully is an $\ve$-dominant strategy. Assuming approximation \eqref{eq:hes1} holds, a similar statement can be made about Algorithm~\ref{alg:2}. The main question therefore is finding out whether/when players under the \apex Algorithm attain strong low regret, and --- if possible --- how one can compute the outcome of such convergence efficiently.
We should note that unlike some scenarios in algorithmic mechanism design, the algorithm's incentive properties hold assuming it has converged. Therefore, even without theoretical guarantees, a heuristic that almost always converges in practice will have the desired incentive properties. As we have seen in Lemma~\ref{lem:RegToEq1}, and will see in Section~\ref{sec:HZ} (Theorem~\ref{thm:reg}) again, results can be typically phrased as ``if the algorithm converges to a low-regret solution, then...". 

Beyond convergence of the algorithm --- or, rather, assuming it converges (either provably or in practice) --- we need to consider whether the outcome of the algorithm is ``good". In the setting without money it is impossible to define a common utility function and thus it is an interesting problem to even define efficiency (beyond Pareto efficiency) in these settings. Some of the questions that come up here are philosophical in nature (e.g. defining ``fairness" of a decision procedure --- most definitions are necessarily under-specified). 

Additional interesting questions arise when one tries to adapt the mechanisms to the setting {\em with} money. Generally speaking, mechanism design with money is easier than without money, since it is easier to state common objectives such as utility using the common currency. However, the introduction of money takes away one degree of freedom from the mechanism --- the exchange ratio between a player's utility and tokens, potentially making the problem more difficult. In addition, the direct link between payments within the mechanism and money opens the opportunity for collusion through outside transfers\footnote{Collusion is possible --- an indeed is sometimes unavoidable --- even in mechanisms without money, but the ability to measure collusion in money simplifies collusion between untrusting parties.}. 

In the direction opposite to mechanisms with money, the bandits with knapsacks setup actually allows one to use {\em multiple} non-exchangeable token currencies with which participants are endowed. Bandits with knapsacks with multiple currencies (multiple knapsack constraints in the BwK terminology) are considerably more complex to analyze. Therefore, it may be more difficult to get algorithms with multiple currencies to converge. At the same time, having multiple token currencies would allow to express more complex normative requirements from the resulting mechanism (e.g. ``equal treatment with respect to multiple non-substitutable categories of outcomes"). 

At the core of our reductions from algorithms (or heuristics) to mechanisms is the bandits with knapsacks setting. While it has receive substantial attention in the past decade (both directly, and indirectly -- e.g. in the context of online advertisement campaigns with budgets), it is still not nearly as well-understood as the general bandits setting. Further development of the theory of BwK --- particularly in terms of sufficient conditions for the existence of low-regret strategies --- would further our ability to develop new generic mechanisms. 

Finally, throughout the reduction we have treated the participants' utilities $f_i$ as fixed and known to the participants at the start of the algorithm. In practice, often these utilities themselves are being learned by the participants in a multi-round process. While the time $t$ in Algorithm~\ref{alg:1} is entirely fictitious --- representing epochs of an optimization procedure. However, it is not hard to adapt the algorithm into an online version where participants adjust their function $f_i$ over time, as new information arrives. As there is a tight link between optimization and online optimization, one can expect this link to extend to the reduction given by Algorithm~\ref{alg:1}.

Below we will address these points in greater detail, formulating specific problems and directions.

\subsubsection{Convergence analysis}\label{sec:conv}

The \apex Algorithm  provides a generic procedure for turning optimization heuristics into mechanisms. Unfortunately, at this level of generality, there is no hope of proving that the procedure ``works". Even defining what ``works" means is potentially challenging. 

We say that an execution of the \apex Algorithm is {\bf valid} if at the end of the execution all players have low regret with respect to the resulting outcome and prices. One actually has to be careful about defining what low regret here means. In Lemma~\ref{lem:RegToEq1} we took strong $\ve$-regret to mean that the absolute difference between the realized moves and the best moves in hindsight are small, one can also imagine scenarios where a relative measure of regret is more appropriate. Whichever notion is chosen, it makes sense to ask whether a valid execution exists, whether it is attained by a typical execution of the algorithm, and how robust it is. 

\begin{problem}\label{prob:1}
Let an execution be valid if players experience low strong regret.
Provide sufficient conditions on $\cX$, the $f_i$'s $R_t$'s, $\cH$, and the regret notion so that there exists a valid execution of the \apex Algorithm. 
\end{problem}

We expect a valid execution to exist under reasonably mild conditions --- ones that follow from generic fixed point theorems. For example, as we shall see in Section~\ref{sec:HZ}, Brouwer's Fixed Point Theorem is sufficient to prove that there always exists a competitive allocation of items under the Hylland-Zeckhauser scheme are supported by a valid execution of the \apex Algorithm. A more ambitious question is to find sufficient conditions for {\em all} executions to be valid. Note that for all executions to be valid we will need the players' BwK algorithms to be ``good" --- ones attaining low strong regret under reasonable conditions on the game the player is facing. We will leave questions of designing such ``good" BwK algorithms to
Section~\ref{sec:BwKProb}.

\begin{problem}
	\label{prob:2} 
	Let an execution be valid if players experience low strong regret.
	Provide an algorithm for the players and sufficient conditions on $\cX$, the $f_i$'s $R_t$'s, $\cH$ and the regret notion so that the execution of the \apex Algorithm  is valid with high probability. 
\end{problem}

Note that the \apex Algorithm has a parameter $T$ representing the number of rounds or epochs in the optimization. Therefore, in both Problems~\ref{prob:1} and \ref{prob:2} (as in later problems concerning the quality of the resulting solution), the answer may depend on $T$. Just as in optimization for empirical loss minimization of machine learning models, one can expect the quality of the solution to improve as $T\rightarrow \infty$\footnote{One can also envision a 
	version of Algorithm~\ref{alg:1} where the learning rate is lowered over time.}. It is therefore important to understand the dependence of the set of outcomes of valid executions on $T$. 

\begin{problem}
	\label{prob:3}
	Let $\cV_T=\cV_T(\{f_i\},R_t,\cH)\subset \De(\cX)$ be the set of possible outcomes $\barX$ of a valid execution of the \apex Algorithm. Under what conditions does the sequence $\{\cV_T\}$ converge to a set $\cV$ (in the earth-mover metric
	$W^1(\cX)$)? 
\end{problem} 

Building on the above, one can ask whether the resulting solution is essentially unique. 

\begin{problem}
	\label{prob:4}
Under what conditions is the resulting set $\cV$ in Problem~\ref{prob:3} a singleton $\cV=\{\nu\}$? How fast do $\{\cV_T\}$ converge to $\{\nu\}$ in this case?
\end{problem}

\subsubsection{Computational issues in reaching equilibrium}
\label{sec:probComp}

Since our end-goal is to be able to efficiently find the solution $\barX$, 
questions from Section~\ref{sec:conv} may and should be asked in the context of {\em computational efficiency}. One advantage of the approach based on an algorithm (as opposed to one based on an equilibrium definition) is that the \apex Algorithm  is itself a procedure for producing a solution $\barX$. As long as it converges to an $\ve$-equilibrium reasonably fast, say in $s(n,\ve)$ steps, we get an algorithm whose running time is dominated by $O(n\cdot s(n,\ve))$ applications of heuristic $\cH$.

Part of the setup's goal is to be able to treat $\cH$ as a black-box. This would allow us, for example, to deal with cases where the functions $f_i$ are not convex. When we treat $\cH$ as a black-box, our only recourse in terms of accelerating computation is to speed up convergence --- the number of steps $s(n, \ve)$ it takes to converge to an $\ve$-equilibrium.

\begin{problem}
	\label{prob:5}
What is the smallest number of iterations $s(n,\ve)$ does the \apex Algorithm need to converge to an $\ve$-equilibrium? Can the algorithm be tweaked to make this number instance-optimal?
\end{problem} 

One can hope that this number of steps can be reduced by changing the weights in the output $\barX$ to speed up convergence. As in many cases involving iterated minimization, it is likely that there are heuristics that converge much faster than the worst-case guaranteed convergence speed. 

\smallskip

The special case where the underlying problem is convex (and thus heuristic $\cH$ is not strictly necessary) is important in a number of potential applications, including the ones we'll see in Section~\ref{sec:app}. In this case, it is entirely plausible that Algorithm~\ref{alg:1} can be rewritten as a (larger) convex program, featuring variables $\la_i$, and potentially other auxiliary variables. This is indeed the case with correlated Nash equilibria, which can be attained via play among appropriate low-regret players, but can also be computed directly via a linear program\footnote{Which method is faster or better depends on the application domain. From the theoretic perspective what's important is that this equivalence exists.}. 

\begin{problem}
	\label{prob:6}
Suppose $f_0$ and the $f_j$'s are concave, and that $\cX$ is a convex set. Further suppose that $R_t=0$, and that $\cH$ is just the algorithm that finds the maximum of a function on $\cX$. When can an outcome $\cX$ be computed by a convex program, and what is the convex program computing it?
\end{problem} 

A likely prerequisite for an affirmative answer to Problem~\ref{prob:6} is that the set of possible $\ve$-regret outcomes is convex.

\subsubsection{Bandits with knapsacks}
\label{sec:BwKProb}

The technically least specified part of the \apex Algorithm has to do with the low-regret algorithm the players are supposed to run. While quite a bit of work has been done on bandits with knapsacks, there are many outstanding questions remaining. 

In its full generality, in the bandits with knapsacks setting, at time $t$ the agent can pull one of $k$ arms. After pulling arm $i_t$ at time $t$, in addition to the reward $r_{i_t, t}$, the player experiences a $d$-dimensional cost vector $c_{i_t,t}\in \RR_{\ge 0}^d$, corresponding to the cost of pulling the arm in terms of  $d$ constrained resources\footnote{These are the capacity-constrained ``knapsacks".}. The player is constrained by a budget vector $B\in \RR_{\ge 0}^d$  Once the sum of the costs in one of the constraints $\ell\in[d]$ is exceeded, that is:
$$
\sum_{\tau=1}^{T_0} c_{i_\tau,\tau,\ell} \ge B_\ell, 
$$
the player has to stop and can't collect further rewards. For all preceding discussions, we are only interested in the special case of $d=1$. The case $d>1$ is potentially interesting for some generalizations discussed is Section~\ref{sec:multi}, but for all standard applications $d=1$ is the case to consider. 

As we noted earlier, unlike the standard multi-arm bandits setting, in the BwK setting we cannot guarantee vanishing regret in hindsight. In standard bandit settings with bounded rewards, over $T$ rounds, one can hope to attain $O(T^{1/2})$ regret. In the case with knapsacks, there is no way to attain a $o(T)$ regret, and, in fact, there may be a {\em multiplicative} regret of as much as $\times\log T$ \cite{immorlica2019adversarial}.

 The big reason for BwK being more difficult, which we alluded to earlier, is that the optimal ``bang-per-buck" may change drastically over time. Consider a simple scenario where at each round there is a zero-arm with cost and reward zero\footnote{It is often assumed by default that such an arm --- the option of ``not playing" is available.}. The second arm costs $c_t=1$ to pull. The total budget is $B=T/2$. In rounds $t=1..T/2$ the reward $r_t=1$. There are two scenarios with respect to rewards in the second half: either the reward is  $r_t=2$ for all for $t=T/2+1..T$, or the reward is $r_t=0$ for all for $t=T/2+1..T$. The player needs to decide whether to exhaust its budget in the first half of the game, before learning whether this was the right decision. It is not hard to see that the best {\em additive} regret the player can attain is $T/2$, and the best {\em multiplicative ratio} attainable is $\frac{2}{3}$. Thus, even in this toy example, vanishing regret is impossible. 
Interestingly, this effect seems to persist even in the experts with knapsacks model, where the payoffs and costs of all arms is revealed. 

One can specialize the general BwK scenario to the following concave-reward game. At every round, the player is presented with a concave, non-decreasing cost-reward-function $R_t:c\mapsto r_t$, satisfying $R_t(0)=0$. The player chooses a cost $c_t$, subject to the global constraint $\sum_{t=1}^T c_t\le B$. 
The reward is calculated as 
$$
R(c) = \sum_{t=1}^{T} R_t(c_t). 
$$
Models of this kind have been considered in \cite{agrawal2019bandits}.

On the face of it, the concave-reward game is easier than the general BwK game. However, we believe that, in fact, it captures the difficult part of the BwK, and the gap between these two games in fact vanishes in the same way as the regret of bandits without knapsacks is vanishing. It would be interesting to formulate the exact sufficient conditions for this. 

\begin{problem}
	\label{prob:7}
	Under what conditions are the regrets of the following games with budgets the same up to an additive $o(T)$? How small is the gap between regrets?
	The scenarios are:
	\begin{enumerate}
		\item \label{BwK1} general BwK, with a menu of cost/rewards $(c_{i,t},r_{i,t})$, where cost/reward information is only revealed about the arm pulled; 
		\item BwK in the experts setting, where the cost/reward information is revealed about all arms; 
		\item BwK with stochastic closure of the arms: we are allowed to pull an arm $i$ with probability $p\in [0,1]$, and experience cost $p\cdot c_{i,t}$ and reward $p\cdot r_{i,t}$;
		\item the setting above in the experts regime:  where the cost/reward information is revealed about all arms; 
		\item \label{BwK5} the setting with stochastic closure, where the cost/reward information is revealed {\em before} the decision about $i_t$ is made. Note that the cost-reward function in this case is given by
		\begin{equation}
			R_t(c):=\max_{i} \min(r_{i,t} \cdot c/c_{i,t}, r_{i,t}).
		\end{equation}
	    $R_t$ is concave and non-decreasing --- corresponding to the cost-reward game scenario. 
	\end{enumerate}
\end{problem}

In particular, while nominally scenario \ref{BwK1} is much harder than scenario \ref{BwK5}, we believe that they are in fact equivalent under reasonable assumptions.

\subsubsection{Efficiency and fairness of the outcome}
\label{sec:234}

The overall goal of the framework we present is to attain ``good" solutions $\barX$ using a mechanism that leads players to reveal their utility functions $f_i$ truthfully. Since we chose to focus on mechanisms without money, actually defining efficiency appears to be non-trivial\footnote{In mechanisms {\em 
		with} money, one can define the utility of the outcome in units of the common currency, and compare this utility to the maximum attainable total utility.}.

\paragraph{Pareto efficiency.} One relatively weak benchmark is Pareto efficiency --- the resulting outcome $\barX$ cannot be replaced with an outcome $\barX'$ under which {\em all} players are at least as well-off as under $\barX$, and at least one player is strictly better off.
	
	We should note that unlike mechanisms with money, in the world without money Pareto optimality is a fairly weak condition. To illustrate, in the context of voting, all Pareto efficiency requires is that if {\em all} voters prefer option A over option B, then option B is never selected.
	
	 There are two main obstacles to our mechanism being Pareto efficient: (1) the heuristic $\cH$ may fail to optimize correctly (an algorithmic failure to locate a solution that is ``better for everyone" will necessarily map to a mechanism failure); and (2) whenever a regularizer is used, a (small) fraction of utility is sacrificed by adding a regularizer. 
	 Given these obstacles, it is possible for the outcome to not be entirely Pareto efficient. In the voting example, even if all voters prefer A over B, it is possible that the regularizer will allow for B to be selected with some (vanishing) probability.
	 
	  A natural approach would be relax the Pareto optimality condition, to allow for deviations that lead to vanishing improvements.
	  One natural definition of approximate Pareto efficiency is given in \cite{immorlica2017approximate}, saying that an outcome $\barX$ is $(1+\ve)$ Pareto efficient, if there is no alternative solution $\barX'$ where the utility of each player is increased by a factor $(1+\ve)$. A weaker definition would say that there is no $\barX'$ where no player is worse-off, and at least one player is better off by a factor $>(1+\ve)$. We believe that in most cases $\barX$ will satisfy at least approximate Pareto efficiency.
	
	\begin{problem}
		\label{prob:8} $ $
		\begin{enumerate}
			\item 
			Under what conditions do all solutions $\barX$ given by the \apex Algorithm satisfy Pareto efficiency?
			\item 
			What is the correct notion of approximate Pareto efficiency in this setting?
			Under what conditions do all solutions $\barX$ given by the \apex Algorithm satisfy approximate Pareto efficiency with approximation ratio $1+o_n(1)$? 
		\end{enumerate}
	\end{problem}

\paragraph{Efficiency beyond Pareto.} As noted above, Pareto efficiency appears to be a fairly weak efficiency guarantee. While one would be suspicious of a mechanism that fails to be Pareto efficient, there are Pareto efficient schemes that are clearly ``inefficient". 

Consider the example of $n$ voters choosing between two alternatives A and B. A mechanism that picks A and B with probability $\frac{1}{2}$ each unless there is unanimous support for one of the alternatives (in which case that alternative is picked), is Pareto efficient, even though intuitively it is inefficient to select B with probability $\frac{1}{2}$ if $n-1$ participants prefer A and only $1$ participant prefers B. 

On the other hand, this simple example already illustrates the difficulty in defining efficiency without money --- it fails to take into account intensities of preferences. If there are $n-1$ participants having a very weak preference for A and $1$ participant with a very strong preference for B, then perhaps choosing B with probability $\frac{1}{2}$ (or even with probability $1$) is the efficient outcome. It is hard to imagine a practically ``efficient" mechanism in which A will not be selected with an overwhelming probability. Thus the question is not just how to attain efficiency by a truthful mechanism, but how to define it properly. 

\begin{problem}
	\label{prob:9}
	Is there a generic definition of efficiency in mechanisms without money that extends beyond Pareto efficiency and that is consistent with truthful mechanisms?
\end{problem}

\paragraph{Fairness.} Once one moves beyond Pareto efficiency, a tension arises between fairness and efficiency. It is very challenging to define fairness in mechanisms without money. A minimum requirement akin to Pareto efficiency is equal treatment of equals: identical players should (at least ex-ante in the case of lotteries) experience identical outcomes. In allocation problems, this can be attained by a pseudomarket based on equal endowments such as the 
Hylland-Zeckhauser scheme \cite{budish2011combinatorial,he2018pseudo}\footnote{In the context of allocations using pseudomarkets, we also wish to have the property of {\em envy-freeness}: no player wishes the bundle of another player. Note, however, that the concept of envy-freeness does not make sense in scenarios such as voting or even two-sided matching.}. In social choice context, this can be attained by a symmetric social choice function. 

It would be appealing to have a definition of fairness that moves beyond `equal treatment of equals'. A natural definition of efficiency that is not attached to prior beliefs about values is `maximize sum-total welfare of participants'\footnote{As seen above, there are significant implementation barriers to realizing efficiency without money.}. What should a similar definition of fairness? Without any additional context, fairness will translate into equal treatment of participants --- of course, it is unclear what that would actually mean. 

A compelling extension of equal treatment of equals is equalizing the externalities participants exert on other participants: the amount of utility reduction they inflict on other players by participating. A recent detailed discussion of this extension in the context of algorithmic mechanisms without money (and additional references) can be found in \cite{immorlica2019equality}. 

To illustrate equalizing externalities, consider an example with two players Alice and Bob with utility functions $U_A:\cX\ra\RR^+$ and $U_B:\cX\ra\RR^+$. Let $o_A:=\max_x U_A(x)$ and $o_B:=\max_x U_B(x)$ be the maximum utilities attainable by the individual players. A solution $y_{opt}$ will be {\em efficient} if 
$$
U_A(y_{opt})+U_B(y_{opt}) = \max_{x} (U_A(x)+U_B(x)). 
$$
The externality Bob causes in solution $y$ is $Ext_B(y):=o_A-U_A(y)\ge 0$. The externality Alice causes os $Ext_A(y)=o_B-U_B(y)\ge 0$. If we are lucky, 
we will have $$Ext_B(y_{opt})=Ext_A(y_{opt}),$$ or, more broadly
\begin{equation}
	\label{eq:ext1}
\E_{y\sim\mu} Ext_B(y) = \E_{y\sim \mu} Ext_A(y), 
\end{equation}
where the distribution $\mu$ of outcomes is supported on points maximizing $U_A(y)+U_B(y)$:
\begin{equation}
	\label{eq:ext2}
  \forall y\in \text{supp}(\mu)~~ 
  U_A(y)+U_B(y) = \max_{x} (U_A(x)+U_B(x)). 
\end{equation}
Note that there is no reason to believe that \eqref{eq:ext1} and \eqref{eq:ext2} can be satisfied simultaneously --- most likely they cannot. One solution is to assign weights to players so as to make both conditions hold --- the weights correspond to a competitive equilibrium. In terms of good fairness properties for a mechanism to have, one can ask that it finds an externality equalizing distribution over optimal outcomes {\em whenever one exists}.

\begin{problem}
	\label{prob:10}
	Is there a generic definition of fairness in mechanisms without money that is consistent with truthful mechanisms?
\end{problem}

It is quite possible that there is no generic answer to Problem~\ref{prob:10}, and that the answer will depend on the precise setting. For example, in the case of voting, it makes sense to extend equal-treatment-of-equals to require that two ``diametrically opposite" voters (approximately) cancel out. On the other hand, in the case of allocation mechanisms ex-ante envy freeness  is a natural condition. 

\paragraph{Towards axiomatization?} The discussion of both efficiency beyond Pareto and fairness thus far focused on definitions as they pertain to the underlying optimization problem. The additional truthfulness constraints in the context of mechanism design will make attaining these properties even more difficult. On the other hand, the need for a truthful implementation might actually simplify the problem of reaching the ``right" definitions, by limiting the scope of what is possible. 

\begin{problem}
	\label{prob:11} 
Are there natural axiomatic properties pertaining to efficiency, fairness, and truthfulness, that together yield a set of mechanisms without money that can be presented in a general form, along the lines of the \apex Algorithm? \end{problem}

\subsubsection{Extensions to mechanisms with money and with multiple token currencies}
\label{sec:multi}

Many mechanisms without money over continuous domains use some kind of token pseudo-currency within their calculations. These tokens can be interpreted as representing a view on the relative importance of participants' preferences. For example, under most schemes, participants that are given equal token endowments will have an equal opportunity to affect the outcome of the mechanism. The \apex Algorithm, along with applications we will discuss in Section~\ref{sec:app} fall into the single-token category. 

\paragraph{Mechanisms with money.} It is natural to ask whether these mechanisms apply in settings with money. At a high level, money makes attaining efficiency easier, since it provides an absolute efficiency scale. At the same time, it may make truthfulness more difficult to attain, since one needs to not only consider deviations leading to a better outcome for player $p_i$, but also deviations leading to an identical outcome where $p_i$ has more money in the end. In addition, participation constraints which are not an issue in mechanisms without money may become an issue\footnote{Participation constraints assert that a player is not worse-off participating in a mechanism than not participating.}. A closely related issue --- which for example limits the utility of the VCG mechanism in the context of public projects --- is that the amount of revenue raised by VCG is highly unstable in the inputs.

An important example of a successful mechanism with money which combines elements of online learning and repeated auctions is the sponsored search ad placement mechanism \cite{lahaie2007sponsored}. In this setting a search engine such as Google needs to decide which ads to display along with its search results. The resulting mechanisms often feature an advertisement budget, which makes them share some features with the no-money setting (the problem becomes in part ``get the best set of ads displayed in exchange for budget $B$"). 

\begin{problem}
	\label{prob:12}
	$~$
	\begin{enumerate}
		\item
	To what extent can the framework of the \apex Algorithm  be adapted to a 
	setting with money, in particular with budget constraints? Can results such as the Fisher market be recovered?
     \item 
     Can the framework be extended to a hybrid setting with both tokens and money, to attain higher level of efficiency while maintaining a degree of fairness?
		\end{enumerate}
\end{problem}

Another question altogether is the best way of attaining truthfulness and efficiency with money, where the underlying preferences are very complex, and possibly implicit --- given only via a gradient oracle, or evolving over time. In practical terms, it might be best to keep the internal workings of the \apex Algorithm denominated in token units (and not in money), and wrap a money-for-token exchange around it. 

\paragraph{Multiple token currencies.} The bandits with knapsacks framework extends naturally to a setting with $d$ different types of constrained resources. This should allow our framework to extend seamlessly to a setting with multiple currencies. It remains to be seen whether there are natural scenarios where using multiple token currencies is preferred to using a single one. On the one hand, having multiple currencies might allow the designer to state multiple normative constraints of the form ``players are treated equally along multiple axes". On the other hand,  an effective ``exchange rate" may emerge between the currencies, nullifying its benefit. 

\begin{problem}
	\label{prob:13}
		$~$
	\begin{enumerate}
		\item
	Can the \apex Algorithm  be adapted to a setting with multiple token currencies? What properties hold in this case?
	\item 
	Are there settings where multiple token currencies attain an objective 
	not attainable using a single token currency?
\end{enumerate}
\end{problem}

An potential setting to investigate in this context is bipartite matching, with two different currencies used by the two sides of the match, as a way to ensure that both sides' preferences are given equal consideration. 

\subsubsection{Combining with online learning}

Much of algorithmic mechanism design presupposes that utility functions $f_i$ are known to the participants themselves, and that the main challenge is to elicit information about these $f_i$ to arrive at a socially desirable outcome. On the other hand, the key challenge in online learning (even with a single participant), is that the payoff function is unknown and needs to be discovered/maintained over time. In many practical scenarios with multiple participants features from {\em both} mechanism design and online learning are present. For example, advertisers buying impressions online are simultaneously (1) learning the value of these impressions (for example by observing the fraction of impressions that result in a sale); and (2) learning to interact with the mechanism selling ad impressions. 

Citing online advertising as an explicit motivation,  \cite{kandasamy2020mechanism} formalizes the problem of mechanism design where rewards need to be learned\footnote{See also earlier works, e.g. \cite{nazerzadeh2008dynamic,babaioff2013multi}.}.  For the setting with money, it gives a VCG-based mechanism that has both good asymptotic regret properties and is asymptotically truthful --- at least when deviations by a single player are considered. This immediately raises the question of whether one can produce a good mechanism {\em without money} for agents that are learning over time.

\begin{problem}
\label{prob:14}
Design mechanisms without money for a setting where players learn their type over time. 
\end{problem}

For best results, the mechanism would interpret ``learn" broadly in the following sense.
Traditionally, regret bounds are frames in max-min terms, against the worst possible environment, while in practice learning algorithms may perform much better than these guarantees. 
 Ideally, the performance of the mechanism should be comparable to the heuristic performance of the best learning algorithm in hindsight, and not to the max-min regret performance. 
 
A natural candidate to address Problem~\ref{prob:14} is an adaptation of Algorithm~\ref{alg:1}, where instead of $f_i$ the players submit function $f_{i,t}$ based on what they've learned about the environment up to that point. In the non-strategic settings, algorithms such as ``follow the regularized leader" are already framed in terms of optimizing an objective function that evolves based on past feedback. 

\begin{problem}
	\label{prob:15}
	Analyze the extension of Algorithm~\ref{alg:1} based on utility functions $f_{i,t}$ that evolve over time. 
\end{problem}

Note that as stated, ``time" in Algorithm~\ref{alg:1} corresponds to optimization epochs, therefore it is likely that the correct blending of the algorithm with online learning would involve updating the functions $f_{i,t}$ only every $T$ rounds --- interlacing $T$ rounds of optimization with a single round of performing an action, observing the outcome, and updating utility functions based on these observations. 

\section{Applications}
\label{sec:app}

In this section we present a preliminary discussion on applications to main domains where mechanisms without money are used.

As we will discuss, in many cases there are inherent incentive issues, such as collusion, that are beyond the reach of any mechanism. On the other hand, our framework is sufficiently flexible to fit most optimization algorithms, and in many cases it is first-order approximately individually truthful, which means that we can hope to have first-order approximate efficiency and (competitive-equilibrium) truthfulness even in cases where known negative results rule out efficient truthful mechanisms. 

In other words, one can decompose the problem of coordination via a mechanism into 
the following three components: (1) algorithmic: figuring out individual utility functions, and solving the aggregate optimization problem; (2) individual incentives: incentivizing participants to reveal their preferences truthfully; (3) policy: preventing mechanism failure through actions outside the mechanism (such as collusion). Algorithmic mechanism design deals primarily with (2). The best one can hope for is to attain (2) without putting constraints on (1), and without making (3) worse than necessary. 

We will discuss three main applications: voting, one-sided allocation, and two-sided allocation. In the case of one-sided allocation, we will show a new  connection to existing pseudo-market mechanisms. In the other two cases, we will give a general high-level discussion, leaving results to subsequent works. 

\subsection{Voting with cardinal preferences}\label{sec:voting}

We consider the problem of aggregating cardinal preferences of $n$ over a discrete set of possibilities $[k]$ with $k\ge 2$. ``Cardinal" (as opposed to ordinal) means that each player $i\in[n]$ has a utility vector $u_i \in \RR^k$, 
where $u_{ij}$ represents how happy player $i$ would be with outcome $j$. Since the aggregation mechanism doesn't use money, the output should be the same whether player $i$ reports $u_i$ or $2 u_i$, which means that $u_i$ should be treated as normalized direction vectors. 

\paragraph{Impossibility: strategy-proofness and efficiency} Generally speaking, the only case in which truthful, symmetric (or even just non-dictatorial), and Pareto efficient voting is possible is when $k=2$. Whenever there are more than two possibilities to choose from, there will be some opportunity for strategic voting. This is true in the ordinal case \cite{gibbard1973manipulation,satterthwaite1975strategy}, and in the case with cardinal voting \cite{gibbard1978straightforwardness,hylland1980strategy}. 

A dictatorial scheme is truthful and Pareto efficient; it can be made symmetric by turning it into a {\em randomized dictatorship} scheme, where an index $i\in [n]$ is selected at random, and then player $i$ picks her favorite alternative. 
Note that even in the case with two alternatives, randomized dictatorship is not very efficient --- if $90\%$ of the voters prefer alternative A, and $10\%$ prefer alternative B, the disfavored alternative will be chosen $10\%$ of the time. In addition, randomized dictatorship discards all quantitative information about the preferences. For example, suppose are three alternatives A,B, and C. Half the voters have preference $A\gtrapprox C \gg B$ (that is, slightly prefer A over C, and strongly disfavor B), and half the voters have preference $B\gtrapprox C\gg A$. In such a scenario, the clearly best alternative is C, but a randomized dictatorship will select A and B with equal probability, never selecting C. 

As noted in Section~\ref{sec:234}, even in the case of two alternatives, efficiency is somewhat elusive due to normalization. For $k=2$, and the standard majority rule, the voting rule does not pick an alternative $j$ maximizing 
$$
U(j):=\sum_{i\in[n]} u_{ij}.
$$
Rather, if we denote $\la_i:=\frac{1}{|u_{i1}-u_{i2}|}$, the majority rule maximized 
\begin{equation}
	\label{eq:tildeU}
\tilde{U}(j):= \sum_{i\in [n]} \la_i u_{ij}. 
\end{equation}
In other words, each voter is scaled so that the difference between their more preferred alternative and less preferred alternative is $1$. In such a scheme, voters who have stronger preferences are scaled down, and voters who have weaker preferences are scaled up. Absent money (or some other persistent value-tracking mechanism), such scaling is unavoidable, since there is no cost for player $i$ to report $2 u_i$ instead of $u_i$, and thus such report shouldn't increase the player's influence. 

One could hope to define efficiency in terms of the sum of universally normalized  utilities, maximizing $\sum_{i\in[n]} \hat{u}_{ij}$, where $\hat{u}$ is the unit vector in the direction of $u$ according to some norm. This is indeed the form of \eqref{eq:tildeU} for the two-alternative majority rule. With more than two alternatives, truthfulness implies that the norm in the scaling will have to depend on the outcome being considered. Consider an example where there are $k=3$ alternatives $(A,B,C)$ and there is an approximately equal number of voters with utility vectors given by 
$u_1 = (12,11,-23)$, $u_2=(11,12,-23)$, $u_3 = (2,1,-3)$, $u_4=(1,2,-3)$.
The preferences $u_1$ and $u_3$ are identical with respect to alternatives $A$ and $B$, but $u_1$ has a much stronger preferences against $C$\footnote{Note that the utilities are given up to scaling and shifting. If we add $10$ to all values in $u_3$, we will get $(12,11,7)$, demonstrating that $u_1$ indeed dislikes $C$ much more than $u_3$.}. Note that all players dislike alternative $C$, and thus the choice will be between alternatives $A$ and $B$. In this example, we then should expect $\hat{u}_1\approx\hat{u}_3$, but this means that the norm with respect to which normalization will happen will have to give very little weight to the $C$ component. Otherwise, players with type $u_1$ will be incentivized to misreport their type as $u_3$ --- this is how strategic voting typically happens in practice: if an alternative is ``not realistic" voters will try to reallocate their influence to alternatives among which actual choice is happening.

Therefore, in defining efficiency, the normalization factors in \eqref{eq:tildeU} will not only need to depend on the $u_{ij}$'s, but also on the alternatives being considered. The key challenge, of course, is the circularity of such scaling: the outcomes considered depend on the scaling factors, while the scaling factors depend on the outcomes being considered. 

\paragraph{Collusion-proofness.} Typically, truthfulness, or strategy-proofness is concerned with deviations by a single player. Even the strongest notion of truthfulness --- dominant strategy truthfulness --- only requires that a single player cannot improve her outcome by misreporting her type. A truthful mechanism may still be susceptible to {\em collusion}, where a number of players misreport their types to improve their outcomes. In some settings (such as one-sided allocation) it is possible to resist collusion, at least when transfers between players are not allowed. Unfortunately, it appears that in the context of voting, it is impossible to avoid collusion. Continuing the three-alternative example, two players with types $u_1=(1,1,-2)$ and $u_2=(1,-2,1)$ are in perfect agreement about preferring alternative $A$, but work against each other regarding alternatives  $B$ and $C$. They can form a coalition around promoting alternative A, for example by reporting their type as $\tilde{u}=(2,-1,-1)$. Under most voting schemes (including schemes based on normalizing votes), this will increase the collective impact of the two players. In the context of politics, such collusion corresponds to forming a political party.

\paragraph{Quadratic voting.} A natural concept for cardinal voting that has gained some popularity in recent years is {\em quadratic voting}. Under quadratic voting, a voter is given a budget of $1$ token, which she can allocate among the alternatives \cite{lalley2018quadratic}. Giving $v_j$ votes to alternative $j$ costs $v_j^2$ tokens. Suppose the voter has utility $u_j \cdot v_j$ for giving $v_j$ tokens to alternative $j$, and suppose further that $u_j\ge 0$. Then the unit-cost allocation maximizing total utility is given by 
$$
v_j = \frac{u_j}{\sqrt{\sum_i u_i^2}}. 
$$
Thus, the optimal vote is indeed the true type $\hat{u}$ normalized to unit euclidean length. This scheme can work in the context of participatory budgeting, but is not portable ``as-is" to the social choice context. Even if the output is a lottery where alternative $j\in[k]$ is selected with probability $p_j$, the constraints $p_j\ge 0$, $\sum_{j\in[k]} p_j=1$ would make an non-distorted quadratic voting scheme impossible. On the other hand, the equilibria that naturally occur in Algorithm~\ref{alg:1} (and more explicitly in Algorithm~\ref{alg:2}), lead to essentially a ``quadratic-form" voting scheme, where the cost of the vote in direction $v$ is $v^T A v$ for some PSD $A\succcurlyeq 0$, rather than just $v^T v = \sum v_i^2$. 

\paragraph{Specific problems.} The specific problems can be broken down into two parts corresponding to ``theory building" and ``algorithm design". On the algorithm design side, the main problem is to design new voting mechanisms with cardinal utilities based on Algorithm~\ref{alg:1}. These are essentially Problems~\ref{prob:1}--\ref{prob:6} specialized to the voting scenario. 

\begin{problem}
	\label{prob:16}
	\begin{enumerate}
		\item 
		Adapt Algorithms~\ref{alg:1} and \ref{alg:2} to the setting where $\cX=\De_k$ is the set of probability distributions over $[k]$, and utilities $f_i(p):=\sum_{j} u_{ij} p_j$ are linear.
		\item 
		Under what conditions is the output of such an algorithm unique? How hard is it to compute both in theory and in practice?
		\item 
		What kind of competitive equilibrium does it induce? 
		\item 
		What are the competitive-equilibrium truthfulness guarantees, and what is the efficiency-truthfulness trade-off?
	\end{enumerate}
\end{problem}

Giving satisfactory answers to Problem~\ref{prob:16} will yield a new practical family of preference aggregation algorithms. There are some secondary benefits to being approximately strategy-proof, such as allowing for asynchronous voting (since knowing how other participants voted does not have much impact on one's best response). 

In terms of theory-building, there are two main outstanding questions.

\paragraph{Efficiency-truthfulness trade-offs.} The first theory-building question is about mapping out the efficiency-truthfulness frontier. 

\begin{problem}
	\label{prob:17}
	For the $n$-voter, $k$-alternative voting problem with (normalized) cardinal utilities, what is the fundamental trade-off between approximate truthfulness and approximate efficiency?
\end{problem}

Known negative results show that (exact) truthfulness is incompatible even with fairly weak notion of efficiency. Note that one needs to be careful with the definition of ``approximate truthfulness": it is not hard to create a voting scheme that is efficient and $\ve$-truthful with $\ve=o_n(1)$, in the sense that the expected benefit from misreporting one's preferences is bounded by $\ve$. The problem is that in such mechanisms the benefit of voting would also be $O(\ve)$. A proper definition of approximate truthfulness would say that the benefit from misrepresenting one's vote should either be small relative to the benefit of voting at all, or tiny in absolute terms.

A possible definition of an $(\ve,\de)$-truthful voting scheme $\cM$ is that for all $u_{-i}, u_i, u_i'$, 
\begin{equation}
	\label{eq:voteTruth} \underbrace{ 
	u_i^T \cdot (\cM(u_{-i},u_i')-\cM(u_{-i},u_i))}_{\text{benefit from misreporting} 
}\le 
	\de\cdot \|u_i\| + \ve\cdot  \underbrace{ 	u_i^T \cdot (\cM(u_{-i},u_i)-\cM(u_{-i},0))}_{\text{benefit from voting}}.
\end{equation}
Here $\de$ should be very small ($o(n^{-1/2})$, and ideally $O(n^{-1})$ or even $0$), and $\ve$ should be $O(1)$, and ideally $o(1)$. 

\begin{problem}
	\label{prob:18}
	For what values of $(\ve,\de)$ is it possible to attain an $(\ve,\de)$-truthful voting mechanism with vanishing efficiency loss?
\end{problem}

\paragraph{Good properties beyond symmetry?} The second theory-building question is defining ``good" properties one should require of a quantitative voting scheme, and obtaining relationships between these properties. The biggest question is how to define fairness beyond requiring that $\cM$ is symmetric in the votes. One possible extension is that if two players have diametrically opposing views --- that is $u_i+u_j = 0$, then removing them should only change the outcome distribution by a negligible amount. This can be extended to a small set $S$ of voters with $\sum_{i\in S} u_i=0$. We should not expect such a property to hold exactly, since one would expect that adding a pair of voters that is indifferent in aggregate would slightly move the outcome towards the uniform distribution\footnote{Note that a $(1\,000\,100,1\,000\,000)$ vote is much closer than a $(100,0)$ vote.}. 

\subsection{One-sided allocation}
\label{sec:HZ}

In the one-sided allocation setting there are $n$ players and $n$ goods. We will focus on the simplest case, in which each player wishes to obtain exactly one good, and the goal of the mechanism is to produce a matching $\pi:[n]\ra[n]$. Each player has a vector of utilities $u_i$, where $u_{ij}\in[0,1]$ is the utility experienced by player $i$ from obtaining item $j$. Applications of this setting include allocation of scarce resources where money cannot be used, such as school choice and course assignment. Since transfers cannot be used, the solution concept typically involves a lottery, where the outcome is given by a bi-stochastic matrix $X=(x_{ij})_{i,j=1..n}$, with $x_{ij}$ representing the probability that player $i$ receives item $j$. By  Birkhoff–von Neumann  theorem, $X$ can be implemented as a lottery over assignments. 

An important solution concept in this setting was given by Hylland and Zeckhauser
\cite{hylland1979efficient}. The solution fits within the broader {\em 
	competitive equilibrium from equal incomes (CEEI)} framework. In the HZ scheme, each player is given $1$ unit of token endowment. Each item is given a price $C_j$, and each player $i$ is given a bundle $x_{ij}$ with $\sum_j x_{ij}=1$, $x_{ij}\ge 0 $. The outcome is a {\em  competitive equilibrium} if 
\begin{enumerate}
	\item all items get allocated: $\sum_i x_{ij}=1$ for all $j$;  \item each player $i$ stays within her budget:  $\sum_{j} C_j\cdot x_{ij}\le 1$; and 
	\item 
	each player receives her favorite bundle among the ones she can afford: for each $i$, and for each $y_{ij}\ge 0$ with $\sum_j y_{ij}=1$ and 
	$\sum_{j} C_j\cdot y_{ij}\le 1$,
	\begin{equation}
	\sum_{j} u_{ij}\cdot y_{ij}\le 	\sum_{j} u_{ij}\cdot x_{ij}.
	\label{eq:HZeq}
\end{equation}
	\end{enumerate}
Existence of a price vector $C$ inducing a CE follows from general fixed-point results. The price vector needs not be unique. It is still unknown whether such prices can be computed efficiently in general\footnote{Moreover, an approximate competitive equilibrium may be easier to attain than an exact one. See \cite{vazirani2020computational} for a recent
discussion on computational complexity questions.}. The mechanism induced by a HZ scheme needs not be truthful, although in large markets truthfulness does emerge \cite{budish2011combinatorial}. 

The output of Algorithm~\ref{alg:1} when players have vanishing strong regret is a competitive equilibrium allocating items using a token system. It is therefore natural to ask whether there is a correspondence between CEs induced by Algorithm~\ref{alg:1} and HZ equilibria. 
To be specific, we will distinguish two versions of Algorithm~\ref{alg:1}. The {\em non-regularized} version just runs a unit-demand VCG at every step. The {\em regularized} version adds a concave regularizer to the process. 

\subsubsection{Not all HZ equilibria correspond to VCG-competitive equilibria}

A non-regularized version of Algorithm~\ref{alg:1} is just a repeated run of unit-demand VCG auction using tokens, where the bid of player $i$ at time $t$ takes the form $\la_{i,t}\cdot u_i$. In a competitive equilibrium, the sum of these runs would exhaust the token endowment of all players, except those who always get their favorite item. We start mapping out the relationship between these equilibria and HZ equilibria by showing that there exist HZ equilibria that do not correspond to a combination of VCGs. 

Consider the following setting with $4$ players and $4$ items. 

$$
\begin{array}{c|cccc}
	 & A & B & C & D  \\ \hline
	 u_1 & 11 & 9 & 14 & 0 \\
	 u_2 & 11 & 9 & 14 & 0 \\
	 u_3 & 0 & 0 & 10 & 0 \\
	 u_4 & 0 & 0 & 10 & 0 \\
\end{array}
$$

The following prices $P_1$ and allocation $x$ form a HZ equilibrium:

$$
\begin{array}{c|cccc}
	& A & B & C & D  \\ \hline
		\mathbf{P_1} & \mathbf{ 1.1} & \mathbf{0.9 }&\mathbf{ 2 }&\mathbf{ 0}\\ \hline
	x_1 & 0.5 & 0.5 & 0 & 0 \\
	x_2 & 0.5 & 0.5 & 0 & 0 \\
	x_3 & 0 & 0 & 0.5 & 0.5 \\
	x_4 & 0 & 0 & 0.5 & 0.5 \\

\end{array}
$$

We claim that there is no distribution on tuples of the form $(\la_1,\la_2, \la_3,\la_4)$, such that the allocation $x$ is the result of running VCG on utilities $(\la_1 u_1, \la_2 u_2, \la_3 u_3, \la_4 u_4)$, and the payments due from each player average out to $1$. To see this, let $(\la_1,\la_2, \la_3,\la_4)$ be a tuple in the support. Without loss of generality, suppose the resulting allocation is $(A\ra 1; B\ra 2; C\ra 3; D\ra 4)$ (the argument in symmetric for the other three possible allocations). The following conditions hold for the $\la$'s by the optimality of the allocation:
$$
\left\{
\begin{array}{rl}
\la_1 &\ge \la_2\\
\la_3 &\ge \la_4 \\
11 \la_1 + 9\la_2 + 10 \la_3 &\ge 14 \la_1 + 11\la_2
\end{array}
\right.
$$
Next, let us compute the externalities. The price accruing to player $1$ is $C_1 = 2\la_2$. The price accruing to players $2$ and $4$ is $C_2=C_4=0$. The price accruing
to player $3$ is $C_3(14 \la_1 + 11\la_2)-(11\la_1 + 9 \la_2 ) = 3 \la_1 + 2\la_2$. 

We see that $C_3+C_4\ge C_1+C_2$, with equality only when $C_1=C_2=0$. Therefore, there cannot be a distribution over $\la$'s where $C_1+C_2$ and $C_3+C_4$ both average out to $2$. 

There is a different HZ competitive equilibrium allocation $y$, given below, supported by prices $P_2$ that do come from a distribution of VCG allocations.

$$
\begin{array}{c|cccc}
	& A & B & C & D  \\ \hline
	\mathbf{P_2} & \mathbf{ 8/7} & \mathbf{0 }&\mathbf{ 20/7 }&\mathbf{ 0}\\ \hline
	y_1 & 1/2 & 7/20 & 3/20 & 0 \\
	y_2 &1/2 & 7/20 & 3/20 & 0 \\
	y_3 & 0 & 3/20 & 7/20 & 1/2 \\
	y_4 & 0 & 3/20 & 7/20  & 1/2 \\
	
\end{array}
$$

Consider the weights $\la_1=\la_2=4/7$, $\la_3=\la_4 = 2/7$, resulting in scaled utilities:

$$
\begin{array}{c|cccc}
	& A & B & C & D  \\ \hline
	\la_1 u_1 & 44/7 & 36/7 & 56/7 & 0 \\
	\la_2 u_2 &  44/7 & 36/7 & 56/7 & 0\\
	\la_3 u_3 & 0 & 0 & 20/7 & 0 \\
	\la_4 u_4 & 0 & 0 & 20/7 & 0 \\
\end{array}
$$

Allocation $y$ can be represented as a combination of $4$ permutations, each with total utility $100/7$, and VCG payments given by the following table:

$$
\begin{array}{c|cccc}
\text{weight\textbackslash player}	& 1 & 2 & 3 & 4  \\ \hline
7/20 & A & B & C & D \\
\text{VCG payment} & 8/7 & 0 & 20/7 & 0\\
7/20 & B & A & D & C \\
\text{VCG payment} & 0 & 8/7 & 0 & 20/7\\
3/20 & C & A & B & D \\
\text{VCG payment} & 20/7 & 8/7 & 0 & 0\\
3/20 & A & C & D & B \\
\text{VCG payment} & 8/7 & 20/7 & 0 & 0
\end{array}
$$

\subsubsection{All preference profiles admit a VCG competitive equilibrium}
\label{sec:VCGforHZ}

Let $U=\{u_{ij}\}_{i,j=1..n}$ be a matrix utilities with $u_{ij}\ge 0$. Our goal will be to prove the following theorem, which asserts that it is possible to obtain a competitive HZ equilibrium for the allocation problem with utilities $U$, where the prices are VCG prices supported by utilities of the form $\la_i u_i$ for $\la_i\ge 0$. This gives a more refined version of the main result in \cite{hylland1979efficient}. To make extensions and generalizations easier we only use Brouwer's fixed-point theorem (and not the more general Kakutani's theorem as in the original proof). 

Note that while Theorem~\ref{thm:VCG-HZ} was found using the \apex framework, it is proven directly without relying on any convergence assumptions. Later, in Theorem~\ref{thm:reg} we will show that a convergent low-regret execution of \apex on a (regularized) allocation optimizer gives a constructive way of finding approximate HZ prices\footnote{This does not quite resolve the problem
of computing a HZ competitive equilibrium efficiently, because there are no general low-regret algorithms for BwK. It remains to be seen whether low-regret algorithms with good convergence properties can be found for the specific setting corresponding to one-sided allocation.}.

\begin{thm}
	\label{thm:VCG-HZ}
Let $U=\{u_{ij}\}_{i,j=1..n}$ be a matrix utilities with $u_{ij}\ge 0$. Then there exist numbers $\la_i\ge 0$, prices $C=\{C_j\}$ and an allocation $X=\{x_{ij}\}$ with the following properties. 
\begin{enumerate}
	\item 
	$X$ is a valid allocation: $\forall j: \sum_i x_{ij}=1$ and
	$\forall i:\sum_j x_{ij}=1$; 
	\item
	$C_j$ are the VCG prices for utilities given by $u'_{ij} = \la_i u_{ij}$;\footnote{Recall that for VCG for unit-demand allocation, the payment accruing to player $i$ depends only on the item $j$ she receives, and that if there are multiple optimal solutions, in all of them, the same item $j$ will be sold for the same price $C_j$.}
	\item
	$X$ is a combination of optimal allocations under $u'$: for every $\pi:[n]\ra[n]$ with $\forall i~x_{i\pi(i)}>0$ we have 
	$$
	\sum_i u'_{i \pi(i)} = \max_{\si} 	\sum_i u'_{i \si(i)}. 
	$$
	\item 
	The players can purchase their allocations with budget not exceeding $1$. For each player $i$, 
	$$
	\sum_j C_j x_{ij} \le 1. 
	$$
	\item 
	Prices $C_j$ and allocation $X$ form a HZ equilibrium. That is, for every player $i$
	$$
	\sum_{j} u'_{ij} x_{ij} = \max_{\displaystyle{y:} \begin{array}{c}\sum_j C_j y_j \le 1 \\ \sum_j y_j =1\end{array}} \sum_j u'_{ij} y_j. 
	$$
\end{enumerate}
\end{thm}

\begin{proof}
	Without loss of generality we can scale the problem so that $u_{ij}\in [0,1]$. 
	Fix a parameter $\ve>0$ (we will later take $\ve\ra 0$). For a vector of $\la$ with $\la_i\ge 0$, define the following function $\Phi_\ve(\la)$:
	\begin{multline*}
	\Phi_\ve(\la)_i:=\\ \E_{(\la'_1,\ldots,\la'_n)\sim U_{[\la_1,\la_1+\ve]\times [\la_2, \la_2+\ve]\times\ldots\times[\la_n,\la_n+\ve]}} [\text{Payment due from player $i$ in $VCG(\la'_1 u_1,\ldots,\la'_n u_n)$}].
\end{multline*}
Here $VCG(v_1,\ldots,v_n)$ denotes the unit-demand VCG mechanism with given valuations. Denote the following adjustment mapping $\Psi_\ve$ from the space of $\la$'s to itself:
$$
\Psi_\ve(\la)_i := \min\left(\max(0, \la_i + (1-\Phi_\ve(\la)_i)),\lamax\right), 
$$
where
$$
\lamax:= 1+  \frac{n}{\min_{i,j_1,j_1:u_{i j_1}\neq u_{i j_2}} |u_{i j_1}-u_{i j_2}|}.
$$
In other words, we adjust $\la_i$ by adding $(1-\Phi_\ve(\la)_i)$ to it, so that if player $i$ pays more than $1$, $\la_i$ gets decreased, and if she pays less than $1$, $\la_i$ gets increased. We then snap it to the interval $[0,\lamax]$ if the adjustment causes $\la_i$ to escape this interval, where $\lamax$ is chosen to be sufficiently large. 

Consider a $\la$ in the closed, convex set $M := [0,\lamax]^n$. On this set $\Psi_\ve(\la)_i$ is bounded by $\lamax$. 
When we vary $\la_i$ by  $\de\ll \ve$, the distribution under the expectation in  $\Phi_\ve(\la)_j$  only varies by $\frac{\de}{\ve}$ in statistical distance, and thus $\Phi_\ve(\la)_j$ changes by at most $\frac{\de\cdot \lamax}{\ve}$, and 
$\Psi_\ve(\la)$ also changes by at most $\frac{\de\cdot \lamax}{\ve}$ in each coordinate. Therefore, $\Psi_\ve(\la)$ is a continuous mapping from $M$ to itself, and by the Brouwer fixed-point theorem
 it admits a fixed point $\la^\ve$ such that 
$$
\Psi_\ve(\la^\ve) = \la^\ve. 
$$

Each $\la^\ve$ induces an allocation $X^\ve$ and prices $C^\ve$ on items given by considering the expected allocation and expected VCG prices for $\la$ sampled uniformly from $[\la^\ve_1, \la^\ve_1+\ve]\times [\la^\ve_2, \la^\ve_2+\ve]\times \ldots \times [\la^\ve_n, \la^\ve_n+\ve]$.
The allocations $X^\ve$ belong to the compact region of bi-stochastic matrices in  $[0,1]^{n\times n}$ by definition. 

The prices $C^\ve_j$ are non-negative. We claim that they are also uniformly bounded. 
Note that 
$$
\sum_j C^\ve_j = \sum_i \Phi_\ve(\la^\ve)_i, 
$$
therefore, it suffices to show that $\Phi_\ve(\la^\ve)_i$ are uniformly bounded for each $i$. Note that whenever $\Phi_\ve(\la^\ve)_i>1$,  our assumption that 
$$
\Psi_\ve(\la^\ve)_i = \la^\ve_i
$$
implies  $\la^\ve_i = 0$, which in turn implies that $\Phi_\ve(\la^\ve)_i \le \ve$. For $\ve<1$ this implies 
\begin{equation}
	\label{eq:VCGHZ1}
\Phi_\ve(\la^\ve)_i\le 1
\end{equation}
for all $i$, and $C^\ve_j \le n$ for all $j$.

Thus the points $X^\ve$, $\la^\ve$ and $C^\ve$ belong to compact sets.
Thus the sequence $\{(X^{1/k},C^{1/k},\la^{1/k})\}_{k=1..\infty}$ contains a converging subsequence. 

More precisely, we get a sequence $\ve_k\ra 0$ such that 
$$
\lim_{k\ra\infty} X^{\ve_k} =: X; 
$$
$$
\lim_{k\ra\infty} C^{\ve_k} =: C; 
$$
and 
$$
\lim_{k\ra\infty} \la^{\ve_k} =: \la, 
$$
 We claim that $X$, $C$, and $\la$  satisfy the conditions of the theorem. 

The {\bf first} condition in the theorem holds because $X^{\ve_k}$ is a valid allocation for each $k$, and thus the limit is also a valid allocation. 

When $(\la'_1,\ldots,\la'_n)$ varies within $[\la_1,\la_1+\ve]\times [\la_2, \la_2+\ve]\times\ldots\times[\la_n,\la_n+\ve]$, item prices vary by at most $\ve n$. Therefore, whenever $X^{\ve}_{ij}>0$, the amount player $i$ pays per unit of $j$ on average differs from $C^{\ve}_j$ by at most $\ve n$. Hence 
\begin{equation}\label{eq:HZVCG2}
\left| \Phi_\ve(\la^\ve)_i - \sum_j X^{\ve}_{ij}\cdot C^{\ve}_j \right|
\le \ve n. 
\end{equation}
By \eqref{eq:VCGHZ1} this implies 
$$
\sum_j X^{\ve}_{ij}\cdot C^{\ve}_j \le 1+ \ve n. 
$$
By taking the limit over $X^{\ve_k}$ and $C^{\ve_k}$, we get 
$$
\sum_j X_{ij}\cdot C_j \le 1,
$$
for all $i$, implying the {\bf fourth} condition of the theorem.

The {\bf second} and {\bf third} conditions follow from the fact that the optimal value attainable by an allocation is uniformly continuous in the vector $\la$. One characterization of the VCG prices $C_j$ is the difference between the optimal utility attainable when two copies of item $j$ are available, vs. the utility when only a single copy is available. By this characterization, whenever $\la^{\ve_k}\ra\la$, the VCG prices corresponding to any 
$\la'\in [\la^{\ve_k}_1, \la^{\ve_k}_1+\ve_k]\times [\la^{\ve_k}_2, \la^{\ve_k}_2+\ve_k]\times \ldots \times [\la^{\ve_k}_n, \la^{\ve_k}_n+\ve_k]$ will uniformly (in $\ve_k$) converge to VCG prices corresponding to $\la$. Thus, $C^{\ve_k}$ converge to VCG prices corresponding to $\la$ --- implying that $C$ gives us the prices corresponding to VCG on $(\la_i u_i)$. 

Similarly, if $\pi$ is a permutation such that $X_{i\pi(i)}>\de$ for all $i$ and some $\de>0$, then for all sufficiently large $k$
$$
X^{\ve_k}_{i\pi(i)}>0,
$$
which implies 
$$
\sum_i \la^{\ve_k}_i u_{i\pi(i)} > \max_{\si} 	\la^{\ve_k}_i u_{i\si(i)} - 2n \ve_k.
$$
Taking $k\ra\infty$, this implies 	$\sum_i \la_i u_{i \pi(i)} \ge \max_{\si} 	\sum_i \la_i u_{i \si(i)}$.

Taken together, the first four properties imply that $X$ is a viable VCG outcome for utilities $u_i'=\la_i u_i$, supported by prices $C_j$. 

To establish the {\bf fifth} property, we consider players $i$ who exhaust their budgets ($\sum_j C_j X_{ij}=1$), and those who don't exhaust ($\sum_j C_j X_{ij}<1$) separately. 

By the envy-freeness of VCG, players who pay $1$ unit cannot obtain a better bundle for one unit, which is exactly what the fifth property asserts. 
If a player $i$ pays strictly less than $1$ unit, then by \eqref{eq:HZVCG2} for all sufficiently large $k$, $\Phi_{\ve_k}(\la^{\ve_k})_i<1$. By the fixed point property, this means that $\la^{\ve_k}_i = \lamax$, and thus $\la_i=\lamax$. We finish the proof by claiming that whenever $\la_i=\lamax$, the VCG unit-demand mechanism corresponding to $(u'_t)_{t=1}^n = (\la_t u_t)_{t=1}^n$ will always allocate player $i$ her favorite items only, making the {\bf fifth} property hold automatically. 

\begin{claim}
	\label{cl:vcg1}
	Suppose $\la_i=\lamax$, then for all $j$ with $X_{ij}>0$, 
	$$
	u_{ij} = \max_{\ell} u_{i\ell}=: u_i^*. 
	$$
\end{claim}

In other words, a player $i$ with $\la_i=\lamax$ only gets allocated her favorite item(s).

\begin{proof}[Proof of Claim~\ref{cl:vcg1}.]
	Suppose $X_{ij}>0$ for some $j$ with $u_{ij}<u_{i\ell} = u_i^*$. 
	Every allocation under VCG is envy-free. Therefore, whenever player $i$ is allocated item $j$ under $X$, whoever is allocated item $\ell$ pays at least $\la_i \cdot (u_i^* - u_{ij})$ for the item, and thus the cost of item $\ell$ is at least $C_\ell\ge\la_i \cdot (u_i^* - u_{ij})$. Therefore, by  \eqref{eq:VCGHZ1}, 
	$$
	n\ge \sum_q C_q \ge C_\ell \ge\la_i \cdot (u_i^* - u_{ij})= 
	\lamax \cdot (u_i^* - u_{ij})>n,
	$$
contradiction.
\end{proof}
\end{proof}

\subsubsection{Not all no-regret repeated VCG executions correspond to a HZ equilibrium}
\label{sec:notallVCG}

One way to interpret Theorem~\ref{thm:VCG-HZ} is that in the setting of allocation with cardinal preferences, there is always a competitive equilibrium in the sense of Hylland and Zeckhauser that is supported by weights $\la$ and a (combination of) VCG executions on utilities $u'_i = \la_i u_i$. If we initialized Algorithm~\ref{alg:1} to weights $\la_i$, and ran it without a regularizer, using time to alternate between the permutations that make up $X$, we would obtain a valid execution of the algorithm corresponding to outcome $X$. The strong no-regret property in this case follows directly from the truthfulness of VCG. 

It is reasonable to ask whether {\em all} valid (i.e. low strong-regret) executions of repeated VCG --- corresponding to running Algorithm~\ref{alg:1} without a regularizer --- lead to an (approximate) HZ competitive equilibrium.
The following examples shows that the answer is `no'. The example is somewhat pathological, but is illuminating nonetheless. 

Consider the following simple setting with just two players and two items. 
$$
\begin{array}{c|cc}
	& A & B   \\ \hline
	u_1 & 1 & 0\\
	u_2 & 1 & 0 \\
\end{array}
$$
Both players (equally) prefer item A to item B. Any allocation coming from a HZ competitive equilibrium\footnote{And, indeed, any reasonable allocation.} would divide the items equally among the two players.

Consider the following submissions of $\la_{1,t}$ and $\la_{2,t}$ to Algorithm~\ref{alg:1}:
\begin{equation}\label{eq:exec1}
\begin{array}{c|cccccccccc}
	\mathbf{t} & 1 & 2 & 3& 4 & 5 & 6& 7 & 8 & 9 & \ldots   \\ \hline
	\la_{1,t} & 3 & 3 & 3 & 3 & 3 & 3 & 3 & 3 & 3 &\ldots\\
   \la_{2,t} & 4 & 1.5 & 1.5 & 4 & 1.5 & 1.5 & 4 & 1.5 & 1.5 & \ldots\\
\end{array}
\end{equation}
During the execution, player $1$ receives item A during times $t=3i+2$ and $t=3i+3$ for $i\ge 0$. Whenever this happens, player $1$ pays $1.5$ units, which averages to $1$ unit per time step. Player $2$ receives item B during times $t=3i+1$ for $i\ge0$. Whenever this happens, player $2$ pays $3$ units, which also averages to $1$ unit per time step. Thus, the allocation we end up with is:
$$
\begin{array}{c|cc}
	& A & B   \\ \hline
	X_1 & \frac{2}{3} & \frac{1}{3}\\
	X_2 & \frac{1}{3} & \frac{2}{3} \\
\end{array}
$$
Such an allocation cannot be supported by a HZ equilibrium. To see that the execution \eqref{eq:exec1} has the strong low-regret property, observe that player~$1$ has no useful deviation from playing $\la_{1,t}=3$ at every round. Suppose she deviated in a way that gives her item $A$ during $\al T\le T/3$ of the rounds in which player $2$ plays $4$, and $\be T\le 2T/3$ of the rounds in which player $2$ plays $1.5$. The payment in the first case is $4$ and in the second is $1.5$ for a total of $(4\al+1.5\be)T\le T$ by the budget constraint. The utility of player $1$ is at most 
$$
\al+\be = \frac{2}{3} \cdot (1.5 \al+ 1.5 \be) \le 
 \frac{2}{3} \cdot (4 \al+ 1.5 \be) \le \frac{2}{3}. 
$$
Player $2$ has no useful deviations either. At any round where she gets item $A$, she has to pay $3$ units, which means that she can at most get item $A$ during $T/3$ rounds.

The example is quite different from the setup in the proof of Theorem~\ref{thm:VCG-HZ}. In the proof of Theorem~\ref{thm:VCG-HZ}, in all the approximate solutions, the players' $\la_i$'s are confined to a small region of size $\ve_k$. In the example above, player $2$'s $\la_2$ oscillates between two very distant values. 
Note that while player $1$ strongly prefers to not deviate from her current play, player $2$ only has a weak incentive. In fact, if instead of playing $(4,1.5)$ the second player played $(4,2.5)$, her payoff would have been the same, but the game would no longer be feasible for player $1$. 

Even a very modest incentive to keep $\la_{2,t}$'s close to each other would have ruled out this kind of example. A concave regularizer provides this kind of  incentive.

\subsubsection{Regularized VCG corresponds to an approximate HZ equilibrium}
\label{sec:reg}

In this section we will show that a {\em regularized} execution of Algorithm~\ref{alg:1} avoids the pathological example from previous section, and does lead to an approximate HZ equilibrium. As in other fields, such as learning and optimization, regularization sacrifices a small amount of efficiency to attain stability. 

Before considering executions of ``regularized VCG"  in Algorithm~\ref{alg:1}, let us see what property we hope would lead to the outcome being an approximate HZ competitive equilibrium. The property we need is that for each player  $i$ the values of $\la_{i,t}$ do not vary greatly throughout the execution. We can then define the ``typical" value of $\la_{i,t}$ to be $\la_i$, and use the VCG prices induced by the $\la_i$'s to define an approximate HZ equilibrium. 

After that, we will see under what conditions low strong regret implies ``$\la_{i,t}$ do not vary greatly throughout the execution". As we will see, this happens whenever the function mapping payment $C_{i,t}$ to the utility experienced by player $i$ at round $t$ is smooth and strictly concave --- having  second derivative bounded away from $0$. This is a property that fails to hold for standard VCG. In the example from Section~\ref{sec:notallVCG} the cost/utility function for player $2$ (defined as ``how much utility can one derive by spending an average of $c$ units of cost) is given by 
$$
U_2(c)=\left\{
\begin{array}{ll}
	c/3 &\text{ when $c\le 3$}\\
	1 & \text{ when $c>3$} \\
	\end{array}
\right.
$$
The function $U_2$ is neither smooth --- its derivative drops from $1/3$ to $0$, nor strictly concave --- it is linear on the interval $[0,3]$. This causes player $2$ to be indifferent among all the $\la_2$'s in $(0,3)$ and among all $\la_2$'s in $(3,\infty)$, and makes a solution where player $2$ alternates between  $\{1.5,4\}$, as opposed to alternating between $\{3-\ve,3+\ve\}$ a low-regret solution for player $2$.

\paragraph{The role of regularization.} Next, we will see how regularization yields an approximate HZ-equilibrium. 

Rather than try to give a general implication, we will work out an example of one specific regularizer, to show that low strong-regret executions with this regularizer correspond to approximate HZ equilibria. It should be noted that regularizers in general reduce the efficiency of the mechanism, since adding a ``utility function" for the principal necessarily reduces the utility of players. 

The main theorem of the section states that a valid execution of Algorithm~\ref{alg:1} with regularization indeed leads to an approximate HZ-equilibrium supported by VCG prices. The proof is not technically deep but requires careful calculations which we defer to Appendix~\ref{app:thm:reg}. 

\medskip \noindent {{\bf Theorem~\ref{thm:reg}.} {\em [restated]}~~~}
{\em 	In the unit-demand allocation setting without money with $n$ players and $n$ items, let $\{u_{ij}\}$ be utilities such that $u_{ij}\in [0,1]$, and for each $i$, $\min_j u_{ij} = 0$ and $\max_j u_{ij}=1$. 
	
	For each $\de>0$, there are $\be=(\de/n)^{O(1)}$, $\ve=(\de/n)^{O(1)}$, and 
	$\bla=O(1/\de)$, such that if we use the concave regularizer
	$$
	F_0(x):=\sum_{ij} -\be/x_{ij}\le 0,
	$$
	the following holds. 
	
	Consider an execution of Algorithm~\ref{alg:1}, with $\la_{i,t}\in[0,\bla]$. 
	Suppose that each player has strong regret $<\ve \cdot T$. Let $x$ be the resulting allocation.
	
	Let $\la_i$ be the best response for each player $i$ to the observed sequence of actions. 	
	Let $C_j$ be VCG prices corresponding to utilities $\{\la_i u_{ij}\}$. 
	
Then $x$ is a $\de$-competitive equilibrium at budgets $1$ supported by prices $C_j'=(1-\de)\cdot C_j$.}
	
\medskip

\ignore{

\vspace{50mm}

\begin{proof}
	Applying Claim~\ref{cl:reg5} to each player $i$ and taking $B:=\cup_i B_i$, we obtain a set $B\subset [T]$ with $|B|<T\cdot \ve^{0.1}\cdot n^{O(1)}$ and

By Claim~\ref{cl:reg3} we know that since the sum of payments across all rounds is $\le n\cdot T$, $\max_{i} \la_{i,t} < O(n^{12})$ for at least half the rounds, and that replacing $\la_{i,t}$ with $\la_i>\omega(n^{12})$ will cause player $i$ to exceed her budget. Therefore, 
\begin{equation}\label{eq:561}
\bla:=\max_i	\la_i < n^{O(1)}. 
\end{equation}  
Note that for all $j$, $C_j\le \bla$. 

Let 
$$
\de:= \max\left( \frac{10 \bla |B|}{T}, \frac{10 s}{T}, \ve^{0.1}\right)
< \ve^{0.1} \cdot n^{O(1)}.
$$

Fix a round $t\in[T]$. Let $P_{i,t}$ be the price player $i$ paid in round $t$. Let $C_{j,t}$ be the VCG price of item $j$ with utilities $\{\la_{i,t}u_{ij}\}$. 
Denote
$$
s_t:= \sum_i |\la_{i,t}-\la_i|.
$$
By Claim~\ref{cl:app1} we have $|C_j-C_{j,t}|\le 2 s_t$. By applying Claim~\ref{cl:reg1} with $M=2$ we get 
$$
\eta < n \cdot \bla + s_t + 2n^3 < n^{O(1)} + s_t. 
$$
\end{proof}

\vspace{50mm}

Is the following true: given $(\la_1,\ldots,\la_n)$ and $(\la'_1,\ldots,\la'_n)$, if deviating in each coordinate from $\la'_i$ to $\la_i$ doesn't change $\ell_2$ distance by much, or doesn't change function value much, then the two points are not too far?

[Concave regularizer $\Rightarrow$ low regret implies concentration]

[Concentration implies a competitive equilibrium based on the VCG prices around values on which the concentration happens]. }

\ignore{
Our main result in this section will be a correspondence between HZ competitive equilibria and valid executions of Algorithm~\ref{alg:1}. It remains to be seen whether this correspondence can be used to obtain better heuristics for the CEEI solutions for one-sided matching. 

As in earlier discussions, we will say that an allocation supported by prices is a competitive $\ve$-equilibrium if \eqref{eq:HZeq} holds up to an additive $\ve$\footnote{Recall that the utilities are normalized to $u_{ij}\in[0,1]$.}:
	\begin{equation}
	\sum_{j} u_{ij}\cdot q_{ij}\le 	\sum_{j} u_{ij}\cdot p_{ij}+\ve.
	\label{eq:HZeqapprox}
\end{equation}

We can run Algorithm~\ref{alg:1} with $f_i(p_i) := u_i^T p_i$, where the output of each round is a bipartite matching allocating items to players. 
The theorem below asserts that running Algorithm~\ref{alg:1} ``as-is" without a regularizer yields solutions closely related to HZ competitive equilibria. 

\begin{thm}
	\label{thm:HZ} Suppose utilities $\{u_{ij}\}$ are in general position. 
	\begin{enumerate}
		\item 
		Let $p_{ij}$ be an allocation, and $C_j$ be prices supporting a competitive HZ equilibrium, then there is a valid execution of Algorithm~\ref{alg:1} that yields the same allocation and  where players experience zero
 strong regret.
		\footnote{It is, in fact, 
		possible to formulate a version of this statement starting from an 
	approximate HZ equilibrium. However, the result will be an execution of Algorithm~\ref{alg:1} where the optimization is also approximate (and not just $\ve$-strong regret experienced by players). We stick to the exact version of this direction to keep the presentation simple --- the interesting approximate direction is the converse one.}
		\item 
		Conversely, suppose $p_{ij}$ is an allocation that resulted from an execution of Algorithm~\ref{alg:1} where each player's strong regret is bounded by $\ve$. Then there are prices $C_j$  that make the allocation 
		 a competitive HZ $O(\ve)$-equilibrium.
	\end{enumerate}
\end{thm} 

\begin{proof} 
{\bf From HZ equilibrium to a valid execution.}  Suppose $\{p_{ij}\}$ is a valid HZ allocation supported by prices $C_j\ge 0$. For each player~$i$, one of two possibilities hold:
\begin{itemize}
	\item 
	Player $i$ does not spend all of her budget --- meaning that she is allocated her favorite item; or
	\item 
	Player $i$ spends all of her budget: $\sum_j p_{ij}\cdot C_j = 1$. 
\end{itemize}
By duality, there are numbers $a_i\ge 0$ and $b_i\ge 0$ such that for each $j$
\begin{equation}
	\label{eq:HZpf1}
	u_{ij} \le a_i + C_j b_i, 
\end{equation}
with equality holding whenever $p_{ij}>0$. In the first case, $b_i=0$ and the second case, 
$b_i>0$.

We partition the sets of players and items into sets $S_r$ of players and $T_r$ of items such that $|S_r|=|T_r|$, and (1) players from $S_r$ are only ever allocated items from $T_r$; (2) the partition is minimal with respect to this property: there isn't $S\subsetneq S_r$ and $T\subsetneq T_r$ with $|S|=|T|$ so that players from $S$ are only allocated items in $T$. For a player $i$ let $r(i)$ be such that $S_{r(i)}\ni i$. For an item $j$ let $r(j)$ be such that $T_{r(j)}\ni j$.   Note that players $i$ who do not exhaust their budget belong to a singleton $S_{r(i)}=\{i\}$. 

Let $\rho_r\ge 0$ be parameters to be chosen later. Set 
\begin{equation}
	\la_i :=\left\{ \begin{array}{ll}
       \rho_{r(i)}& ~~\text{ if $b_i=0$} \\
     \rho_{r(i)}/b_i& ~~\text{ if $b_i>0$} 
	\end{array}\right.
\end{equation}

Let $\si:[n]\ra[n]$ be an allocation, and let 
$$
F(\si):=\sum_i \la_i \cdot u_{i\si(i)}. 
$$ 
$F$ extends naturally to fractional matching matrices $q_{ij}$:
$$
F(q):=\sum_{ij} \la_i u_{ij}. 
$$
We claim that if $\pi$ is such that $p_{i\pi(i)}>0$ for all $i$, then 
$\pi \in \arg\max_{\si} F(\si)$. 
Denote by $Z$ the set of items allocated to players who do not exhaust their budget.
Note that 
$$
F(\si)=\sum_i \la_i \cdot u_{i\si(i)}\le \sum_i \la_i a_i + 
\sum_{i:b_i>0} \rho_{r(i)}\cdot  C_{\si(i)} =  \sum_i \la_i a_i + \sum_{j\notin Z} \rho_{r(j)} \cdot C_j. 
$$ 
Whenever $p_{i\pi(i)}>0$ for all $i$ equality holds. 

By the Birkhoff–von Neumann  theorem, we can write $p$ as a convex combination of permutations $\pi$, each of which maximizes $F(\si)$, and thus $p$ also maximizes $F$. Setting $F_t:=F$ for all $t=1..T$, we obtain $p$ as the result of a valid execution of Algorithm~\ref{alg:1}. This is true for any choice of the values of $\rho_r$. 

Next, we will select the values of $\rho_r$ so that no player exceeds her budget under the payments prescribed by the algorithm, and all players who  exhaust their budgets under the HZ equilibrium also exhaust their budget under Algorithm~\ref{alg:1}. Whenever $|S_r|>1$, it must be the case that not all items in $T_r$ are priced the same. Let
\begin{equation}
	\rho_r :=\left\{ \begin{array}{ll}
		R& ~~\text{ if $|T_r|=1$} \\
		\frac{|T_r|}{\left(\sum_{j\in T_r} C_j\right) - |T_r|\cdot \min_{j\in T_r} C_j}& ~~\text{ if $|T_r|>1$} 
	\end{array}\right.
\end{equation}

Here $R>0$ is a large constant chosen so that the following property holds. For any player $i$ who does not exceed her budget and receives her favorite item $j$, whenever $\la_{k}$ is set to $0$ for a $k\neq i$, player $i$ still receives her favorite item $j$ under any allocation $\si$ maximizing 
$$F^{-i}(\si)=F(\si)-\la_k \cdot u_{k\si(k)}.$$

We need to calculate the prices Algorithm~\ref{alg:1} charges the players. 

\noindent \underline{Players not exhausting their budget under the HZ equilibrium.} Let $i$ be a player not exhausting her budget under the HZ equilibrium. Let $\si$ be an allocation maximizing $F^{-i}(\si)$ (note that since $F^{-i}$ is linear, there is always a maximizing permutation). s

It remains to verify the low strong regret property. A player who does not exhaust her budget in the HZ equilibrium will always get allocated her favorite item, and thus will experience no regret. Consider a player $i$ who exhausts her budget. 
\end{proof}
}

\subsection{Two-sided matching}
\label{sec:twosided}

The third important application of mechanisms without money is that of two-sided matching. The most famous algorithm is this area is the Gale-Shapley deferred acceptance algorithm for stable matching. A pair of players form a {\em blocking pair} for a matching $M$ if they are not matched to each other under $M$, but prefer each other to their current partners. A match $M$ is {\em stable} if 
there it has no blocking pairs. Stability is a desirable property since it makes enforcing that the players follow $M$ easy --- there are no useful deviations from $M$ that would benefit all deviating players. 

Stability is a notion that only depends on ordinal preferences. As with voting and one-sided matching, it is often desirable to incorporate cardinal utilities into the preference model, with the goal of attaining cardinally efficient outcomes. Unfortunately, stability is generally completely incompatible with efficiency. 

Consider the following example with $n=2$. There are two hospitals $h_1$ and $h_2$ and two doctors $d_1$ and $d_2$. 

\begin{equation}
	\label{eq:coup1}
	\begin{array}{c|cc}
		\multicolumn{3}{c}{\text{\em Hospitals' utilities}}
		\\
		& d_1 & d_2\\\hline
		h_1 & 9 & 10\\
		h_2 & 0 & 9
	\end{array}~~~~~~~~~~~~~~~~~~~~
\begin{array}{c|cc}
	\multicolumn{3}{c}{\text{\em Doctors' utilities}}
	\\
	& h_1 & h_2\\\hline
	d_1 & 9 & 0\\
	d_2 & 10 & 9
\end{array}
\end{equation}

The only stable matching in \eqref{eq:coup1} is $M_1:=\{(h_1,d_2), (h_2,d_1)\}$, since otherwise $(h_1,d_2)$ would form a blocking pair. The total utility of such a matching is $20$, while the utility of matching $M_2:=\{(h_1,d_1), (h_2,d_2)\}$ is $36$. This is because even though $(h_1,d_2)$ form a blocking pair for $M_2$, they are almost indifferent between the two matchings, while the other two participants strongly prefer $M_2$ to $M_1$. 

Introducing money transfers can help address the efficiency problem\footnote{For example, under preferences \eqref{eq:coup1} $h_2$ could pay $d_2$ a little bit to make her prefer $M_2$ over $M_1$.}. Without money, it is generally unknown how to achieve efficiency and truthfulness. As in other settings without money, the solution concept has to be invariant to scaling players' utilities, which means that only a sum of scaled utilities can be maximized --- the outcome should be invariant to scaling and shifting of individuals' entire utility vectors. Thus \eqref{eq:coup1} becomes

\begin{equation}
	\label{eq:coup2}
	\begin{array}{c|cc}
		\multicolumn{3}{c}{\text{\em Hospitals' utilities}}
		\\
		& d_1 & d_2\\\hline
		h_1 & 0 & 1\\
		h_2 & 0 & 1
	\end{array}~~~~~~~~~~~~~~~~~~~~
	\begin{array}{c|cc}
		\multicolumn{3}{c}{\text{\em Doctors' utilities}}
		\\
		& h_1 & h_2\\\hline
		d_1 & 1 & 0\\
		d_2 & 1 & 0
	\end{array}
\end{equation}

We see that $h_1$ and $d_2$ are the desirable participants. Under a stable match, they will be matched to each other, even though $M_1$ and $M_2$ have the same total utility. 

Applying Algorithm~\ref{alg:1} to \eqref{eq:coup2} would lead to assigning the same weight $\la_i=1$ to all participants, and to an outcome $\mu:= \frac{1}{2}\cdot M_1+\frac{1}{2}\cdot M_2$. 

To see that the externalities are indeed equalized, note that $h_1$ and $d_2$ prefer $M_1$ and $h_2$ and $d_1$ prefer $M_2$. 
Under $\mu$ the total utility of players $h_1,h_2,d_1$ is $\frac{3}{2}$. 
If we ignored the preferences of (say) $d_2$, then $\mu$ would be replaced with $M_2$ with probability $1$. The total utility  of players $h_1,h_2,d_1$  under $M_2$ is $2$. Thus $d_2$ causes $\frac{1}{2}$ unit of externality on other players. This calculation can be repeated to see that each player's externality is $\frac{1}{2}$, showing that $\mu$ indeed equalizes externalities across players. 

\paragraph{\bf Is desirability treated as an endowment?} In the example above there is a significant difference between the uniform distribution that Algorithm~\ref{alg:1} outputs and the single stable matching $M_1$ in which the more desirable hospital matches to the more desirable doctor. This example can be expanded to a setting with $n>2$, where the contrast is even more stark. 
Consider the most straightforward setting where each doctor derives utility $i$ from hospital $h_i$ and each hospital derives utility $i$ from doctor $d_i$. Thus all participants agree on the ranking $h_1\prec h_{2}\prec\ldots \prec h_n$ and
$d_1\prec d_2 \prec \ldots \prec d_n$.  

The only stable matching in this case is the assortative $a:=\{(h_i,d_i)\}_{i=1}^n$ matching. This is easy to see by induction: in a stable match $(h_n,d_n)$ must be together, otherwise they will form a blocking pair. Assuming $(h_n,d_n)$ are matched to each other, $(h_{n-1},d_{n-1})$ will form a blocking pair unless they are matched to each other, and so on. 

Under Algorithm~\ref{alg:1}, all preferences are the same, and by symmetry, in the resulting distribution\footnote{There are many ways to implement such distribution --- the algorithm only specifies the {\em ex-ante} marginal probabilities of pairs $(h_i,d_j)$.}, each pair $(h_i,d_j)$ will appear with equal probability $\frac{1}{n}$. 

We believe that both outcomes give meaningful solutions under very different solution concepts. A deeper investigation of the solution concept given by Algorithm~\ref{alg:1} will need to be deferred to future work, but we can offer some preliminary comments here. 

The stark difference between the two outcomes can be traced to how the desirability is treated by the mechanism. Under Algorithm~\ref{alg:1} the desirability of $h_n$ is just part of the input landscape. The ``benefit" from $h_n$ being so desirable doesn't accrue to $h_n$, and thus $h_n$ gets the same outcome as the least desirable hospital $h_1$. Put differently, Algorithm~\ref{alg:1} measures the (negative) externality caused by $h_n$ having preferences, but not the (positive) externality caused by $h_n$ being present. 

On the other hand, when we say that $(h_n,d_n)$ are a blocking pair, it is implied that $h_n$'s and $d_n$'s desirability accrue to them, and they can internalize them by matching with each other. 

There are variants with Hylland-Zeckhauser with endowments \cite{echenique2019constrained,echenique2019fairness,garg2020arrow}, and it is possible to adapt Algorithm~\ref{alg:1} to treat desirability as an endowment and to give match participants credit for being desirable. It is also possible to create a hybrid approach, where the match is partially redistributive, for example finding a solution with lowest level of inequality among externalities subject to (inequality-promoting) stability constraints. We   leave investigating these adaptations to future works.

\bibliographystyle{alpha}
\bibliography{ref}

\begin{appendix}
\section{Properties of unit-demand VCG}

In this section we summarize some useful properties of unit-demand VCG. A detailed discussion on the properties of unit-demand VCG can be found e.g. in \cite{leonard1983elicitation}.

\paragraph{Notation.} Suppose there are $n$ players and $n$ items. Player $i$ has utility $u_{ij}\in [0,1]$ for item $j$. Let $OPT$ denote the maximum  utility attainable by a permutation. 
$$
OPT := \max_{\pi:[n]\hookrightarrow[n]} \sum_i u_{i\pi(i)}. 
$$
Let 
$$
OPT_{+j} := \max_{\pi:[n]\hookrightarrow[n]\cup\{j'\}} \sum_i u_{i\pi(i)},
$$
where $u_{ij'}:=u_{ij}$, be the maximum attainable utility if a second copy of item $j$ becomes available. 

The optimization problem can be made convex by replacing the permutation with bi-stochastic matrices. Bi-stochastic matrices correspond to distributions over permutations. Thus, one gets a linear program:
\begin{equation}
	\label{eq:app1}
\left\{\begin{array}{l}
\text{maximize } \sum_{ij} u_{ij} x_{ij} \text{ subject to:}\\
\forall i ~ \sum_j x_{ij}\le 1\\
\forall j~\sum_i x_{ij} \le 1
\end{array}\right.
\end{equation}
The dual to \eqref{eq:app1} finds variables $a_i$ and $b_j$ such that 
\begin{equation}
	\label{eq:app2}
	\forall i,j~~ u_{ij}\le a_i + b_j,
\end{equation}
where equality holds whenever the in the optimal solution $x^{*}$, $x^*_{ij}>0$. Thus
$$
OPT = \sum_i a_i + \sum_j b_j. 
$$

\paragraph{VCG prices.} We state some properties of VCG prices. 

\begin{claim}
	\label{cl:app11}
\begin{enumerate}
	\item 
	{\bf Item price independent of receiver.} Whenever there are multiple optimal solutions, the same item $j$ is sold for the same price $C_j$ --- the VCG price of $j$.
	\item 
	{\bf Item price is benefit from a second copy.} This price is equal to $OPT_{+j}-OPT$ --- the extra welfare from having another copy of $j$. 
   \item 
{\bf Prices as dual variables.}  Let $\pi$ be an optimal allocation. Define $$b_j:=C_j,~~ a_i:=u_{i\pi(i)} - C_i.$$ Then $\pi$ with prices $C_j$ results in an envy-free allocation --- equivalently, $\{a_i\}$, $\{b_j\}$ form a valid solution for the dual program \eqref{eq:app2}. 
  \item 
  {\bf Fractional augmentation.} Let $y$ be an allocation vector with $y_j\ge 0$, $\sum_j y_j\le 1$. Let $OPT_{+y}$ be the value of an optimal allocation where the amount of each item $j$ available is $1+y_j$ rather than $1$. Then 
  \begin{equation}
  	\label{app:3}
  	OPT_{+y} = OPT + \sum_j y_j C_j. 	
  \end{equation}
\item 
{\bf Continuity of prices in utilities.} Let $u_{ij}\ge 0$ and $\ti{u}_{ij}$ be two sets of utilities. Let $C_j$ and $\ti{C}_j$ be the corresponding $VCG$ prices. Then for all $j$, 
\begin{equation}
	\label{app:4}
|C_j - \ti{C}_j | \le  2\cdot \sum_{i}\|u_{i}-\ti{u}_{i}\|_{\infty}=
2\cdot \sum_{i}\max_j |u_{ij}-\ti{u}_{ij}|.
\end{equation}
\end{enumerate}
\end{claim}
\begin{proof}
	The first three statements are standard properties of unit-demand VCG. 
	
	For the {\em fractional augmentation} property, we will prove an inequality in both directions to obtain equality.
	Let $x$ be an optimal allocation realizing $OPT$.  Let $x^{+j}$ be an allocation realizing $OPT_{+j}$, thus 
	$$\forall \ell~~\sum_i x^{+j}_{i\ell}\le 1 + {\mathbf 1}_{\ell=j}~~~\text{ and }~~~\sum_{i\ell} x^{+j}_{i\ell} u_{i\ell} = OPT_{+j}.$$
	Consider $\ti{x}:=(1-\sum_j y_j)\cdot x+\sum_j ( y_j\cdot x^{+j})$. Then each player is allocated a total of one unit under $\ti{x}$. We have 
	$$
	\forall \ell \sum_i \ti{x}_{i\ell}\le 1 + \sum_{j} y_j \cdot {\mathbf 1}_{j=\ell} = 1 + y_j, 
	$$ 
	making $\ti{x}$ a feasible solution for $OPT_{+y}$. We have 
	$$
	\sum_{ij}\ti{x}_{ij} u_{ij} = (1-\sum_j y_j)\cdot OPT + \sum_j y_j OPT_{+j} = OPT + \sum_j y_j C_j.
	$$
	Thus $OPT_{+y} \ge OPT + \sum_j y_j C_j$. 	
	
	For the converse inequality, we have that for each $i,j$, $u_{ij}\le a_i + C_j$, where $OPT=\sum_i a_i + \sum_j C_j$. Let $z$ be a solution for realizing $OPT_{+y}$. We have 
	$$
	OPT_{+y}=\sum_{ij}  u_{ij}z_{ij} \le \sum_{ij} (a_i + C_j)z_{ij} \le \sum_i a_i + 
	\sum_j (1+y_j) C_j  = OPT +	\sum_j y_j C_j.
	$$
	We note that the $\le$ direction of this claim continues to hold even when $\sum_j y_j>1$, but the inequality may no longer be tight. 
	
	\smallskip
	For the {\em continuity of prices in utilities} property, we will prove that $\ti{C}_j \le C_j +  2\cdot \sum_{i}\|u_{i}-\ti{u}_{i}\|_{\infty}$. Together with the same inequality with $\ti{C}_j$ and $C_j$ swapped, \eqref{app:4} follows. Let $x$ be a utility-maximizing allocation under $u$, and let $x^{+j}$ be a utility-maximizing allocation under $u$ when a second copy of item $j$ is available. Similarly, let $y$ and $y^{+j}$ be the corresponding optimal allocations under $\ti{u}$. We have 
	\begin{align*}
		\ti{C}_j = &\sum_{i\ell} y^{+j}_{i\ell} \ti{u}_{i\ell} - 
	\sum_{i\ell} y_{i\ell} \ti{u}_{i\ell} \\& \le \sum_i \|u_i-\ti{u}_i\|_{\infty} + 
	 \sum_{i\ell} y^{+j}_{i\ell} {u}_{i\ell} - 
	\sum_{i\ell} y_{i\ell} \ti{u}_{i\ell}\\& \le 
	 \sum_i \|u_i-\ti{u}_i\|_{\infty} + 
	\sum_{i\ell} x^{+j}_{i\ell} {u}_{i\ell} - 
	\sum_{i\ell} x_{i\ell} \ti{u}_{i\ell}\\&  \le 
	2\cdot \sum_i \|u_i-\ti{u}_i\|_{\infty} + 
	\sum_{i\ell} x^{+j}_{i\ell} {u}_{i\ell} - 
	\sum_{i\ell} x_{i\ell} {u}_{i\ell} \\& = 	2\cdot \sum_i \|u_i-\ti{u}_i\|_{\infty} + C_j,
\end{align*}
where the second inequality is by optimality of $x^{+j}$ and of $y$. 
\end{proof}

\section{Proof of Theorem~\ref{thm:reg}}
\label{app:thm:reg}

In this section we continue the discussion from Section~\ref{sec:reg} to give a proof of Theorem~\ref{thm:reg}.

Fix the setting where there are $n$ players and $n$ items, and utilities $\{u_{ij}\}\in [0,1]^{n\times n}$. Further, by scaling and shifting, we may assume without loss of generality that for each $i$, $\max_j u_{ij}=1$ and $\min_j u_{ij}=0$. Consider the regularizer
$$
F_0(x):=\sum_{ij} f_0(x_{ij})\le 0,
$$
where $x$ is a (fractional) allocation, and $f_0:(0,1]\ra \RR^{\le0}$ is a real-valued function such that 
\begin{itemize}
	\item $f_0(z)\le 0$;
	\item $f_0(z)$ is increasing, with $\lim_{z\ra 0^+}f_0(z)=-\infty$; and 
	\item $f_0(z)$ is strictly concave with $f_0''(z)\le -\ga$ for a parameter $\ga>0$
\end{itemize}

\begin{claim}\label{cl:reg1}
	Let $u_{ij}\in[0,1]$ be utilities as above. Let $\la_i\ge 0$ be multipliers. 
	Let $x_{ij}$ and $C_j$ be the allocation and prices resulting from running the VCG mechanism on utilities $\la_i u_i$ with regularizer $f_0$.  Let $M>1$ be a parameter. Denote 
	$$
	\eta=\eta(\la):=\frac{\sum_i \la_i}{M} -n^2 \cdot  f_0\left(\frac{1}{n\cdot M}\right).
	$$ Then the following properties hold:
	\begin{enumerate}
		\item {\bf Small efficiency loss due to regularization.} Let $OPT:=\max_{\pi} \sum_{i} \la_i u_{i\pi(i)}$. Then 
		\begin{equation}
			\label{eq:470}
			\sum_{ij} x_{ij} \la_i u_{ij} \ge OPT-  \eta.
		\end{equation}
		\item {\bf Prices close to VCG prices.} Under the regularized mechanism the payment $P_i$ from player $i$ satisfies:
		\begin{equation}\label{eq:471}
			\left|P_i - \sum_j x_{ij} C_j \right| \le \eta. 
		\end{equation}
		\item {\bf Allocation close to a competitive equilibrium at prices $C_j$.} For any alternative bundle $y_{ij}$ player $i$ receives with $\sum_j y_{ij}=1$, 
		\begin{equation}\label{eq:472}
			\sum_j y_{ij} \la_i u_{ij} \le \sum_j x_{ij} \la_i u_{ij} + \sum_j y_{ij} C_j - P_i +  \eta. 
		\end{equation}
	\end{enumerate}
\end{claim}

\begin{proof}
	{\em Small efficiency loss due to regularization:}  Let $y$ be a solution realizing $OPT=OPT(\la)$. Let $e$ be the all-uniform allocation with $e_{ij}=\frac{1}{n}$. Let $y':=(1-1/M) y + e/M$. Then 
	$$
	F_0(y') + \sum_{ij}\la_i u_{ij} y'_{ij} \ge n^2 \cdot f_0(1/(n\cdot M))+OPT - \frac{\sum_i \la_i}{M}.
	$$
	By the optimality of $x$, we have
	\begin{align*}
		\sum_{ij} x_{ij} \la_i u_{ij} & \ge F_0(x) +\sum_{ij} x_{ij} \la_i u_{ij} \\&
		\ge  	F_0(y') + \sum_{ij}\la_i u_{ij} y'_{ij}\\& \ge
		n^2 \cdot f_0(1/(n\cdot M))+OPT - \frac{\sum_i \la_i}{M} \\& = OPT-\eta.
	\end{align*}
	Denote by 
	$$
	OPT_0(\la) = \max_x \left(F_0(x) + \sum_{ij} x_{ij} \la_i u_{ij} \right). 
	$$
	Then we have just shown that 
	\begin{equation}
		\label{eq:481}
		OPT(\la)-\eta(\la) \le OPT_0(\la) \le OPT(\la). 
	\end{equation}
	To show that {\em prices are close to VCG prices}, we will prove two inequalities. It is first useful to remember some general facts about the unit-demand VCG, summarized in Claim~\ref{cl:app11}. Specifically, there exist dual parameters $a_i$ (corresponding to the ``welfare" of player $i$), such that 
	$$
	\forall k,\ell~~~\la_k \cdot u_{k\ell} \le a_k + C_\ell,
	$$ 
	with equality at an optimal allocation, and 
	$$
	OPT(\la) = \sum_k a_k + \sum_\ell C_\ell. 
	$$
	Denote by $\la^{-i}$ the setting where $\la^{-i}_k = \la_k$ for $k\neq i$, and $\la^{-i}_i=0$. We have, by definition, 
	\begin{equation}
		\label{eq:480}
		P_i = OPT_0(\la^{-i}) - \left( OPT_0 (\la)-\sum_j x_{ij}\la_i u_{ij} \right). 
	\end{equation}
	Let $\pi$ be an allocation attaining $OPT(\la)$. Then 
	$$
	OPT(\la^{-i}) = OPT(\la) - \la_i u_{i\pi(i)}+C_{\pi(i)} = \sum_{k\neq i} a_k + \sum_j C_j. 
	$$

	We prove \eqref{eq:471} by proving two inequalities on $P_i$. Using \eqref{eq:480},
	\begin{align*}
		P_i & \ge OPT(\la^{-i}) -\eta(\la^{-i}) - \sum_{k,j} x_{kj} \la_k u_{kj} + \sum_{j} x_{ij}\la_i u_{ij} \\ & = \sum_{k\neq i} a_k + \sum_j C_j- \eta(\la^{-i}) -  \sum_{k\neq i;j\in[n]} x_{kj} \la_k u_{kj} \\ & \ge
		\sum_{k\neq i} a_k + \sum_j C_j- \eta(\la) -  \sum_{k\neq i;j\in[n]} x_{kj}(a_k + C_j) \\ & = \sum_j x_{ij} C_j- \eta(\la).
	\end{align*}
	Again using \eqref{eq:480}, 
	\begin{align*}
		P_i & \le OPT(\la^{-i})- OPT (\la) + \eta(\la) + \sum_{j} x_{ij}\la_i u_{ij} 
		\\ & = \sum_{k\neq i} a_k + \sum_j C_j + \eta(\la) - \sum_k a_k - \sum_j C_j + \sum_{j} x_{ij}\la_i u_{ij} \\ & = \eta(\la)-a_i+ \sum_{j} x_{ij}\la_i u_{ij} \\ & \le \eta(\la)-a_i+ \sum_{j} x_{ij} (a_i + C_j) \\ & =\eta(\la) + \sum_j x_{ij} C_j.
	\end{align*}
	The {\em allocation close to a competitive equilibrium at prices $C_j$} follows similarly.  Let $y_{ij}$ be any alternative allocation to player $i$ with  $\sum_j y_{ij}=1$. Then 
	\begin{align*}
		\sum_j y_{ij} \la_i u_{ij} & \le  \sum_j y_{ij} (a_i + C_j) 
		= a_i + \sum_j y_{ij} C_j \\& = \left(\sum_j y_{ij} C_j - P_i\right) + P_i + a_i 
		\\ & \le  \left(\sum_j y_{ij} C_j - P_i\right)  +  OPT(\la^{-i})- OPT (\la) + \eta(\la) + \sum_{j} x_{ij}\la_i u_{ij} +a_i  \\ 
		&  = \left(\sum_j y_{ij} C_j - P_i\right) + \sum_{j} x_{ij}\la_i u_{ij} + \eta. 
	\end{align*}

	\ignore{is now a simple corollary of \eqref{eq:471}. Let $y_{ij}$ be any alternative allocation to player $i$ with  $\sum_j y_{ij}=1$ and $ \sum_j y_{ij} C_j \le P_i$. Then 
		using \eqref{eq:471} we get:
		\begin{align*}
			\sum y_{ij} \la_i u_{ij} & \le  \sum y_{ij} (a_i + C_j) 
			= a_i + \sum_j y_{ij} C_j \\ & \le a_i + P_i \le a_i + \sum_j x_{ij} C_j + \eta \\ & = OPT(\la) - \sum_{k\neq i} a_k - \sum_{k\neq i;~j\in[n]} x_{kj} C_j +\eta\\ & \le OPT(\la) -\sum_{k\neq i,~j\in [n]} x_{kj}\la_k u_{kj} +\eta\\& = OPT(\la) - OPT_0(\la) + \sum_j x_{ij} \la_i u_{ij} -f_0(x) + \eta \\& \le \sum_j x_{ij} \la_i u_{ij}+ 2\eta. 
	\end{align*}}
\end{proof}

Informally, the next claim shows that strict concavity of the regularizer implies strict truthfulness.  

\begin{claim}
	\label{cl:reg2}
	Consider an execution of the regularized VCG mechanism with $\la_i = \la$, resulting in an allocation $x_{kj}$ and payment $P_i$ from player $i$. Consider an alternative execution where $\la_i$ is changed to $\la'$, resulting in an allocation $x'_{kj}$ and payment $P'_i$ from player $i$. Then 
	\begin{equation}
		\label{eq:500}
		\left(\sum_j \la_i u_{ij} x_{ij} - P_i \right)-
		\left(\sum_j \la_i u_{ij} x'_{ij} - P'_i\right) 
		\ge \frac{\ga}{2} \cdot \sum_{k,j} |x_{kj}-x'_{kj}|^2. 
	\end{equation}
\end{claim}

\begin{proof}
	Denote 
	$$
	\Psi(x):= F_0(x) + \sum_{k,j\in[n]} \la_k u_{kj} x_{kj}. 
	$$
	We have 
	$$
	\left(\sum_j \la_i u_{ij} x_{ij} - P_i \right)-
	\left(\sum_j \la_i u_{ij} x'_{ij} - P'_i\right)  = \Psi(x)- \Psi(x'). 
	$$
	The function $\Psi$ is strongly concave\footnote{This is where we use the regularizer --- without it the function is merely concave, potentially with regions where the gradient is constant $0$.}, with 
	$$\nabla^2 \Psi(x) = \nabla^2 F_0 \preceq -\ga \cdot I_{n^2}. 
	$$
	Further, since $x$ maximizes $\Psi$, and thus $\nabla\Psi(x)=0$,  we have 
	$$
	\Psi(x)-\Psi(x') \ge \frac{\ga}{2} \cdot \sum_{k,j} |x_{kj}-x'_{kj}|^2. 
	$$
\end{proof}

Consider an execution $\{\la_{i,t}\}_{t=1}^T$ of the algorithm with low strong regret. Regularization on its own is not enough to force $\la_{i,t}$'s to not grow to $\infty$ --- pathological examples can be constructed with arbitrarily large $\la_{i,t}$'s. Instead {\bf  we limit the game space to $\la_{i,t}\in [0,\bla]$. } Here $\bla$ is a parameter (on which the function $f_0$ may depend). As in the proof of Theorem~\ref{thm:VCG-HZ}, for large enough $\bla$, if the limitation of $\la_{i,t}\le \bla$ becomes relevant, then player $i$ gets allocated (close to) her favorite bundle, and the competitive equilibrium condition will hold automatically. 

We will show that for a reasonably chosen $f_0$, and for a large enough $T=n^{O(1)}$, a low strong regret solution translates into an approximate HZ equilibrium based on VCG prices. Our goal will be to streamline the proof --- almost certainly the upper bound we get on $T$ can be tightened to a lower power of $n$.

\paragraph{There is a unique best-response $\la_i$.} 

\begin{claim}\label{cl:reg25} Fix $\{\la_{i,t}\}_{t=1}^T$. Fix a $t$ and an $i$, and consider the payment $P_{i,t}(\la_t)$ and utility $U_{i,t}(\la_t)$ experienced in round $t$ by player $i$ if she reports $\la_i$ instead of $\la_{i,t}$. Let $\ti{\la}_t$ be a best response maximizing $\sum_t U_{i,t}(\la_t)$ subject to 
	$\sum_t P_{i,t}(\la_t)\le T$. Then
	\begin{enumerate}
		\item 
		$P_{i,t}(\la_t)$ is strictly increasing for  $\la_t\ge 0$; 
		\item 
		$P_{i,t}(\la_t)$ is continuous in $\la_t$; 
		\item 
		$P_{i,t}(0)=0$ and $\lim_{\la_t\ra\infty} P_{i,t}=\infty$; 
		\item 
		there is a  utility maximizer for player $i$ of the form $\ti{\la}_t = \la_i$ for all $t$; 
		\item
		$\la_i> 1$; and 
		\item 
		it is the unique maximizer.
	\end{enumerate}
\end{claim}

\begin{proof}
	Fix a round $t\in [T]$. Let 
	$$\Psi_{\la}(x) := \sum_{k\neq i} \sum_j \la_{k,t} u_{kj} x_{kj} +
	\sum_j \la u_{ij} x_{ij} + F_0(x).
	$$
	$\Psi_\la(x)$ is strictly concave, and tends to $-\infty$ on the boundary where $x_{kj}=0$ for some $k,j$. Thus it has a unique maximizer in the interior. Denote it by $X(\la)$. Further, note that if $\la'\neq \la$, then 
	$$
	\nabla \Psi_{\la'}(X(\la)) = \nabla \Psi_{\la}(X(\la)) + (\la'-\la) u_i =
	(\la'-\la) u_i\neq 0. 
	$$
	Therefore $X(\la')\neq X(\la)$. 
	
	We have 
	$$
	U_{i,t}(\la) = \sum_j u_{ij} X(\la)_{ij}. 
	$$
	
	Let $X(0)$ be the point maximizing $\Psi_0(x)$. By definition, we have 
	$$
	P_{i,t}(\la) =  \Psi_0(X(0))-\Psi_\la(X(\la)) + \la\cdot  U_{i,t}(\la). 
	$$
	
	We immediately see that $P_{i,t}(0) = 0$. Moreover, if $j$ is such that $u_{ij}=0$ while $u_{ik}=1$ for some other $k$, as $\la\ra\infty$ we will have $X(\la)_{ij}\ra 0$, and thus $F_0(X(\la))\ra-\infty$, and $P_{i,t}(\la)\ra\infty$. 
	
	Next, let us see that $P_{i,t}(\la)$ is strictly increasing. Suppose $\la'>\la\ge 0$. We have 
	\begin{multline}\label{eq:511}
		\la \cdot U_{i,t} (\la)-P_{i,t}(\la)  = \Psi_\la(X(\la)) -  \Psi_0(X(0))\\  >
		\Psi_\la(X(\la')) -  \Psi_0(X(0)) = \la \cdot U_{i,t} (\la') -P_{i,t}(\la').
	\end{multline}
	Similarly, 
	\begin{multline}\label{eq:512}
		\la' \cdot U_{i,t} (\la')-P_{i,t}(\la')  = \Psi_{\la'}(X(\la')) -  \Psi_0(X(0))\\  >
		\Psi_{\la'}(X(\la)) -  \Psi_0(X(0)) = \la' \cdot U_{i,t} (\la) -P_{i,t}(\la).
	\end{multline}
	By taking $\la'\cdot \eqref{eq:511}+\la\cdot \eqref{eq:512}$, and noting that since $\la'>0$ the inequality remains strict, we obtain
	$$
	-\la' \cdot P_{i,t}(\la) - \la\cdot P_{i,t}(\la')  >
	-\la' \cdot P_{i,t}(\la') - \la\cdot P_{i,t}(\la), 
	$$
	thus
	$$
	(\la'-\la) \cdot  P_{i,t}(\la')> 
	(\la'-\la) \cdot  P_{i,t}(\la), 
	$$
	implying $P_{i,t}(\la')>P_{i,t}(\la)$. 
	
	Similarly, taking $\eqref{eq:511}+\eqref{eq:512}$ yields $U_{i,t}(\la')>U_{i,t}(\la)$. 
	
	By strong concavity of $\Psi_\la$, the value of $X(\la)$ varies continuously in $\la$. Therefore, $U_{i,t}(\la)$ and $P_{i,t}(\la)$ also change continuously in $\la$. 
	
	We have $P_{i,t}(\la)$ a continuous, non-decreasing function that starts at 
	$P_{i,t}(0)=0$ and tends to $\infty$ as $\la\ra\infty$. Therefore the function 
	$$
	P_i (\la) := \sum_{t=1}^T P_{i,t}(\la)
	$$
	also has those properties. In particular, there exists a unique $\la_i$ such that 
	$$
	P_i(\la_i) = T. 
	$$
	By plugging in $\la=0$ and $\la'=1$ into \eqref{eq:512}, we get
	$$
	P_{i,t}(1)< U_{i,t}(1)-U_{i,t}(0) \le 1, 
	$$
	and thus $P_i(1)<T$ and hence $\la_i>1$.

	The strategy $\ti{\la}_t = \la_i$ is a feasible strategy. It remains to be seen that it is a utility-maximizing one --- in fact, the only utility-maximizing strategy. Consider any alternative strategy $\ti{\la}_t$ such that 
	$$
	\sum_t P_{i,t}(\ti{\la}_t) \le T. 
	$$
	By \eqref{eq:511} we have 
	\begin{align*}
		\la_i \cdot \sum_t U_{i,t}(\ti{\la}_t) & \le 	\la_i \cdot \sum_t U_{i,t}(\la_i) +  \sum_t P_{i,t} (\ti{\la}_t) - \sum_t P_{i,t} (\la_i)\\ & \le 	\la_i \cdot \sum_t U_{i,t}(\la_i) + T-T = \\& =	\la_i \cdot \sum_t U_{i,t}(\la_i), 
	\end{align*}
	where the first inequality is strict unless $\ti{\la}_t=\la_i$ for all $t$. 
\end{proof}

We define 
\begin{equation}\label{eq:lai}
	\la_i:= \min\left(\text{player $i$'s best response from Claim~\ref{cl:reg25}},\bla\right)
\end{equation}

The remainder of the proof is conceptually straightforward, despite some calculations that need to be performed. Informally, any execution that has low strong regret must consist of each player $i$ repeatedly playing $\la_{i,t}$ that is close to its best-response value $\la_i$ --- the extent to which this fails to hold corresponds to the extent player $i$ experiences strong regret. Assuming this holds the outcome is close to a repeated execution of each player playing $\la_i$. The resulting prices, by Claim~\ref{cl:reg1} are close to VCG prices under preferences $(\la_i u_{ij})$, completing the picture. 

From now on, we fix 
$$
f_0(x_{ij}):= -\frac{\be}{x_{ij}},
$$
as in the statement of Theorem~\ref{thm:reg}, where $\be\ll 1$ is a parameter to be selected later.

\begin{claim}\label{cl:reg4}Fix a player $i$. 
	Let $\la_i$ be the best-response $\la$'s defined in \eqref{eq:lai}. 
	Consider a round $t$ in which all $\la'_k:=\la_{k,t}\le \bla$. Let $x'$ and $P'_i$ be the resulting allocation and  price charged from player $i$, and let  $x$ and $P_i$ be the resulting allocation and price when player $i$ reports $\la_i$ instead of $\la'_i$. 
	Then 
	\begin{equation}
		\label{eq:540}
		\left(\sum_j \la_i u_{ij} x_{ij} - P_i \right)-
		\left(\sum_j \la_i u_{ij} x'_{ij} - P'_i\right) 
		\ge (\la_i-\la'_i)^2 \cdot \Omega\left(\be^2 \cdot n^{-12}\cdot \bla^{-3}\right). 
	\end{equation}
\end{claim}

\begin{proof}
	Denote $$ \Phi (x):= F_0(x) + \sum_{ij} \la_i u_{ij} x_{ij}. $$
	Let $x$ be the maximizer of $\Phi(x)$. Note that $x$ is the resulting allocation on input $\la$. 
	
	The fact that $x$ maximizes $\Phi(x)$ implies that there exist $a_m$, $b_j$, $c_k$, and $d_\ell$ such that for all $m$ and $j$, 
	\begin{equation}\label{eq:531}
		\la_m u_{mj} + \frac{\be}{(x_{mj})^2} = a_m + b_j, 
	\end{equation}
	where without loss of generality (by adding a constant to all $a$'s and subtracting from all $b$'s) $\sum b_j =0$. This implies $a_m>0$ for all $m$. 
	
	Recall that $\la_k\le \bla$ for all $k$. There exists\footnote{By the 
		Birkhoff-von Neumann theorem, $x$ can be written as a convex combination of matchings. Moreover, by Caratheodory's theorem about convex hulls, $x$ can be written as a convex combination of at most $(n-1)^2+1 < n^2$ matchings. 
		The highest-weight matching $\pi$ will appear with weight $>1/n^2$.} a matching $\pi$ such that for all $k$, $x_{k\pi(k)}\ge 1/n^2$, and thus 
	$$
	|a_k + b_{\pi(k)} | \le \be\cdot n^4 + \bla. 
	$$ 
	Therefore, for all $j$, 
	\begin{equation}\label{eq:521}
		b_j <  \be\cdot n^4 +\bla. 
	\end{equation}
	and for some $m$, 
	\begin{equation}\label{eq:522}
		a_m\le  \be\cdot n^4 +\bla.
	\end{equation}
	
	We claim that for all $m,j$, \begin{equation}\label{eq:523}x_{mj}> \frac{1}{3}\cdot \left(\frac{\be}{n^4 \bla}\right)^{1/2}.\end{equation}
	If this is not the case, then $a_m+b_j>8 \bla n^4$, and thus by \eqref{eq:521} $a_m>6\bla n^4$. Therefore $b_{\pi(m)}<-4 \bla n^4$. For all $k$ we have
	$a_k+b_{\pi(m)} >0$, and thus $a_k>4 \bla n^4$, contradicting \eqref{eq:522}.

	Let $r$ and $s$ be such that $u_{ir}=1$ and $u_{is}=0$, and let $k$ be arbitrary. By adding equation \eqref{eq:531} with $(i,r)$ and $(k,s)$ and subtracting it with $(i,s)$ and $(k,r)$, we get 
	\begin{equation}\label{eq:542}
		\la_i - \la_k u_{kr} + \la_k u_{ks}=
		\la_i u_{ir}-\la_i u_{is} - \la_k u_{kr} + \la_k u_{ks} = 
		-\frac{\be}{x_{ir}^2} + \frac{\be}{x_{is}^2} +
		\frac{\be}{x_{kr}^2} - \frac{\be}{x_{ks}^2}. 
	\end{equation}
	Repeating this process for $x'$ instead of $x$ we get:
	\begin{equation}\label{eq:543}
		\la'_i - \la_k u_{kr} + \la_k u_{ks}=
		\la'_i u_{ir}-\la'_i u_{is} - \la_k u_{kr} + \la_k u_{ks} = 
		-\frac{\be}{{x'}_{ir}^2} + \frac{\be}{{x'}_{is}^2} +
		\frac{\be}{{x'}_{kr}^2} - \frac{\be}{{x'}_{ks}^2}, 
	\end{equation}
	and thus for some $m,j$,
	\begin{equation}
		\label{eq:544}
		\left| \frac{\be}{x_{mj}^2}  - \frac{\be}{{x'}_{mj}^2} \right|\ge
		\frac{|\la_i-\la'_i|}{4}. 
	\end{equation}
	Together with \eqref{eq:523} this implies 
	\begin{equation}
		\label{eq:545} |x_{mj}-x'_{mj}|> \Omega\left(
		\be^{1/2} \cdot n^{-6} \cdot \bla^{-3/2}
		\right) \cdot |\la_i - \la'_i|. 
	\end{equation}
	By Claim~\ref{cl:reg2}, with $\ga=\be$, this implies \eqref{eq:540}.
\end{proof}

\begin{claim}
	\label{cl:reg5}
	Let $\{\la_{i,t}\}$ be a feasible execution of the regularized algorithm with
	$\la_{i,t}\in[0,\bla]$, 
	such that player $i$ experiences total strong regret $<\ve T$. 
	Let $\la_i$ be the best-response strategy for player $i$ as described in \eqref{eq:lai}. Then 
	\begin{equation}
		\sum_{t} |\la_{i,t}-\la_i| < O\left( \ve^{1/2}\cdot T \cdot \bla^2 \cdot  n^6 \cdot \be^{-1}\right). 
	\end{equation}
\end{claim}

\begin{proof}
	Let $x^t$ be the result obtained in round $t$ when player $i$ deviates to $\la_i$, and let $\ti{x}^t$ be the original result of playing $\la_{i,t}$.
	Let $P^t_i$ and $\ti{P}^t_i$ be the corresponding prices.  
	
	Note that if $\la_i=\bla$, then $\sum_t P^t_i \ge \sum_t \ti{P}^t_i$ since $P^t_i$ is increasing in $\la_i$ and in this case $\la_i=\bla\ge\la_{i,t}$. On the other hand, if $\la_i<\bla$, then $\sum_t P^t_i = T \ge \sum_t \ti{P}^t_i$ by the budget constraint.

	Using Claim~\ref{cl:reg4} we get the following chain of inequalities involving the regret player $i$ experiences:
	\begin{align*}
		\ve\cdot T & > \sum_t \sum_j (u_{ij} x^t_{ij} -u_{ij} \ti{x}^t_{ij} ) \\
		& \ge\bla^{-1}\cdot  \sum_t  \left[	\left(\sum_j \la_i u_{ij} x^t_{ij} - P^t_i \right)-
		\left(\sum_j \la_i u_{ij} \ti{x}^t_{ij} - \ti{P}^t_i\right) \right] \\ &
		\ge\Omega(1)\cdot \bla^{-1}\cdot  \sum_{t} (\la_i-\la_{i,t})^2 \cdot \be^2 \cdot n^{-12} \cdot \bla^{-3}\\ &
		\ge\Omega(1)\cdot\bla^{-4}\cdot \be^2 \cdot n^{-12}\cdot  \left( \sum_{t} |\la_i-\la_{i,t}|\right)^2\cdot T^{-1}.
	\end{align*}
	Thus 
	$$
	\sum_{t} |\la_i-\la_{i,t}| < O\left( \ve^{1/2}\cdot T \cdot \bla^2 \cdot  n^6 \cdot \be^{-1}\right).
	$$
\end{proof}

\begin{thm}\label{thm:reg}
	In the unit-demand allocation setting without money, let $\{u_{ij}\}$ be utilities such that $u_{ij}\in [0,1]$, and for each $i$, $\min_j u_{ij} = 0$ and $\max_j u_{ij}=1$. 
	
	Consider an execution of the algorithm with regularizer $f_0(x_{ij})=-\bla^{-3}\cdot n^{-4}/x_{ij}$, that is with  $$\be=\bla^{-3}\cdot n^{-4}$$ and with $\la_{i,t}\in[0,\bla]$, where $\bla>10$. Suppose that each player has strong regret $<\ve \cdot T$ with\footnote{Importantly, the bound we need on $\ve$ does not depend on $T$.}
	$$
	\ve=o\left(\bla^{-12}\cdot n^{-22}\right).
	$$ Let $x$ be the resulting allocation, and
	$P_i$ be the price charged to player $i$ during the execution of the algorithm.
	
	Let $\la_i$ be the best response for each player $i$ as in \eqref{eq:lai}. 	
	Let $C_j$ be VCG prices corresponding to utilities $\la_i u_{ij}$. 
	
	Then for  $$\de=10\cdot \bla^{-1}<1$$ 
	the
	allocation $x$ is a $\de$-competitive equilibrium at budgets $1$ supported by prices $C_j'=(1-\de/3)\cdot C_j$.
\end{thm}

\begin{proof}
	Applying Claim~\ref{cl:reg5} to each player $i$ and taking the sum we get 
	\begin{equation}
		\label{eq:560}
		\sum_{t} \sum_i |\la_{i,t}-\la_i| =: s\cdot T=  O\left( \ve^{1/2}\cdot T \cdot \bla^2 \cdot  n^7 \cdot \be^{-1}\right)=
		o(T\cdot \bla^{-1}).
	\end{equation}
	We apply Claim~\ref{cl:reg1} to $\{\la_{i,t}\}$ for each $t$. Since $\la_{i,t}\le \bla$,  we can choose $M=\bla^{1/2}\cdot \be^{-1/2} \cdot n^{-1}$ to obtain the statement with 
	$$
	\eta \le 2 \cdot \bla^{1/2} \cdot n^2\cdot \be^{1/2}= 2\cdot \bla^{-1}. 
	$$
	Denote 
	$$ 
	s_t :=\sum_i |\la_{i,t}-\la_i|, 
	$$
	the contribution of round $t$ to $s$. Note that $s=\frac{1}{T}\sum_t s_t$. 
	
	Let $C_{j,t}$ be the VCG price of item $j$ with utilities $\la_{i,t}u_{ij}$.	
	By Claim~\ref{cl:app11} we have 
	\begin{equation}
		|C_j-C_{j,t}| \le 2\cdot s_t. 
	\end{equation}
	Applying \eqref{eq:471} we get 
	\begin{align*}
		\sum_t \sum_j C'_j x^t_{ij} & \le (1-\de/3)\cdot  \sum_t \left(2 s_t + \sum_j C_{j,t}\cdot  x^t_{ij} \right)\\ & \le (1-\de/3) \cdot \sum_t \left(2 s_t + P_{i,t}+\eta\right)\\ &= 
		(1-\de/3)\cdot ( 2 s\cdot T + \eta\cdot T + P_i) \\
		& \le T. 
	\end{align*}
	Consider an alternative allocation $y_{ij}$ such that 
	$$
	\sum_j C'_{j} y_{ij} \le 1,
	$$
	and thus 
	$$
	\sum_j C_{j} y_{ij} < 1+\de/2. 
	$$
	We need to bound $\sum_j u_{ij} y_{ij} - \sum_j u_{ij} x_{ij}$.

	Consider the execution of the algorithm where $\la_{i,t}$ is replaced with $\la_i$ with all other bids remaining $\la_{k,t}$. Let $\ti{C}_{j,t}$ be the resulting VCG prices, $\ti{x}^t_{kj}$ be the resulting allocation and $\ti{P}_{i,t}$ the resulting payment due from player $i$. 
	We know that 
	$$
	\sum_t \ti{P}_{i,t}\le T, 
	$$
	and that by the low strong regret condition, 
	$$
	T^{-1} \cdot \sum_t \sum_j u_{ij}\ti{x}^t_{ij} \le  \sum_j u_{ij} x_{ij} + \ve. 
	$$
	Moreover $\sum_t \ti{P}_{i,t}=T$ unless $\la_i=\bla$. 
	
	By condition \eqref{eq:472} we get
	\begin{align*}
		T\cdot  \sum_j y_{ij} \la_i u_{ij} &=\sum_t \sum_j y_{ij} \la_i u_{ij} \\ & \le \sum_t \sum_j \ti{x}^t_{ij} \la_i u_{ij} + \sum_t \sum_j y_{ij} \ti{C}_{j,t} - \sum_t \ti{P}_{i,t} + T\cdot  \eta\\
		& \le T \cdot \la_i  \cdot \sum_j u_{ij} x_{ij} + \ve\la_i T + T\cdot \sum y_{ij} C_j + 
		\sum_t \sum_j y_{ij} |\ti{C}_{j,t}-C_j| -   \sum_t \ti{P}_{i,t} + T\cdot \eta \\ & \le T \cdot \la_i \cdot \sum_j u_{ij} x_{ij} + \ve\la_iT+ T\cdot(1+\de/2) + 2 s\cdot T  - \sum_t \ti{P}_{i,t} + T\eta \\
		& = T \cdot \la_i \cdot \sum_j u_{ij} x_{ij}   - \sum_t \ti{P}_{i,t} + T\cdot (1+\de/2+\eta +  \ve\la_i +2 s)
	\end{align*}
	Recall that $\la_i\ge 1$. If $\la_i<\bla$, then we get 
	$$
	\sum_j u_{ij} y_{ij} - \sum_j u_{ij} x_{ij} \le \frac{1}{\la_i} \cdot (\de/2+\eta +  \ve\la_i +2 s) < \de. 
	$$
	If $\la_i= \bla$, we can just drop the  $ - \sum_t \ti{P}_{i,t}$ term and divide by $\la_i\cdot T=\bla\cdot T$ to get 
	$$
	\sum_j u_{ij} y_{ij} - \sum_j u_{ij} x_{ij} 
	\le \ve +  (1+\de/2+\eta  +2 s)\cdot \bla^{-1} 
	< \ve + 2\cdot \bla^{-1}<\de.
	$$
\end{proof}

\end{appendix}
\end{document}